\pdfoutput=1
\documentclass[%
 reprint,
nobibnotes,
 amsmath,amssymb, amsthm,
 aps,
prb,
onecolumn,
]{revtex4-2}

\usepackage{ytableau}
\usepackage{graphicx}% Include figure files
\usepackage{microtype}
\usepackage{dcolumn}% Align table columns on decimal point
\usepackage{bm}% bold math hypertext 
\usepackage[backref=none, breaklinks=True]{hyperref} %% This provides \href command
\usepackage[usenames,dvipsnames]{xcolor}
\usepackage{color}
\usepackage[framemethod=TikZ]{mdframed}
\usepackage{mathtools}
\usepackage{algorithm}      % keep for the algorithm counter (no float used)
\usepackage{algpseudocode}  % optional; not strictly required below
\usepackage{orcidlink}
\usepackage{amsthm}
\usepackage{mdframed}

\usepackage[normalem]{ulem}

\newenvironment{boxedtheorem}[1][]{%
  \begin{mdframed}[
    linewidth=0.6pt,
    roundcorner=2pt,
    nobreak=true,
    innerleftmargin=8pt,
    innerrightmargin=8pt,
    innertopmargin=0.6\baselineskip,
    innerbottommargin=0.6\baselineskip,
    skipabove=0.8\baselineskip,
    skipbelow=0.8\baselineskip
  ]%
  \begin{theorem}[#1]%
}{%
  \end{theorem}%
  \end{mdframed}%
}
\newmdenv[
  linewidth=0.4pt,
  roundcorner=6pt,
  linecolor=black!20,
  innerleftmargin=8pt,innerrightmargin=8pt,
  innertopmargin=6pt,innerbottommargin=6pt,
  skipabove=6pt,skipbelow=6pt,
  backgroundcolor=white
]{algobox}

\newcommand\myshade{85}
\colorlet{mylinkcolor}{violet}
\colorlet{mycitecolor}{YellowOrange}
\colorlet{myurlcolor}{Aquamarine}
\hypersetup{
  linkcolor  = mylinkcolor!\myshade!black,
  citecolor  = mycitecolor!\myshade!black,
  urlcolor   = myurlcolor!\myshade!black,
  colorlinks = true,
}
\newcommand{\ket}[1]{\left|#1\right\rangle}
\newcommand{\bra}[1]{\left\langle#1\right|}
\newcommand{\proj}[2]{\ket{#1}\!\bra{#2}}

\newtheorem{theorem}{Theorem}
\newtheorem{lemma}{Lemma}
\newtheorem{proposition}{Proposition}

\newtheorem{remark}{Remark}
\newtheorem{definition}{Definition}

\makeatletter
\renewcommand*\l@subsection[2]{}
\renewcommand*\l@subsubsection[2]{}
\renewcommand*\l@section[2]{%
  \ifnum\c@secnumdepth>\z@
    \@dottedtocline{1}{1.5em}{2.3em}{#1}{#2}%
  \fi
}
\makeatother
\begin{document}

\preprint{APS/123-QED}

\title{Efficient Quantum Circuits for Coherent Conversion Between General First- and Second-Quantized Many-Body Representations}% Force line breaks with \\

\author{Jack S. Baker$^1$ \orcidlink{0000-0001-6635-1397}}
\email{jack.baker@lge.com}

\author{Gaurav Saxena$^1$ \orcidlink{0000-0001-6551-1782}}

\author{Thi Ha Kyaw$^1$ \orcidlink{0000-0002-3557-2709}}%
\email{thiha.kyaw@lge.com}
\affiliation{$^1$LG Electronics Toronto AI Lab, Toronto, Ontario M5V 1M3, Canada} 

\date{\today}% It is always \today, today,

\begin{abstract}
Quantum simulation at fixed particle number admits two equivalent descriptions, a
first-quantized (particle) representation and a second-quantized
(occupation-number) representation. Their quantum resource costs differ sharply
across computational tasks, so the ability to convert coherently between them is
valuable. We construct an explicit unitary $Q$, with inverse $Q^\dagger$, that
maps a first-quantized state to its fixed-$N$ occupation-number form while
diagnosing the input's particle-exchange symmetry. The conversion is therefore
symmetry-agnostic at the input yet fully resolved at the output, and it applies
uniformly to bosonic, fermionic, and parastatistical sectors. At its foundation
lies a structural identification that we place at the center of this work: the
quantum Schur transform supplied by Schur--Weyl duality is the non-abelian Fourier
transform of the commuting pair $(S_N,U(d))$, and the occupation-number
representation is its weight basis, retaining only the labels shared by both
factors, the irrep $\lambda$ and the $\mathfrak{u}(d)$ weight. This reduction is
lossless for bosons and fermions, while a canonical Gelfand--Tsetlin promise
renders it one-to-one for the remaining sectors. Algorithmically, $Q$ composes the
strong Schur transform with reversible arithmetic that computes occupations as
successive row-sum differences of the Gelfand--Tsetlin pattern, yielding gate
complexity $\mathrm{poly}(N,d,\log(1/\epsilon))$. The converted state is prepared
efficiently in quantum memory. Any classical algorithm that outputs it explicitly,
however, pays a cost set by the sector dimension, which is polynomial of degree
$N$ in $d$ at fixed $N$ and exponential in $N$ when $d=\Theta(N)$. Finally, an
efficient classical sampler for the induced occupation-number distribution would
yield one for arbitrary quantum circuits, contrary to standard complexity
assumptions.
\end{abstract}

\maketitle
\setcounter{tocdepth}{1}
\tableofcontents

%  REWORKED INTRODUCTION
% =====================================================================
\section{Introduction}
\label{sec:introduction}

Quantum simulation has been a central motivation for quantum computation since Feynman's observation that controllable quantum devices can efficiently reproduce the dynamics of quantum systems that appear classically intractable to emulate \cite{Feynman1982}. Over the past several decades, digital quantum simulation has developed into a mature algorithmic toolbox spanning product-formula methods \cite{Trotter1959,Suzuki1990,ChildsSuTranWiebeZhu2021}, linear-combination-of-unitaries and truncated-series techniques \cite{ChildsWiebe2012,BerryChildsCleveKothariSomma2015}, and the qubitization and quantum-signal-processing framework for asymptotically optimal Hamiltonian simulation in standard oracle models \cite{LowChuang2017,LowChuang2019,GilyenSuLowWiebe2019}. In the near term, noisy devices have motivated variational and other hybrid simulation strategies that trade rigorous asymptotic guarantees for lower circuit depth and greater hardware accessibility \cite{Preskill2018,Peruzzo2014,McClean2016,Cerezo2021}, while the longer-term fault-tolerant setting has sharpened the case for quantum simulation as a leading application domain through increasingly detailed resource estimates for chemically and technologically relevant problems \cite{CampbellTerhalVuillot2017,Reiher2017, Baker2024, Casares2026, Motlagh2025, loaiza2025simulatingnearinfraredspectroscopyquantum, fomichev2024simulatingxrayabsorptionspectroscopy, fomichev2025fastsimulationsxrayabsorption, naranjo2026designingquantumtechnologiesquantum}. Across these regimes, the same theme persists: quantum many-body simulation remains one of the most natural and scientifically consequential candidates for large-scale quantum advantage.

Under the umbrella of quantum simulation, there are two canonical wavefunction formalisms: the particle (first-quantized) representation and the occupation-number (second-quantized) representation.  For \(N\) particles in a \(d\)-dimensional single-particle basis, first quantization embeds a many-body quantum state in \((\mathbb{C}^d)^{\otimes N}\), which is space-efficient when \(N\ll d\), requiring \(N\lceil \log_2 d\rceil\) qubits under standard encodings. The second quantized formalism represents the same physics in a fixed-\(N\) sector of Fock space using occupation vectors \((n_1,\dots,n_d)\), which exposes creation/annihilation algebra and particle-number flexibility but often incurs a larger space footprint (typically $\mathcal{O}(d)$) and distinct operator-locality trade-offs.  While these formalisms are mathematically equivalent at fixed \(N\), they can differ dramatically in practical resources once one specifies the task (state preparation, time evolution, spectroscopy, sampling, measurement strategy) and the Hamiltonian structure.  Consequently, quantum workflows should not be monolithic: combining primitives whose natural and efficient implementations may live in different representations is key. This creates a concrete need for a \emph{coherent} representation change that can move many-body quantum states between first and second quantization.

Recent work has begun to address this problem \cite{Ku2025HybridQuantization, Liu2024LowDepthSymmetrization}. Ref. \cite{Ku2025HybridQuantization} develops a hybrid quantization framework for electronic-structure simulation that switches between first- and second-quantized representations with conversion cost $\mathcal{O}(N\log N\log d)$ and qubit cost $\mathcal{O}(N\log d)$. Their emphasis is firmly fermionic and chemistry-driven. The conversion is designed to interoperate with plane-wave and molecular-orbital workflows, and to enable representation switching at points where measurement, electron non-conserving operations, or basis changes become more efficient in one description than the other \cite{Ku2025HybridQuantization}. In parallel, ref. \cite{Liu2024LowDepthSymmetrization} provides low-depth algorithms for quantum symmetrization with repeated entries and, as a consequence, a polylogarithmic-depth converter between second- and first-quantized representations. Their construction is especially significant for bosonic settings, where repeated occupations make symmetrization subtler than antisymmetrization, and they obtain an $\mathcal{O}(\log^3 N)$-depth conversion procedure with $\mathcal{O}(N\log N)$ ancillas \cite{Liu2024LowDepthSymmetrization}. Taken together, these works show that coherent conversion can be highly useful when the target statistics class and encoding structure are known in advance.

The present work addresses a different point in this design space. Rather than optimizing for a fixed statistics class or a particular chemistry-oriented encoding pair, we construct a representation-theoretic unitary $Q$ that is universal at the interface level. That is, it acts on any input state supported on a single Schur--Weyl sector of $(\mathbb{C}^d)^{\otimes N}$ (a single particle statistic) and routes it coherently to the corresponding fixed-$N$ occupation-number description, with $Q^\dagger$ performing the inverse map. This includes the familiar bosonic and fermionic sectors, but also the higher-dimensional sectors associated with parastatistics. This generality is not merely formal. That is, recent theoretical and experimental work has renewed interest in paraparticle structures, including the demonstration that nontrivial parastatistics inequivalent to ordinary bosons and fermions can arise in physically consistent models, as well as trapped-ion realizations of parabosonic and parafermionic oscillator dynamics \cite{Wang2025BeyondBosonsFermions,HuertaAlderete2025ExperimentalPara}. Our construction is therefore best viewed as a symmetry-agnostic adapter. The incoming particle statistics need not be hard-wired into a bespoke conversion routine beyond the promise that the incoming state spans a single Schur-Weyl symmetry sector.  Simply, the particle statistics of the incoming quantum state may be \textit{unknown}, converted into the correct second quantized form and the particle statistic identification detected and stored.

Underpinning this construction is a structural identification that we regard as a central conceptual contribution of the present work and develop in full in Sec.~\ref{subsec:2Q_as_fourier_basis}. At fixed $N$,  we show that the occupation-number representation is precisely the \emph{generalized group Fourier basis} of first quantization. Schur--Weyl duality endows $(\mathbb{C}^d)^{\otimes N}$ with two mutually commuting actions, those of the symmetric group $S_N$ and the unitary group $U(d)$. The strong quantum Schur transform is the associated non-abelian Fourier transform. It simultaneously reduces both representations into irreducibles, expressed in a subgroup-adapted basis labelled by a Young--Yamanouchi word $\sigma$ for $S_N$ and a Gelfand--Tsetlin pattern $\mu$ for $U(d)$. In this language, the first-quantized basis of particle configurations is the pre-Fourier (computational) basis, the Schur transform is the Fourier transform of the pair $(S_N,U(d))$, and the second-quantized occupation vector is the $\mathfrak{u}(d)$ weight. That weight is the joint eigenvalue of the commuting mode-number operators $\hat n_p=J_{pp}$, read off as the commutative reduction of the full Gelfand--Tsetlin label. Whether this reduction is faithful is distinguished by the statistics class. For bosons and fermions every weight space is one-dimensional, since the relevant Kostka numbers never exceed unity. The weight is therefore already a complete Fourier label, and the change of representation is a lossless re-indexing. In these sectors, second quantization is nothing more than the first-quantized state written in its Fourier basis. For the remaining parastatistical sectors, the weight is a genuine coarsening of the Gelfand--Tsetlin label. A single additional ingredient, the canonical Gelfand--Tsetlin promise, then restores a one-to-one Fourier labelling by designating one representative in each weight space.

The universality of $Q$ comes with a trade-off. The specialized conversion routines of Refs.~\cite{Ku2025HybridQuantization,Liu2024LowDepthSymmetrization} are more efficient in the settings they target. The former exploits fermionic sorted-list structure together with chemistry-specific workflow assumptions, while the latter leverages low-depth symmetrization machinery tailored to bosonic occupation data. Our circuit is correspondingly less specialized, and one should therefore not expect it to outperform these methods (asymptotically, at least) in their native regimes. The purpose of $Q$ is different. It is best regarded as a root, unspecialized algorithm that exposes the representation-theoretic structure underlying coherent quantization change in full generality. When the particle statistics are known \emph{a priori} and one requires only a fermionic chemistry-oriented converter or a bosonic symmetrization-based converter, known specialized algorithms are most likely the natural choice. By contrast, $Q$ is the more appropriate primitive when the incoming many-body state is not assumed to arrive with its particle statistics built into the conversion logic, including settings that may extend beyond Bose and Fermi statistics to more general parastatistical sectors. At the same time, this universality should not be viewed as the endpoint of optimization. On the contrary, one may expect statistics-aware specializations of $Q$ to improve substantially on the base universal construction, both asymptotically and in constant-factor resources. In this sense, the structure of $Q$ may be viewed as a blueprint from which more optimized routines can be derived, potentially even becoming competitive with the converters of Refs.~\cite{Ku2025HybridQuantization,Liu2024LowDepthSymmetrization} in the regimes they target. Overall, our contribution is therefore complementary to the existing literature: we trade some efficiency for universality, modularity, and immediate extensibility to symmetry sectors beyond those covered by currently specialized converters.
\begin{figure}
    \centering
    \includegraphics[width=1.0\linewidth]{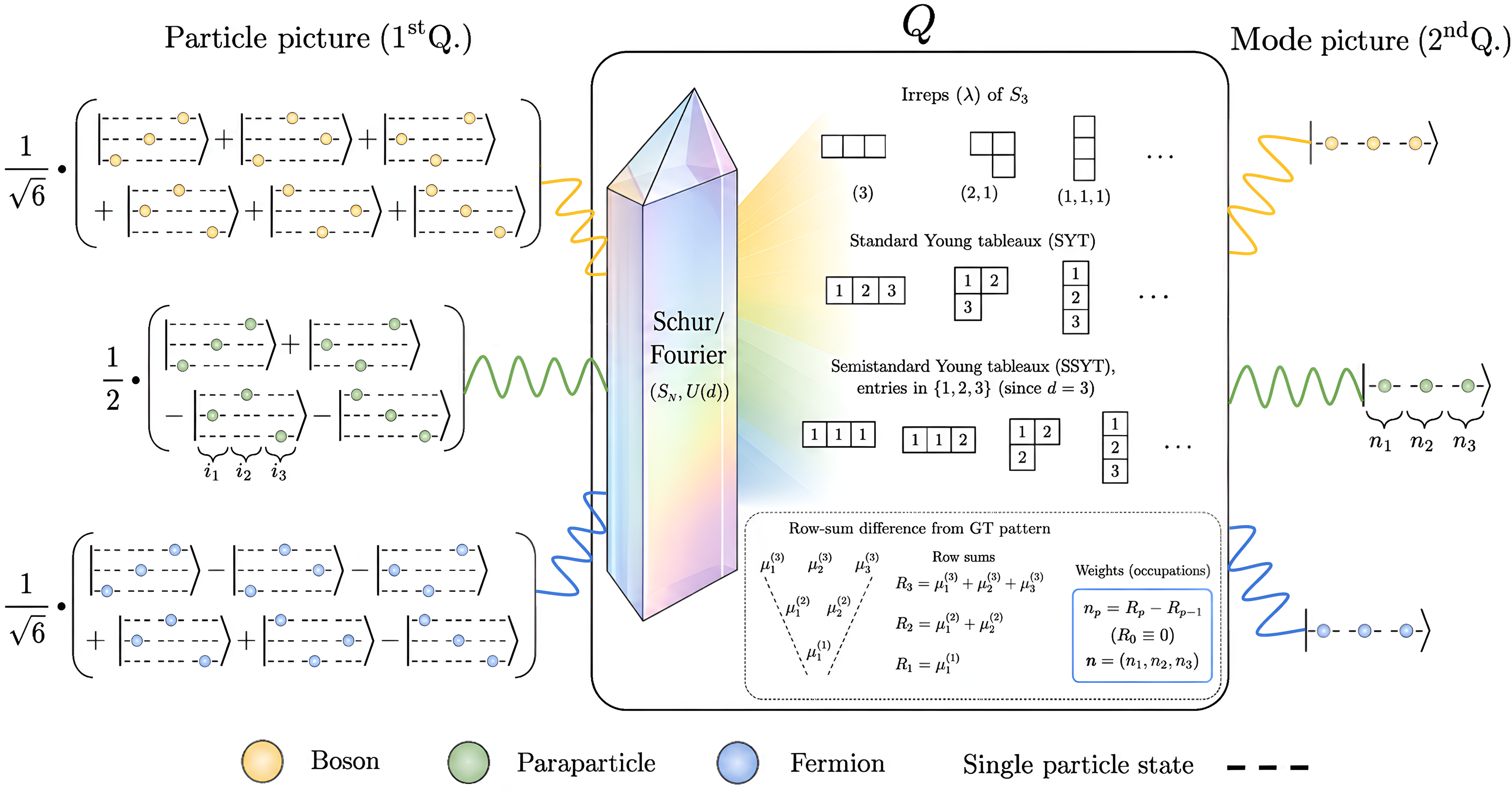}
    \caption{Pictorial summary of the transform $Q := U_{\mathrm{JS}}\,U_{\mathrm{Schur}}$
acting on the smallest instance that simultaneously realizes bosonic, parastatistical, and fermionic
sectors: $N=3$ particles in $d=3$ single-particle modes. \emph{Left (particle picture, $1^{\mathrm{st}}$Q.).} Three first-quantized inputs $|\psi_{1Q}\rangle \in (\mathbb{C}^3)^{\otimes 3}$ are shown, one per row, each an appropriately symmetrized superposition of the product configurations. Each ket is rendered as three single-particle registers (the particles $i_1 i_2 i_3$), with dashed lines marking the available single-particle states and a colored dot indicating the state each particle occupies; color encodes the exchange symmetry (yellow: boson, green: paraparticle, blue: fermion; see legend). For example, the top (bosonic, $\lambda=(3)$) row is the totally symmetric combination $|\psi_{1Q}^{\mathrm{bos}}\rangle = \tfrac{1}{\sqrt{6}}\big(|123\rangle + |132\rangle + |213\rangle + |231\rangle + |312\rangle + |321\rangle\big) \in \mathcal{H}_{\mathrm{sym}}$. The bottom (fermionic, $\lambda=(1^3)$) row is the totally antisymmetric Slater determinant and the middle (parastatistical, $\lambda=(2,1)$) row is a mixed-symmetry combination of four configurations. \emph{Centre (Unitary transform $Q$.)} Incoming quantum states are processed by the strong quantum Schur transform $U_{\mathrm{Schur}}$, the generalized group-Fourier transform of the commuting pair $(S_N,U(d))$, depicted by a prism. It resolves each input into a single irrep $\lambda$ of $S_3$, recording the $S_N$ multiplicity as a standard Young tableau (SYT, label $\sigma$) and the $U(d)$ weight data as a semistandard Young tableau (SSYT, label $\mu$). Next, the reversible arithmetic $U_{\mathrm{JS}}$ reads the occupations off the associated Gelfand--Tsetlin pattern (derived from the SSYT) as successive row-sum differences. \emph{Right (mode picture, $2^{\mathrm{nd}}$Q.).} The three corresponding second-quantized outputs $|n_1,n_2,n_3\rangle$, drawn as $d=3$ modes with the occupations marked by colored dots.}
    \label{fig:hero}
\end{figure}
The remainder of this work is structured as follows.  Section~\ref{sec:quantiztions} formalizes first- and second-quantized many-body representations and makes their resource trade-offs explicit for digital simulation tasks.  Section~\ref{sec:js_map} develops the representation-theoretic bridge between the two pictures via the Jordan--Schwinger map and Schur--Weyl duality, establishing the weight-basis identification that underpins conversion and proving, in Sec.~\ref{subsec:2Q_as_fourier_basis}, that the occupation-number representation is the group-Fourier basis of first quantization---lossless for bosons and fermions and, under the canonical Gelfand--Tsetlin promise, for all remaining sectors.  Section~\ref{sec:qqt} presents the coherent quantization transform \(Q\) and \(Q^\dagger\), including the strong Schur transform component and the arithmetic subroutine that computes occupations from Gelfand--Tsetlin patterns. Fig. \ref{fig:hero} condenses the key ideas and contributions of Secs. \ref{sec:quantiztions}-\ref{sec:qqt}. Section~\ref{sec:Q_worked_example_N2d3} gives a worked example with explicit intermediate wavefunctions and qubit encodings to illustrate the conversion end-to-end.  Section~\ref{sec:speed_up} analyzes quantum speed-ups and classical output barriers associated with explicit conversion and sampling tasks.  Section~\ref{sec:optimizing_algo} discusses open questions and extensions, and Section~\ref{sec:application} outlines application settings where coherent routing between representations is operationally useful.  We conclude in Section~\ref{sec:conclusions} with a summary of results and their implications for modular quantum simulation workflows.

\section{Digital quantum simulations of the quantum many-body problem \label{sec:quantiztions}}

Digital quantum simulations of many-body systems admit two canonical formulations: the particle (or \emph{first-quantized}) picture and the occupation (or \emph{second-quantized}) picture. For a fixed particle number $N$, these two pictures provide equivalent descriptions of many-body physics; what separates them in practice are their resource profiles on a digital quantum computer. Depending on $N$, the size $d$ of the single-particle basis, and the observables or tasks of interest, one or the other becomes more resource-efficient in space, time, or both. In this section, we formalize both pictures and make their resource trade-offs explicit, thereby establishing the mathematical structure that will later facilitate an efficient unitary transformation between them. This section then serves as a self-contained foundation for understanding the remainder of the work.

\subsection{Particle picture: first quantization \label{sec:first_q}}

\subsubsection{Structure of Hilbert space \label{subsec:1q_struc_of_hilbert_space}}

In the particle (first-quantized) picture, the Hilbert space for $N$ particles, each with a
$d$-dimensional single-particle space, is
\begin{equation}
  \mathcal H^{(N)}_{1\mathrm{Q}} := (\mathbb{C}^d)^{\otimes N}.
\end{equation}
A convenient orthonormal basis is given by product states
\begin{equation}
  |\mathbf{i}\rangle
  := |i_1\rangle \otimes \cdots \otimes |i_{N}\rangle
  \equiv |i_1 i_2 \cdots i_{N}\rangle,
  \qquad i_k \in \{1,\ldots,d\}.
\end{equation}
Accordingly, any pure state of $N$ distinguishable particles can be expanded as
\begin{equation}
  |\Psi\rangle
  = \sum_{\mathbf{i}\in\{1,\ldots,d\}^N} c_{\mathbf{i}}\,|\mathbf{i}\rangle.
\end{equation}
For example, the basis state $|122\rangle$ describes three particles for which the first occupies
the single-particle basis state labeled $1$, while the second and third occupy the basis state
labeled $2$.

In most physical applications, however, one is concerned with systems of identical particles.
In that setting, particle labels carry no physical meaning. Equivalently, physical observables in
the first-quantized description must be invariant under relabelings of the tensor factors. Let
$P(\pi)$ denote the unitary operator implementing a permutation $\pi \in S_N$ on
$(\mathbb{C}^d)^{\otimes N}$. Its action on a product basis state is
\begin{equation}
P(\pi)\,|i_1 i_2 \ldots i_{N}\rangle
\;=\;
|i_{\pi^{-1}(1)}\, i_{\pi^{-1}(2)} \ldots i_{\pi^{-1}(N)}\rangle,
\label{eq:perm_action}
\end{equation}
that is, $P(\pi)$ permutes the tensor factors according to $\pi$. For $N=3$, the symmetric group is
\begin{equation}
  S_3 = \{e,(1,2),(1,3),(2,3),(1,2,3),(1,3,2)\},
\end{equation}
consisting of the identity, three transpositions, and two three-cycles. For instance,
\begin{equation}
  P_{(1,2)}|122\rangle = |212\rangle,
  \qquad
  P_{(1,3,2)}|122\rangle = |221\rangle.
\end{equation}

For identical particles, a physical first-quantized observable $O_{1\mathrm{Q}}$ must commute with
every particle permutation,
\begin{equation}
  [O_{1\mathrm{Q}},P(\pi)]=0,
  \qquad \forall\,\pi\in S_N.
\end{equation}
As a consequence, a state and any of its permuted images are operationally indistinguishable with
respect to such observables. Indeed, for any $|\Psi\rangle \in (\mathbb{C}^d)^{\otimes N}$,
\begin{equation}
  \langle\Psi|P(\pi)^\dagger O_{1\mathrm{Q}} P(\pi)|\Psi\rangle
  =
  \langle\Psi|O_{1\mathrm{Q}}|\Psi\rangle,
  \qquad
  \forall\,\pi\in S_N,
\end{equation}
whenever $[O_{1\mathrm{Q}},P(\pi)]=0$. Thus, permutation-related states lie in the same
operational equivalence class under permutation-invariant observables.

From a representation-theoretic perspective, the permutation operators
$\{P(\pi)\}_{\pi\in S_N}$ furnish a unitary representation of $S_N$ on
$(\mathbb{C}^d)^{\otimes N}$, so the Hilbert space decomposes into permutation-symmetry
sectors labelled by the irreducible symmetry types of $S_N$. At the algebraic level, a choice
of particle statistics is a choice of which such sector is allowed. The full Schur--Weyl
decomposition of $(\mathbb{C}^d)^{\otimes N}$ into joint $U(d)\times S_N$ sectors, and the
precise labelling of these sectors by Young diagrams, is developed in
Sec.~\ref{subsec:schur_wey_decomp} [Eq.~\eqref{eq:SW_decomp}]; here we record only the
statistics sectors needed in the remainder of this section. Ordinary bosons and fermions
correspond to the two one-dimensional irreducible representations of $S_N$,
\begin{align}
  \text{bosons: }\quad &
  P(\pi)|\Psi\rangle = |\Psi\rangle
  \quad \forall\,\pi\in S_N,
  \quad\Rightarrow\quad
  |\Psi\rangle \in \mathcal H_{\mathrm{sym}}, \\
  \text{fermions: }\quad &
  P(\pi)|\Psi\rangle = \operatorname{sgn}(\pi)\,|\Psi\rangle
  \quad \forall\,\pi\in S_N,
  \quad\Rightarrow\quad
  |\Psi\rangle \in \mathcal H_{\mathrm{asym}},
\end{align}
namely the totally symmetric and totally antisymmetric subspaces $\mathcal H_{\mathrm{sym}}$ and
$\mathcal H_{\mathrm{asym}}$, with $\operatorname{sgn}(\pi)=\pm1$ for even/odd permutations. The
remaining higher-dimensional irreducible representations of $S_N$ correspond to
\emph{parastatistics}; these more general symmetry sectors, and their role in the present
construction, are treated in Sec.~\ref{subsec:schur_wey_decomp}.

\subsubsection{Structure of Hermitian and unitary operators}
\label{subsec:1q_struc_of_ops}

We now describe the structure of operators acting on the first-quantized
$N$-particle Hilbert space $(\mathbb{C}^d)^{\otimes N}$. Let $\mathrm{End}\!\left((\mathbb{C}^d)^{\otimes N}\right)$ denote the space of linear operators on this Hilbert space. A convenient way to organize such operators is by their support on subsets of tensor factors. For each subset $\{r_1,\ldots,r_k\}\subseteq \{1,\ldots,N\}$ with $1\le r_1<\cdots<r_k\le N$, one may consider operators that act non-trivially only on particles $r_1,\ldots,r_k$ and as the identity on all remaining particles. Expanding in this manner yields the general form
\begin{equation}
  O_{1\mathrm{Q}}
  \;=\;
  \sum_{k=0}^{N}
  \;\sum_{1 \le r_1 < \cdots < r_k \le N}
  O^{(k)}_{\,r_1 \ldots r_k},
  \label{eq:Hsumk_corrected}
\end{equation}
where each $O^{(k)}_{\,r_1\ldots r_k}$ is supported on the specified set of $k$ particles. Here the $k=0$ term represents a scalar multiple of the identity. This expansion is an operator decomposition by support. In physical applications one often truncates to small values of $k$, corresponding to one-body, two-body, or more generally few-body observables, whereas taking $k$ up to $N$ yields a completely general operator on $(\mathbb{C}^d)^{\otimes N}$.

To write $O^{(k)}_{\,r_1\ldots r_k}$ explicitly, fix the standard matrix units $E_{pq}:=|p\rangle\langle q|$ on $\mathbb{C}^d$, with $p,q\in\{1,\ldots,d\}$. For each particle label $r\in\{1,\ldots,N\}$, define the embedded one-particle operator
\begin{equation}
  E_{pq}^{(r)}
  :=
  I^{\otimes (r-1)} \otimes E_{pq} \otimes I^{\otimes (N-r)}.
\end{equation}
Then any operator supported on particles $r_1,\ldots,r_k$ can be expanded as
\begin{align}
  O^{(k)}_{\,r_1 \ldots r_k}
  &=
  \sum_{p_1,\ldots,p_k=1}^{d}
  \sum_{q_1,\ldots,q_k=1}^{d}
  o^{(k)}_{p_1\ldots p_k,\,q_1\ldots q_k}
  E_{p_1 q_1}^{(r_1)}
  E_{p_2 q_2}^{(r_2)}
  \cdots
  E_{p_k q_k}^{(r_k)}.
  \label{eq:k-body-term_corrected}
\end{align}
The coefficients  $o^{(k)}_{p_1\ldots p_k,\,q_1\ldots q_k}$ are simply the matrix elements of the corresponding $k$-particle operator in the chosen single-particle basis. If $\hat{o}^{(k)}\in \mathrm{End}\!\bigl((\mathbb{C}^d)^{\otimes k}\bigr)$ denotes the abstract operator acting on the selected $k$ particles, then
\begin{equation}
  o^{(k)}
  =
  \langle p_1\cdots p_k | \hat{o}^{(k)} | q_1\cdots q_k\rangle.
  \label{eq:k-body-matrix-elements_corrected}
\end{equation}
When the finite-dimensional single-particle space $\mathbb{C}^d$ arises from a truncation of an underlying continuum one-particle basis $\{\phi_i(\mathbf r)\}_{i=1}^{d}$, these coefficients may additionally be written in the familiar integral form
\begin{align}
  o^{(k)}
  &=
  \int \mathrm{d}\mathbf r_1 \cdots \mathrm{d}\mathbf r_k\,
  \phi_{p_1}^*(\mathbf r_1)\cdots \phi_{p_k}^*(\mathbf r_k)\,
  \hat{o}^{(k)}
  \nonumber\\
  &\qquad\qquad\times
  \phi_{q_1}(\mathbf r_1)\cdots \phi_{q_k}(\mathbf r_k),
  \label{eq:k-body-matrix-elements_integral}
\end{align}
provided $\hat{o}^{(k)}$ admits such a coordinate-space representation.

Taking the sum in Eq.~\eqref{eq:Hsumk_corrected} all the way to $k=N$ gives an exhaustive expansion: every Hermitian operator on $(\mathbb{C}^d)^{\otimes N}$ can be written in this form for a suitable choice of coefficients.  Moreover, in finite dimension every unitary operator $U$ on
$(\mathbb{C}^d)^{\otimes N}$ can be expressed as $U = e^{-iH}$ for some Hermitian operator $H$. Thus Hermitian generators provide a complete description of all unitary evolutions and basis changes on the first-quantized Hilbert space.

For systems of identical particles, one must further restrict attention to physical first-quantized observables $O_{1\mathrm{Q}}$ satisfying $[O_{1\mathrm{Q}},P(\pi)]=0, \forall\,\pi\in S_N$ which arise from the properties of the tensor with elements given by Eq. \eqref{eq:k-body-matrix-elements_integral}.

\subsubsection{Resource considerations in digital quantum simulation \label{subsec:1q_resource_tradeoff}}

A standard route to digital quantum simulation is to express a Hermitian operator as a \emph{linear combination of unitaries} (LCU).  In first quantization, we may expand each outer-product operator $|p\rangle\langle q|$ appearing in Eq. \eqref{eq:k-body-term_corrected} in a suitable operator basis.  For a $d$-dimensional single-particle Hilbert space one could, in principle, decompose into the $d^2$ generalized Gell-Mann matrices (qudit operators) and the identity.  On qubit hardware, however, we must choose an encoding of the $d$ levels into $n$ qubits with $2^n \ge d$ (typically a binary or Gray-code encoding \cite{Sawaya2020}).  In the binary encoding, for example, one maps each label $p\in\{1,\ldots,d\}$ to the $n$-qubit computational basis state $|p_n p_{n-1}\cdots p_1\rangle$ whose bit string is the binary representation of $p-1$, i.e.\ $p-1 = \sum_{j=1}^n 2^{j-1} p_j$ with $p_j \in \{0,1\}$.  Then any outer product factorizes over qubits as
\begin{equation}
  |p\rangle\langle q|
  \;=\;
  \bigotimes_{j=1}^{n} |p_j\rangle\langle q_j|,
  \label{eq:outerprod-factorization}
\end{equation}
and each single-qubit projector is expanded in the Pauli basis via
\begin{align}
  |0\rangle\langle 0| = \frac{I + Z}{2},\quad
  |1\rangle\langle 1| = \frac{I - Z}{2},\quad \nonumber\\
  |0\rangle\langle 1| = \frac{X + i Y}{2},\quad
  |1\rangle\langle 0| = \frac{X - i Y}{2}.
  \label{eq:singlequbit-pauli}
\end{align}
Combining Eqs. \eqref{eq:outerprod-factorization} and \eqref{eq:singlequbit-pauli}, any $|p\rangle\langle q|$ becomes a sum of tensor products of single-qubit Pauli operators.  This yields the \emph{naive Pauli LCU}: each term in $O_{1\textrm{Q}}$ is expanded into a (potentially large) weighted sum of Pauli strings. Alternatively, an efficient decomposition scheme into the Pauli LCU is derived in Ref. \cite{Georges2025LCU} utilizing the Walsh-Hadamard transform.

From the perspective of qubit count, first quantization is economical when $N \ll d$. Each particle requires $n = \lceil \log_2 d \rceil$ qubits to represent its state, so an $N$-particle system occupies
\begin{equation}
  N_{\mathrm{qubits}}^{(1\mathrm{Q})}
  \;=\;
  N \,\big\lceil \log_2 d \big\rceil
  \label{eq:nqubits-1q}
\end{equation}
qubits.  In contrast, standard second-quantized encodings require $O(d)$ qubits \cite{JordanWigner1928, Sawaya2020}.  Thus for fixed $N$ and growing $d$, first quantization offers an exponential saving in qubit number
\cite{Babbush2019,Su2021}. 

Historically, some of the first proposals for simulating interacting fermions on a quantum computer were formulated in first quantization \cite{Abrams1997,Ortiz2001}.  In quantum chemistry, early first-quantized algorithms showed that chemical dynamics and electronic structure problems, under reasonable assumptions, could be solved in polynomial time on a quantum computer \cite{Kassal2008,Babbush2018}, but with resource costs that were not yet competitive with the best second-quantized approaches.  Subsequent work introduced a variety of techniques to reduce Hamiltonian sparsity and to exploit symmetries, including spin-free and symmetry-adapted representations \cite{Whitfield2013}.  These methods
demonstrate that careful choice of single-particle basis, encoding, and
operator factorization can substantially shrink the number of nontrivial terms
appearing in the naive Pauli-LCU construction as well as reducing the LCU 1-norm, which too has direct implications to computational cost.

More recent work has established first quantization as a leading candidate for
asymptotically optimal digital quantum simulations of electronic structure.
It has been shown that plane-wave first-quantized algorithms
combined with modern block-encoding and qubitization techniques yield gate
complexities with sublinear dependence on the basis size $d$
\cite{Babbush2019,Babbush2018}.  Ref. \cite{Su2021} provided detailed fault-tolerant resource estimates for such algorithms, demonstrating that for large plane-wave bases the spacetime volume can be far smaller than that of state-of-the-art second-quantized schemes. These developments
were extended to realistic materials \cite{Berry2024} , where
non-local pseudopotentials were used in a first-quantized plane-wave framework without substantially increasing the asymptotic cost \cite{ShokrianZini2023, Berry2024}.  Most recently, first quantized simulations of chemistry have been studied in a basis-set agnostic way \cite{Georges2025}, providing a qubitized first-quantized algorithm that achieves favorable scaling in both basis size and target precision \cite{Georges2025} by exploiting a dual plane-wave expansion to diagonalize the Coulomb potential.  Together, these works show that while the naive Pauli-LCU construction is prohibitively expensive, modern first-quantized algorithms exploit Hamiltonian structure, encoding choice, basis choice, and advanced block-encoding techniques to achieve substantial savings in both gate complexity and ancilla overhead, all while retaining $N \log_2 d$ qubit scaling.

\subsection{Occupation picture: second quantization \label{sec:second_q}}

\subsubsection{Ordinary bosonic and fermionic ladder operators in Fock space}
\label{subsec:regular_boson_fermion}

In the occupation picture, operators are built from mode creation and annihilation operators, also called ladder operators. For each single-particle mode
$p\in\{1,\dots,d\}$, one introduces operators that raise or lower the occupation of that mode subject to the particle statistics under consideration. The vacuum state $|\mathrm{vac}\rangle$ is defined by
\begin{equation}
  a_p |\mathrm{vac}\rangle = 0,
  \qquad \forall\,p\in\{1,\dots,d\}.
\end{equation}
Acting repeatedly with creation operators on the vacuum generates the corresponding Fock space.

For ordinary bosons, the mode operators obey the canonical commutation relations
\begin{align}
  [b_p,b_q^\dagger] &= \delta_{pq}, &
  [b_p,b_q] &= 0, &
  [b_p^\dagger,b_q^\dagger] &= 0,
  \label{eq:CCR}
\end{align}
and their action on the bosonic number basis is
\begin{align}
  b_p^\dagger |n_1,\dots,n_p,\dots,n_{d}\rangle
  &=
  \sqrt{n_p+1}\,
  |n_1,\dots,n_p+1,\dots,n_{d}\rangle,
  \\
  b_p |n_1,\dots,n_p,\dots,n_{d}\rangle
  &=
  \sqrt{n_p}\,
  |n_1,\dots,n_p-1,\dots,n_{d}\rangle.
\end{align}
For ordinary fermions, the mode operators obey the canonical anticommutation relations
\begin{align}
  \{f_p,f_q^\dagger\} &= \delta_{pq}, &
  \{f_p,f_q\} &= 0, &
  \{f_p^\dagger,f_q^\dagger\} &= 0,
  \label{eq:CAR}
\end{align}
with $n_p\in\{0,1\}$ in the occupation basis. Once an ordering of the modes has been fixed, their action on the fermionic occupation basis is
\begin{align}
  f_p^\dagger |n_1,\dots,n_{d}\rangle
  &=
  (-1)^{\sum_{q=1}^{p-1} n_q}\,(1-n_p)\,
  |n_1,\dots,n_p+1,\dots,n_{d}\rangle,
  \\
  f_p |n_1,\dots,n_{d}\rangle
  &=
  (-1)^{\sum_{q=1}^{p-1} n_q}\,n_p\,
  |n_1,\dots,n_p-1,\dots,n_{d}\rangle.
\end{align}
Thus, in the fermionic case, the occupation update is accompanied by the familiar sign determined by the ordered occupation string.

For both bosons and fermions, the local mode-number operators are
\begin{equation}
  \hat n_p := a_p^\dagger a_p,
\end{equation}
where $a_p=b_p$ in the bosonic case and $a_p=f_p$ in the fermionic case. These act diagonally on the corresponding number basis,
\begin{equation}
  \hat n_p |n\rangle = n_p |n\rangle,
\end{equation}
and the total number operator is
\begin{equation}
  \hat N := \sum_{p=1}^{d} \hat n_p,
\end{equation}
whose eigenvalue on a number state $|n\rangle$ is the total occupation number
\begin{equation}
  N=\sum_{p=1}^{d} n_p .
  \label{eq:action_of_total_number_op}
\end{equation}

The essential distinction from the first-quantized picture is that the exchange symmetry is built directly into the operator algebra. In first quantization one restricts wavefunctions to the appropriate permutation-symmetry sector. In second quantization the same information is encoded in the commutation relations themselves: the CCR generate the bosonic Fock space, while the CAR generate the fermionic Fock space. 

\subsubsection{Paraparticles and color-resolved realizations}
\label{subsubsec:paraparticles_and_color_space}

Parastatistics enlarge the set of allowed permutation-symmetry types beyond the totally symmetric and totally antisymmetric sectors. A convenient realization is Green's ansatz, in which paraparticles are represented using an additional internal variable, often called a color $\mathcal{C}$.  For each mode $p\in\{1,\dots,d\}$ and each color $\alpha\in\{1,\dots,\mathcal{C}\}$, we introduce color-resolved operators
\begin{equation}
  b_p^{(\alpha)},\quad b_p^{(\alpha)\dagger},
  \qquad
  f_p^{(\alpha)},\quad f_p^{(\alpha)\dagger}.
\end{equation}
The full paraboson and parafermion operators are then obtained by summing over color:
\begin{equation}
  B_p := \sum_{\alpha=1}^{\mathcal{C}} b_p^{(\alpha)},
  \qquad
  B_p^\dagger := \sum_{\alpha=1}^{\mathcal{C}} b_p^{(\alpha)\dagger},
\end{equation}
and
\begin{equation}
  F_p := \sum_{\alpha=1}^{\mathcal{C}} f_p^{(\alpha)},
  \qquad
  F_p^\dagger := \sum_{\alpha=1}^{\mathcal{C}} f_p^{(\alpha)\dagger}.
\end{equation}

In Green's ansatz, same-color components obey the ordinary algebra, while different-color components obey the opposite exchange type. One convenient form of these relations is:
\begin{align}
  \text{parabosons:}\qquad
  [b_p^{(\alpha)},b_q^{(\alpha)\dagger}] &= \delta_{pq},
  &
  [b_p^{(\alpha)},b_q^{(\alpha)}] &= 0,
  &
  [b_p^{(\alpha)\dagger},b_q^{(\alpha)\dagger}] &= 0,
  \nonumber\\
  \{b_p^{(\alpha)},b_q^{(\beta)\dagger}\} &= 0,
  &
  \{b_p^{(\alpha)},b_q^{(\beta)}\} &= 0,
  &
  \{b_p^{(\alpha)\dagger},b_q^{(\beta)\dagger}\} &= 0,
  \qquad \alpha\neq\beta,
  \label{eq:green_ansatz_pb}
\end{align}
and
\begin{align}
  \text{parafermions:}\qquad
  \{f_p^{(\alpha)},f_q^{(\alpha)\dagger}\} &= \delta_{pq},
  &
  \{f_p^{(\alpha)},f_q^{(\alpha)}\} &= 0,
  &
  \{f_p^{(\alpha)\dagger},f_q^{(\alpha)\dagger}\} &= 0,
  \nonumber\\
  [f_p^{(\alpha)},f_q^{(\beta)\dagger}] &= 0,
  &
  [f_p^{(\alpha)},f_q^{(\beta)}] &= 0,
  &
  [f_p^{(\alpha)\dagger},f_q^{(\beta)\dagger}] &= 0,
  \qquad \alpha\neq\beta.
  \label{eq:green_ansatz_pf}
\end{align}

It is useful to distinguish the color-resolved auxiliary space from the physical Fock space. In the color-resolved picture, one may write basis states of the schematic form
\begin{equation}
  |n\rangle_G
  :=
  \bigotimes_{\alpha=1}^{\mathcal{C}}
  |n_1^{(\alpha)},\dots,n_{d}^{(\alpha)}\rangle,
\end{equation}
where $n_p^{(\alpha)}$ records the occupation of mode $p$ with color $\alpha$. The total occupation of mode $p$ is then
\begin{equation}
  n_p = \sum_{\alpha=1}^{\mathcal{C}} n_p^{(\alpha)}.
\end{equation}
Physical observables are constructed from the summed operators $B_p,B_p^\dagger$ or
$F_p,F_p^\dagger$ and their polynomials, and therefore do not resolve the auxiliary color labels individually. Consequently, the physical Fock space is obtained from the color-resolved construction by retaining only the physically relevant states and multiplicities associated with the chosen parastatistics.

\subsubsection{Total particle-number-conserving operators}
\label{subsec:particle_conserving_ops_2nd_q}

We now identify the distinguished one-body generators that preserve total particle number. These operators play a central role in second quantization, since they generate the natural action of the single-particle algebra on the many-body Fock space and provide the basic building blocks from which general number-conserving observables are constructed.

For ordinary bosons and ordinary fermions, the relevant bilinears are the familiar operators
\begin{equation}
  J_{pq} := a_p^\dagger a_q,
  \qquad p,q\in\{1,\dots,d\},
  \label{eq:Jpq-standard}
\end{equation}
where \(a_p=b_p\) in the bosonic case and \(a_p=f_p\) in the fermionic case. These commute with the total number operator,
\begin{equation}
  [\hat N,J_{pq}] = 0,
  \label{eq:Jpq-number-conserving}
\end{equation}
and their diagonal members are precisely the ordinary mode-number operators,
\begin{equation}
  J_{pp} = a_p^\dagger a_p = \hat n_p.
\end{equation}

For parabosons and parafermions, it is convenient to define the physical one-body generators directly in terms of the summed Green operators \(B_p,B_p^\dagger\) and \(F_p,F_p^\dagger\). We take
\begin{align}
  \text{parabosons:}\qquad
  J_{pq}
  &:=
  \frac12 \{B_p^\dagger,B_q\}
  - \frac{\mathcal C}{2}\,\delta_{pq},
  \label{eq:Jpq-paraB}
  \\
  \text{parafermions:}\qquad
  J_{pq}
  &:=
  \frac12 [F_p^\dagger,F_q]
  + \frac{\mathcal C}{2}\,\delta_{pq},
  \label{eq:Jpq-paraF}
\end{align}
where \(\mathcal C\) is the Green color index. With this normalization, the diagonal generators coincide directly with the physical mode-occupation operators in the parabose and parafermi cases.

Indeed, for parabosons one finds
\begin{align}
  \{B_p^\dagger,B_p\}
  &=
  \sum_{\alpha,\beta=1}^{\mathcal C}
  \{b_p^{(\alpha)\dagger},b_p^{(\beta)}\}
  =
  \sum_{\alpha=1}^{\mathcal C}
  \{b_p^{(\alpha)\dagger},b_p^{(\alpha)}\}
  \nonumber\\
  &=
  \sum_{\alpha=1}^{\mathcal C}
  \bigl(2\,b_p^{(\alpha)\dagger}b_p^{(\alpha)}+1\bigr)
  =
  2\sum_{\alpha=1}^{\mathcal C} b_p^{(\alpha)\dagger}b_p^{(\alpha)}
  + \mathcal C,
\end{align}
and therefore
\begin{equation}
  J_{pp}
  =
  \frac12\{B_p^\dagger,B_p\}
  - \frac{\mathcal C}{2}
  =
  \sum_{\alpha=1}^{\mathcal C} b_p^{(\alpha)\dagger}b_p^{(\alpha)}
  =: \hat n_p.
\end{equation}
Similarly, for parafermions,
\begin{align}
  [F_p^\dagger,F_p]
  &=
  \sum_{\alpha,\beta=1}^{\mathcal C}
  [f_p^{(\alpha)\dagger},f_p^{(\beta)}]
  =
  \sum_{\alpha=1}^{\mathcal C}
  [f_p^{(\alpha)\dagger},f_p^{(\alpha)}]
  \nonumber\\
  &=
  \sum_{\alpha=1}^{\mathcal C}
  \bigl(2\,f_p^{(\alpha)\dagger}f_p^{(\alpha)}-1\bigr)
  =
  2\sum_{\alpha=1}^{\mathcal C} f_p^{(\alpha)\dagger}f_p^{(\alpha)}
  - \mathcal C,
\end{align}
so that
\begin{equation}
  J_{pp}
  =
  \frac12[F_p^\dagger,F_p]
  + \frac{\mathcal C}{2}
  =
  \sum_{\alpha=1}^{\mathcal C} f_p^{(\alpha)\dagger}f_p^{(\alpha)}
  =: \hat n_p.
\end{equation}
Thus, in all statistics sectors under consideration—bosonic, fermionic, parabosonic, and parafermionic—the diagonal generators \(J_{pp}\) are the physical occupation counters for mode \(p\).

The total number operator may therefore be written uniformly as
\begin{equation}
  \hat N = \sum_{p=1}^{d} J_{pp},
  \label{eq:number-op-from-Jdiag}
\end{equation}
and the one-body generators satisfy
\begin{equation}
  [\hat N,J_{pq}] = 0
\end{equation}
for all \(p,q\). In the special case \(\mathcal C=1\), the parabose and parafermi definitions reduce to the ordinary bosonic and fermionic bilinears, respectively.

Once these one-body generators have been identified, general particle-number-conserving second-quantized observables may be organized as finite polynomials in the \(J_{pq}\). We say that an operator \(O^{\hat N}_{2\mathrm Q}\) is particle-number-conserving if
\begin{equation}
  [O^{\hat N}_{2\mathrm Q},\hat N]=0.
\end{equation}
For fixed statistics and on a fixed particle-number sector, such observables are generated by polynomials in the corresponding one-body generators. Accordingly, one may write
\begin{equation}
  O^{\hat N}_{2\mathrm Q}
  =
  \sum_{r=0}^{k}
  \sum_{\substack{
    p_1,\dots,p_r=1\\
    q_1,\dots,q_r=1
  }}^{d}
  c^{(r)}_{p_1\dots p_r,\,q_1\dots q_r}\,
  J_{p_1q_1}\cdots J_{p_rq_r},
  \label{eq:secondQ_polyJ}
\end{equation}
where \(J_{pq}\) denotes the statistics-appropriate bilinear family defined above, and the \(r=0\) term is a scalar multiple of the identity. The coefficients \(c^{(r)}_{p_1\dots p_r,\,q_1\dots q_r}\) are the second-quantized coefficients associated with the corresponding first-quantized expansion; cf. Eq.~\eqref{eq:k-body-matrix-elements_corrected}.

For ordinary bosons and fermions, one may equivalently use the familiar normal-ordered form
\begin{equation}
  \widetilde O^{\hat N}_{2\mathrm Q}
  =
  \sum_{r=0}^{k}
  \sum_{\substack{
    p_1,\dots,p_r=1\\
    q_1,\dots,q_r=1
  }}^{d}
  \widetilde c^{(r)}_{p_1\dots p_r,\,q_1\dots q_r}\,
  a^\dagger_{p_1}\cdots a^\dagger_{p_r}\,
  a_{q_r}\cdots a_{q_1},
  \label{eq:secondQ_sumK}
\end{equation}
with \(a=b\) for bosons and \(a=f\) for fermions. The coefficients in Eq.~\eqref{eq:secondQ_sumK} need not coincide term-by-term with those in Eq.~\eqref{eq:secondQ_polyJ}; the two expressions are simply different but equivalent ways of organizing the same class of number-conserving operators.

Finally, for fixed statistics and fixed particle number \(N\), the Hermitian particle-number-conserving operators form a real vector space, and multiplication by \(i\) yields the corresponding Lie algebra of number-conserving generators on the fixed-\(N\) sector of Fock space. The associated unitary group is
\begin{equation}
  U\!\bigl(\mathcal F_N^{(\mathrm{stat})}\bigr)
  =
  \bigl\{
    e^{-iH}
    :
    H=H^\dagger
    \text{ and }
    [H,\hat N]=0
    \text{ on }
    \mathcal F_N^{(\mathrm{stat})}
  \bigr\}.
\end{equation}
These are precisely the number-conserving unitaries within the chosen symmetry class, whether bosonic, fermionic, parabosonic, or parafermionic.

\subsubsection{Resource considerations in digital quantum simulation}
\label{subsec:2q_resource_tradeoff}

The dominant resource drivers in second quantized quantum simulation are (i) the qubit cost per mode and (ii) the Pauli weight (the number of qubits individual Pauli-words act on non-trivially) and term count induced when rewriting creation/annihilation polynomials into a qubit-native operator basis.  These choices interact strongly with the target simulation primitive (product formulas, block encodings, quantum signal processing etc.), but many trade-offs recur across methods.

\subsubsection{Fermions \label{subsubsec:2q_fermions}}

For fermions, the standard approach maps $d$ spin-orbitals to $d$ qubits and represents the fermionic algebra via a fermion-to-qubit transform. The Jordan--Wigner (JW) transform \cite{JordanWigner1928} is conceptually simplest: it encodes the fermionic parity string explicitly as a product of $Z$ operators on all lower-index modes.  As a result, a single fermionic operator $f_p$ or $f_p^\dagger$ maps to a Pauli string whose weight scales as $O(d)$ in the worst case, and exponentiating generic two-body terms (the Coulomb interaction, for example) can incur long parity strings and correspondingly large two-qubit gate counts. The Bravyi--Kitaev (BK) transform \cite{BravyiKitaev2002,Seeley2012} reorganizes parity information into a logarithmic-depth data structure, reducing the Pauli weight of individual fermionic operators to $O(\log d)$.  In practice, BK often reduces CNOT counts for product formula simulations relative to JW for chemistry-relevant Hamiltonians, though the advantage is instance-dependent; a broad empirical comparison across many molecular instances can be found in Ref.~\cite{Tranter2018}.

Beyond JW and BK, a large design space trades qubits, Pauli weight, and circuit structure.  Code-based mappings can reduce the number of qubits below $d$ when one is willing to pay extra overhead in multi-controlled logic and nonlocal updates, providing systematic qubit--gate trade-offs \cite{SteudtnerWehner2018}.  On the other end of the spectrum, graph-structured ``superfast'' encodings and their generalizations can reduce Pauli weights for Hamiltonians with sparse interaction graphs, at the cost of increased qubit counts that scale with graph degree \cite{Setia2019}.  These alternatives are most valuable when other constraints such as physical qubit connectivity, native entangling gates, available ancillae and implementation of fault tolerant computation shift the true cost away from a simple Pauli-weight proxy.

\subsubsection{Bosons \label{subsubsec:2q_bosons}}

For bosons on qubit hardware, one first truncates each local mode to a finite occupation cutoff $n_{\max}$ (i.e., truncation to a $d$-level local Hilbert space, where $d = n_{\text{max}} +1$). Three popular encodings then interpolate between qubit count and operator locality: unary (``one-hot''), binary, and Gray-code encodings \cite{Sawaya2020,McArdle2019,Macridin2018}. Unary encodings use $d$ qubits per mode but make ladder operators sparse and local, which can substantially reduce gate depth when qubits are plentiful.

Binary encodings compress the same truncated space into $\lceil\log_2(d)\rceil$ qubits per mode, but ladder operators become sums of Pauli strings obtained by expanding outer products in the computational basis.  Concretely, the same factorization and Pauli expansion used earlier in first quantization, Eqs.~\eqref{eq:outerprod-factorization} and \eqref{eq:singlequbit-pauli}, applies verbatim to the matrix units $|n\rangle\!\langle n^{\prime}|$ that appear when representing truncated bosonic creation/annihilation operators as transition operators between number states.  Gray-code encodings retain the logarithmic qubit count of binary while ensuring that adjacent occupation numbers differ by Hamming distance one, which can reduce Pauli weights and/or the number of nontrivial Pauli terms appearing in ladder-operator decompositions, improving depth constants in digital simulation \cite{Sawaya2020}.

Beyond the scope of this work, it is worth noting that bosonic degrees of freedom are often more natural in continuous-variable (CV) architectures, where modes are hardware-native rather than truncated abstractions; see, e.g., Refs.~\cite{LloydBraunstein1999,Weedbrook2012} for foundational and review-level treatments of CV quantum information.

\subsubsection{Paraparticles \label{subsubsec:2q_paraparticles}}

While many-body parabose/parafermi field theories are less commonly targeted in quantum-simulation benchmarks, Ref. \cite{HuertaAlderete2025} reports a trapped-ion digital simulation of parabosonic and parafermionic oscillators by mapping truncated paraparticle states to qubit registers and compiling the resulting dynamics to qubit gates. This provides evidence that parastatistics-inspired operator algebras can be simulated on gate-based hardware, but this still remains a relatively unexplored area.

\section{Representation-theoretic bridge between first and second quantization \label{sec:js_map}}

\subsection{Jordan--Schwinger map}
\label{subsec:js_map}

The Jordan--Schwinger (JS) map provides an algebraic bridge between the single-particle operator algebra on \(\mathbb C^d\) and the algebra of particle-number-conserving operators on Fock space. Since the relevant generators commute with the total number operator, this representation restricts naturally to each fixed-\(N\) sector. In the present setting, the construction applies uniformly across the bosonic, fermionic, parabosonic, and parafermionic symmetry classes through the statistics-dependent but uniformly normalized bilinears \(J_{pq}\) introduced in Sec.~\ref{subsec:particle_conserving_ops_2nd_q}.

Recall the matrix units \(E_{pq}=|p\rangle\langle q|\) (\(p,q\in\{1,\dots,d\}\)) of
Sec.~\ref{subsec:1q_struc_of_ops}, which span \(\mathfrak{gl}(d,\mathbb C)\) and obey the
commutation relations
\begin{equation}
  [E_{pq},E_{rs}]
  =
  \delta_{qr}E_{ps}-\delta_{ps}E_{rq}.
  \label{eq:gl_comm_matrix_units}
\end{equation}
On the physical Fock space for the chosen statistics sector, the one-body generators \(J_{pq}\) satisfy the same structure constants,
\begin{equation}
  [J_{pq},J_{rs}]
  =
  \delta_{qr}J_{ps}-\delta_{ps}J_{rq}.
  \label{eq:gl_comm_Jpq}
\end{equation}
To verify Eq.~\eqref{eq:gl_comm_Jpq}, recall that for ordinary bosons and fermions
$J_{pq}=a_p^\dagger a_q$ [Eq.~\eqref{eq:Jpq-standard}], so the relation follows from the
standard bilinear computation upon inserting the canonical relations Eqs. \eqref{eq:CCR} and
\eqref{eq:CAR}. For parabosons and parafermions, substituting Green's ansatz
[Eqs.~\eqref{eq:green_ansatz_pb} and \eqref{eq:green_ansatz_pf}] into the definitions
[Eqs.~\eqref{eq:Jpq-paraB} and \eqref{eq:Jpq-paraF}] cancels the constant shift and reduces
the generators to the color-summed diagonal form
$J_{pq}=\sum_{\alpha=1}^{\mathcal C} a_p^{(\alpha)\dagger}a_q^{(\alpha)}$ (with
$a^{(\alpha)}=b^{(\alpha)}$ or $f^{(\alpha)}$ as appropriate), the off-diagonal extension of
the identity $J_{pp}=\hat n_p$ derived above; the same-color (anti)commutators then
reproduce Eq.~\eqref{eq:gl_comm_Jpq} term by term, while every cross-color contribution
cancels. Thus the correspondence \(E_{pq}\mapsto J_{pq}\) furnishes a representation of \(\mathfrak{gl}(d,\mathbb C)\) by number-conserving many-body operators.

This motivates the definition of the Jordan--Schwinger map as the (Hermiticity preserving) linear map
\begin{equation}
  \Phi_{\mathrm{JS}}
  :
  \mathfrak{gl}(d,\mathbb C)
  \to
  \mathrm{End}\!\bigl(\mathcal F^{(\mathrm{stat})}\bigr),
  \qquad
  \Phi_{\mathrm{JS}}(X)
  :=
  \sum_{p,q=1}^{d}
  X_{pq}\,J_{pq},
  \label{eq:js_map}
\end{equation}
for any \(X=(X_{pq})\in\mathfrak{gl}(d,\mathbb C)\). By Eq.~\eqref{eq:gl_comm_Jpq}, \(\Phi_{\mathrm{JS}}\) is a Lie-algebra homomorphism:
\begin{equation}
  \Phi_{\mathrm{JS}}([X,Y])
  =
  [\Phi_{\mathrm{JS}}(X),\Phi_{\mathrm{JS}}(Y)],
  \qquad
  X,Y\in\mathfrak{gl}(d,\mathbb C).
\end{equation}

Restricting \(X\) to any traceless Hermitian generating set yields the corresponding representation of \(\mathfrak{su}(d)\) by number-conserving operators. In particular, if one adopts the Cartan--Weyl generators
\begin{equation}
  H_p := E_{pp}-E_{p+1,p+1},
  \qquad p=1,\dots,d-1,
\end{equation}
together with \(E_{pq}\) for \(p\neq q\), then the JS map gives
\begin{equation}
  \Phi_{\mathrm{JS}}(E_{pq}) = J_{pq},
  \qquad (p\neq q),
\end{equation}
and
\begin{equation}
  \Phi_{\mathrm{JS}}(H_p)
  =
  J_{pp}-J_{p+1,p+1}.
\end{equation}
Because the \(J_{pq}\) have been normalized so that \(J_{pp}=\hat n_p\) in every statistics sector, these Cartan generators are represented simply by differences of physical mode occupations. The same construction may equivalently be expressed in any other basis of \(\mathfrak{su}(d)\), for example the generalized Gell--Mann basis.

Adjoining the identity extends the map from \(\mathfrak{su}(d)\) to \(\mathfrak{u}(d)\cong \mathfrak{su}(d)\oplus\mathfrak{u}(1)\). Writing
\begin{equation}
  I=\sum_{p=1}^{d}E_{pp},
\end{equation}
one obtains
\begin{equation}
  \Phi_{\mathrm{JS}}(I)
  =
  \sum_{p=1}^{d}J_{pp}
  =
  \hat N.
  \label{eq:js_identity_to_number}
\end{equation}
Thus the central \(\mathfrak u(1)\) generator is represented exactly by the total particle-number operator.

We now connect this algebraic construction to the passage from the first-quantized operator expansion of Eq.~\eqref{eq:Hsumk_corrected} to the number-conserving second-quantized form of Eq.~\eqref{eq:secondQ_polyJ}. In first quantization, the operator is written as a finite sum of \(k\)-body terms expanded in single-particle matrix units. Identifying each outer product \(|p\rangle\langle q|\) with \(E_{pq}\) places this expansion inside the noncommutative polynomial algebra generated by the matrix units.

On the \(N\)-particle Hilbert space \((\mathbb C^d)^{\otimes N}\), it is natural to introduce the corresponding collective generators
\begin{equation}
  \Lambda_{pq}
  :=
  \sum_{\ell=1}^{N}
  E_{pq}^{(\ell)},
  \label{eq:collective_generators}
\end{equation}
which realize the diagonal action of \(\mathfrak u(d)\) on \((\mathbb C^d)^{\otimes N}\). These collective operators satisfy the same commutation relations as the \(E_{pq}\), and their Cartan--Weyl combinations generate the usual collective raising, lowering, and Cartan actions on the \(N\)-particle space. Accordingly, the class of first-quantized symmetry-preserving operators considered here is naturally organized as polynomials in the \(\Lambda_{pq}\), that is, as elements of the image of the universal enveloping algebra \(U(\mathfrak u(d))\) under the diagonal action.

Applying the Jordan--Schwinger realization amounts to replacing the abstract collective generators by their Fock-space images,
\begin{equation}
  \Lambda_{pq}\;\longmapsto\; J_{pq}=\Phi_{\mathrm{JS}}(E_{pq}),
\end{equation}
and extending this substitution multiplicatively to polynomials. Performing this replacement in the matrix-unit expansion of Eq.~\eqref{eq:Hsumk_corrected} yields the number-conserving second-quantized operator of Eq.~\eqref{eq:secondQ_polyJ}. For ordinary bosons and fermions, one may then further rewrite the result in the equivalent normal-ordered form of Eq.~\eqref{eq:secondQ_sumK}.

In this sense, the first- and second-quantized operator descriptions are two concrete realizations of the same underlying algebraic object: the former acts on the \(N\)-particle tensor-product space through the collective generators \(\Lambda_{pq}\), while the latter acts on the corresponding fixed-\(N\) sector of Fock space through the bilinears \(J_{pq}\). The role of statistics is entirely contained in the chosen realization of these generators on the target Fock space, whereas the underlying \(\mathfrak u(d)\) representation-theoretic structure is common to all four symmetry classes considered here.

\subsection{Schur--Weyl decomposition of the first-quantized Hilbert space}
\label{subsec:schur_wey_decomp}

Before relating first- and second-quantized wavefunctions, we briefly formalize the
representation-theoretic structure of the first-quantized Hilbert space
\(\mathcal{H}_{\mathrm{1Q}}=(\mathbb{C}^d)^{\otimes N}\) introduced in
Sec.~\ref{subsec:1q_struc_of_hilbert_space}, regarded here as \(N\) distinguishable registers,
each carrying a \(d\)-dimensional single-particle space. This space carries two natural commuting
actions: a collective \(U(d)\) action that rotates each tensor factor identically, and the
permutation action \(P(\pi)\) of the symmetric group \(S_N\) that permutes the tensor factors as
in Eq.~\eqref{eq:perm_action}. Explicitly, for \(U\in U(d)\),
\begin{equation}
U^{\otimes N}:\ |i_1\rangle\otimes\cdots\otimes|i_N\rangle
\ \longmapsto\ (U|i_1\rangle)\otimes\cdots\otimes(U|i_N\rangle).
\end{equation}
These actions commute,
\begin{equation}
[U^{\otimes N},\,P(\pi)] = 0
\qquad
\forall\,U\in U(d),\ \pi\in S_N,
\end{equation}
and, in fact, generate mutual commutants inside
\(\mathrm{End}\bigl((\mathbb{C}^d)^{\otimes N}\bigr)\). This is the content of
Schur--Weyl duality. Equivalently, one may regard the same structure at the Lie-algebra
level in terms of the commuting actions of \(\mathfrak u(d)\) and \(\mathbb C[S_N]\), or,
after complexification, \(\mathfrak{gl}(d,\mathbb C)\) and \(S_N\).

As a consequence, \(\mathcal{H}_{\mathrm{1Q}}\) decomposes as
\begin{equation}
(\mathbb{C}^d)^{\otimes N}
\ \cong\
\bigoplus_{\lambda\vdash N,\ \ell(\lambda)\le d}
\mathcal{Q}^{(d)}_{\lambda}\otimes \mathcal{P}_{\lambda}.
\label{eq:SW_decomp}
\end{equation}
Here \(\lambda\vdash N\) denotes a partition
\(\lambda=(\lambda_1,\lambda_2,\ldots)\) with
\(\lambda_1\ge\lambda_2\ge\cdots\) and \(\sum_i\lambda_i=N\), and \(\ell(\lambda)\) is the
number of nonzero parts, equivalently the number of rows of the associated Young diagram.
The constraint \(\ell(\lambda)\le d\) expresses the fact that only those \(U(d)\) irreducible
representations whose highest weights fit within \(d\) rows can appear in
\((\mathbb C^d)^{\otimes N}\). The factor \(\mathcal{Q}^{(d)}_{\lambda}\) carries the
\(U(d)\) irrep of highest weight \(\lambda\), while \(\mathcal{P}_{\lambda}\) carries the
Specht-module irrep of \(S_N\). Thus Eq.~\eqref{eq:SW_decomp} resolves the first-quantized
Hilbert space into orthogonal sectors labelled by \(\lambda\), each sector carrying an
irreducible permutation-symmetry type together with the compatible \(U(d)\) irrep acting on
the internal degrees of freedom.

To compute with this decomposition it is convenient to fix an orthonormal basis adapted to
it, the \emph{Schur basis} \(\{|\lambda,\mu,\sigma\rangle\}\), with
\begin{equation}
|\lambda,\mu,\sigma\rangle \in \mathcal{Q}^{(d)}_\lambda \otimes \mathcal{P}_\lambda.
\end{equation}
The shape \(\lambda\) selects the Schur--Weyl sector. Within that sector, the label \(\mu\)
indexes a basis of the \(U(d)\) irreducible representation \(\mathcal{Q}^{(d)}_\lambda\), and
the label \(\sigma\) indexes a basis of the \(S_N\) Specht module \(\mathcal{P}_\lambda\). We
develop the two labels in turn. We begin with \(\mu\), for which we adopt the
Gelfand--Tsetlin basis: it carries the \(U(d)\) weight, equivalently the mode-occupation
data, and is therefore the label that bridges directly to the second-quantized description
that is the focus of this work. We then turn to \(\sigma\), for which we adopt the
Young--Yamanouchi basis encoding the permutation symmetry.

\subsubsection{The \(U(d)\) label \(\mu\): Gelfand--Tsetlin patterns and semistandard tableaux}
\label{subsubsec:mu_label}

We take \(\mu\) to be a \emph{Gelfand--Tsetlin (GT) pattern}, a triangular array of integers
\begin{equation}
\mu\equiv
\left(
\begin{array}{cccc}
\mu^{(d)}_1 & \mu^{(d)}_2 & \cdots & \mu^{(d)}_{d} \\
\mu^{(d-1)}_1 & \mu^{(d-1)}_2 & \cdots & \mu^{(d-1)}_{d-1} \\
\vdots & \vdots & \ddots \\
\mu^{(2)}_1 & \mu^{(2)}_2 \\
\mu^{(1)}_1
\end{array}
\right),
\end{equation}
subject to the interlacing conditions
\begin{equation}
\mu^{(r)}_i \ge \mu^{(r-1)}_i \ge \mu^{(r)}_{i+1}
\qquad
\text{for } r=2,\ldots,d,\ \ i=1,\ldots,r-1.
\label{eq:GT_interlacing}
\end{equation}
The top row \(\mu^{(d)}\) is the highest weight of the \(U(d)\) irrep and, in the present
Schur--Weyl decomposition, is fixed by the partition \(\lambda\), padded with trailing zeros
to length \(d\) if necessary. A GT pattern therefore extends \(\lambda\) downward through
\(d-1\) interlacing rows.

Equivalently, and often more transparent combinatorially, the GT patterns for
\(\mathcal{Q}^{(d)}_\lambda\) are in bijection with \emph{semistandard Young tableaux} (SSYT)
of shape \(\lambda\) with entries in \(\{1,2,\ldots,d\}\). As is standard in the tableau
literature, the symbols \(1,\ldots,d\) are identified directly with the mode labels
\(1,\ldots,d\) used throughout the manuscript. An SSYT is a filling of the boxes of
\(\lambda\) with symbols \(1,\ldots,d\) that is weakly increasing along rows and strictly
increasing down columns.

The bijection may be described as follows. Given an SSYT \(S\), for each
\(r\in\{1,\ldots,d\}\) let \(\mu^{(r)}\) be the Young diagram whose \(i\)th row length equals
the number of boxes in row \(i\) of \(S\) whose entry is at most \(r\). Then
\begin{equation}
\mu^{(1)} \subseteq \mu^{(2)} \subseteq \cdots \subseteq \mu^{(d)}=\lambda
\end{equation}
defines a nested sequence of diagrams, equivalently an interlacing GT pattern; conversely,
every GT pattern determines a unique SSYT by reversing this procedure.

The content of the SSYT, namely the number of occurrences of each symbol \(r\), determines
the \(U(d)\) weight of the corresponding basis vector with respect to the diagonal generators
\(E_{rr}\): writing \(n_r\) for the number of entries equal to \(r\), the weight is
\((n_1,\ldots,n_d)\). Under the Jordan--Schwinger realization developed below, these same
components are realized as the mode occupations measured by the diagonal second-quantized
generators \(J_{rr}=\hat n_r\). The GT, or equivalently SSYT, label therefore already
encodes, on the first-quantized side, the occupation data that later appears explicitly in
Fock space.

\paragraph{Example: SSYT and corresponding GT pattern.}
Consider the nontrivial shape \(\lambda=(2,1)\), now with \(d=3\), and the semistandard
Young tableau
\begin{equation}
S=
\begin{array}{|c|c|}
\hline
1 & 1\\
\hline
2\\
\cline{1-1}
\end{array}.
\end{equation}
The first row repeats the symbol \(1\) and is therefore only \emph{weakly} increasing; this is
precisely the feature that distinguishes a semistandard tableau from a standard one, whose entries
are all distinct and whose rows are strictly increasing. The single column \((1,2)\) is strictly
increasing, so \(S\) is indeed semistandard. We now construct the associated GT pattern.

For \(r=1\), the boxes with entry at most \(1\) are the two cells of the first row, forming the
shape \((2)\), so
\begin{equation}
\mu^{(1)}=(2).
\end{equation}
For \(r=2\), all boxes have entry at most \(2\), giving the full shape \((2,1)\), so
\begin{equation}
\mu^{(2)}=(2,1).
\end{equation}
For \(r=3\), no further boxes are added, so the shape is unchanged,
\begin{equation}
\mu^{(3)}=(2,1,0).
\end{equation}
Hence the corresponding GT pattern is
\begin{equation}
\mu=
\left(
\begin{array}{ccc}
2 & 1 & 0\\
2 & 1\\
2
\end{array}
\right),
\end{equation}
and one readily checks the interlacing conditions
\[
2\ge 2\ge 1,
\qquad
1\ge 1\ge 0,
\qquad
2\ge 2\ge 1.
\]
The content of the tableau is \((2,1,0)\): the symbol \(1\) occurs twice and the symbol \(2\)
once. The corresponding \(U(3)\) weight therefore assigns occupation two to mode \(1\), occupation
one to mode \(2\), and zero to mode \(3\), i.e.\ the number state \(\ket{2,1,0}\) in second
quantization.

\subsubsection{The \(S_N\) label \(\sigma\): standard Young tableaux and Young--Yamanouchi words}
\label{subsubsec:sigma_label}

We take \(\sigma\) to be a \emph{Young--Yamanouchi word}, defined through standard Young
tableaux. A \emph{standard Young tableau} (SYT) \(T\) of shape \(\lambda\) is a filling of the
boxes of the Young diagram \(\lambda\) with \(1,2,\ldots,N\) such that entries increase
strictly along each row and strictly down each column. Given such a tableau \(T\), its
associated Young--Yamanouchi word is the sequence
\begin{equation}
\sigma(T) = (\sigma_1,\ldots,\sigma_N),
\qquad
\sigma_m := \text{row index of the box containing } m,
\label{eq:yamanouchi_word}
\end{equation}
where rows are numbered from top to bottom starting at \(1\). The sequence \(\sigma(T)\)
satisfies the ballot, or Yamanouchi, condition: for every prefix
\((\sigma_1,\ldots,\sigma_r)\) and every \(a\ge 1\), the number of occurrences of \(a\) is at
least the number of occurrences of \(a+1\). Conversely, a word
\(\sigma\in\{1,\ldots,\ell(\lambda)\}^N\) with content \(\lambda\) and satisfying the ballot
condition uniquely determines an SYT of shape \(\lambda\). We therefore use \(\sigma\)
interchangeably for the tableau itself and for its Young--Yamanouchi word.

In the Schur basis, the permutation operators \(P(\pi)\) act block-diagonally in \(\lambda\)
and mix only the \(\sigma\) indices within a fixed \(\lambda\) sector. The two extremes are
the one-row diagram \(\lambda=(N)\), whose Specht module is the trivial representation and
which yields the fully symmetric bosonic sector, and the one-column diagram
\(\lambda=(1,1,\ldots,1)\), whose Specht module is the sign representation and which yields
the fully antisymmetric fermionic sector. More general shapes \(\lambda\) carry the
higher-dimensional permutation-symmetry types and provide the natural first-quantized setting
for paraparticle sectors.

\paragraph{Example: SYT and Young--Yamanouchi word.}
Consider the nontrivial shape \(\lambda=(2,1)\), with three boxes arranged in two rows.
One standard Young tableau of this shape is
\begin{equation}
T=
\begin{array}{|c|c|}
\hline
1 & 3\\
\hline
2\\
\cline{1-1}
\end{array}.
\end{equation}
The entry \(1\) lies in row \(1\), the entry \(2\) lies in row \(2\), and the entry \(3\)
lies in row \(1\). Hence the associated Young--Yamanouchi word is
\begin{equation}
\sigma(T)=(1,2,1).
\end{equation}
Its prefixes \((1)\), \((1,2)\), and \((1,2,1)\) all satisfy the ballot condition, since
at every stage the number of \(1\)'s is at least the number of \(2\)'s.

\subsubsection{Action of the commuting symmetries in the Schur basis}
\label{subsubsec:schur_action}

Collecting these ingredients, the Schur basis vectors
\begin{equation}
|\lambda,\mu,\sigma\rangle
\in
\mathcal{Q}^{(d)}_\lambda\otimes\mathcal{P}_\lambda
\end{equation}
form an orthonormal basis of \(\mathcal H_{\mathrm{1Q}}\) in which the separation of the
commuting \(U(d)\) and \(S_N\) actions is manifest. Collective rotations \(U^{\otimes N}\)
act irreducibly on the \(\mu\) label within each fixed \(\lambda\) sector and trivially on
\(\sigma\), whereas permutations \(P(\pi)\) act irreducibly on the \(\sigma\) label within
each fixed \(\lambda\) sector and trivially on \(\mu\). This decomposition is the
representation-theoretic starting point for the bridge to second quantization developed
below, where the \(U(d)\) action is realized through the number-conserving bilinears
\(J_{pq}\) and the weight data carried by \(\mu\) becomes the occupation data of Fock space.

\subsection{Conversion between number and Schur basis states \label{subsec:number_to_schur_math}}

We now make explicit the relationship between the weight data encoded by the
particle-number-conserving second-quantized description and the Schur basis
$|\lambda,\mu,\sigma\rangle$ of the first-quantized Hilbert space. The key observation is
that, under the Jordan--Schwinger (JS) realization, the diagonal bilinears $J_{pp}$ furnish
the Cartan-subalgebra data of the induced $\mathfrak{u}(d)$ action. Consequently, Fock
occupation vectors are precisely the weight labels of this representation.

Recall from Sec.~\ref{subsec:js_map} that the JS map sends the Cartan generators
$H_p = E_{pp}-E_{p+1,p+1}$ (\(p=1,\ldots,d-1\)) to $\Phi_{\mathrm{JS}}(H_p)=J_{pp}-J_{p+1,p+1}$.
Because the one-body generators $J_{pq}$ were defined so that $J_{pp}=\hat n_p$ in every
statistics sector under consideration—bosonic, fermionic, parabosonic, and
parafermionic—it follows immediately that
\begin{equation}
\Phi_{\mathrm{JS}}(H_p)=\hat n_p-\hat n_{p+1},
\qquad p=1,\ldots,d-1.
\label{eq:Cartan_equals_numberdiff}
\end{equation}
Thus the diagonal $\mathfrak{su}(d)$ generators act, in the JS realization, as differences
of mode-occupation operators.

This identifies the occupation-number basis with a canonical weight basis. By definition,
the commuting operators $\{\hat n_p\}_{p=1}^{d}$ act diagonally on number states
$|n\rangle\equiv |n_1,\ldots,n_{d}\rangle$, and hence so do the Cartan differences
$\hat n_p-\hat n_{p+1}$. Combining this with Eq.~\eqref{eq:Cartan_equals_numberdiff} yields
\begin{equation}
\Phi_{\mathrm{JS}}(H_p)\,|n\rangle
=
(n_p-n_{p+1})\,|n\rangle,
\qquad p=1,\ldots,d-1.
\label{eq:weight_from_occup}
\end{equation}
Therefore every occupation basis state is a simultaneous eigenvector of the Cartan
subalgebra image and hence a weight vector for the induced $\mathfrak{su}(d)$ action. In
Dynkin, or simple-root, coordinates, the corresponding weight is
\begin{equation}
w(n)=(w_1,\ldots,w_{d-1}),
\qquad
w_p:=n_p-n_{p+1}.
\label{eq:dynkin_weight_from_n}
\end{equation}

Passing from $\mathfrak{su}(d)$ to $\mathfrak{u}(d)$ amounts to adjoining the identity
generator; by Eq.~\eqref{eq:js_identity_to_number} its JS image is the total number operator
$\Phi_{\mathrm{JS}}(I)=\hat N=\sum_{p=1}^{d}\hat n_p$. On a fixed-$N$ sector, the
$\mathfrak{u}(1)$ label is therefore simply $N$, so the full $\mathfrak{u}(d)$ weight data may
be taken as the pair $(w(n),N)$.

On the first-quantized side, Schur--Weyl duality decomposes
$(\mathbb C^d)^{\otimes N}$ into irreducible $U(d)$ sectors $\mathcal Q_\lambda^{(d)}$
tensored with $S_N$ sectors $\mathcal P_\lambda$ [Eq.~\eqref{eq:SW_decomp}], and the Schur basis
$|\lambda,\mu,\sigma\rangle$ is adapted to this decomposition. Fixing $\lambda$ and $\sigma$
selects a definite copy of the $U(d)$ irrep $\mathcal Q_\lambda^{(d)}$ inside
$(\mathbb C^d)^{\otimes N}$, while the label $\mu$, realized for example as a
Gelfand--Tsetlin pattern, indexes a basis of weight vectors for the $\mathfrak{u}(d)$ action
on that irrep. The highest weight of $\mathcal Q_\lambda^{(d)}$ is determined by the Young
diagram $\lambda=(\lambda_1,\ldots,\lambda_d)$, padded with trailing zeros if necessary, and
its Dynkin labels are
\begin{equation}
z_p=\lambda_p-\lambda_{p+1},
\qquad p=1,\ldots,d-1.
\label{eq:dynkin_from_young}
\end{equation}
Once the statistics sector is fixed, the admissible Young-diagram shape $\lambda$ is fixed
accordingly: $\lambda=(N)$ for bosons, $\lambda=(1,1,\ldots,1)$ with $N$ ones for
fermions, and more general $\lambda$ for paraparticle sectors.

The conversion between occupation vectors and Dynkin-coordinate weights is linear. Given an
occupation vector $n=(n_1,\ldots,n_{d})$ with total particle number
\begin{equation}
N=\sum_{p=1}^{d}n_p,
\end{equation}
define
\begin{equation}
w_p=n_p-n_{p+1},
\qquad p=1,\ldots,d-1.
\end{equation}
Conversely, given $w=(w_1,\ldots,w_{d-1})$ and $N$, one reconstructs the unique occupation
vector solving $n_p-n_{p+1}=w_p$ together with $\sum_p n_p=N$. Writing all occupations in
terms of $n_{d}$ and imposing the total-number constraint yields
\begin{align}
n_{d}
&=
\frac{1}{d}
\left(
N-\sum_{j=1}^{d-1}j\,w_j
\right),
\label{eq:n_from_w2_a}\\
n_p
&=
n_{d}+\sum_{j=p}^{d-1} w_j,
\qquad p=1,\ldots,d-1.
\label{eq:n_from_w2_b}
\end{align}
Thus, on a fixed-$N$ sector, occupation vectors and $\mathfrak{su}(d)$ weights in Dynkin
coordinates contain equivalent information.

The Schur label $\mu$ refines this weight data by resolving possible weight multiplicities
inside $\mathcal Q_\lambda^{(d)}$. Every Schur basis state
$|\lambda,\mu,\sigma\rangle$ has a well-defined $\mathfrak{u}(d)$ weight, and hence a
well-defined occupation vector $n(\mu)$ obtained from its Dynkin coordinates through
Eqs.~\eqref{eq:n_from_w2_a}--\eqref{eq:n_from_w2_b}. Equivalently, in the tableau language of
Sec.~\ref{subsec:schur_wey_decomp}, the content of the semistandard tableau associated with
$\mu$ is exactly the occupation vector, with tableau entries $1,\ldots,d$ identified directly
with the mode labels $1,\ldots,d$. The explicit conversion between GT data and
occupation vectors will be discussed further below.

We may now formalize the relationship between the second-quantized Jordan--Schwinger
realization and the Schur realization as an equivariant isomorphism of $U(d)$ modules.

\begin{boxedtheorem}[Equivariant identification of Schur and JS realizations]
\label{thm:equivariant_bijection}
Fix a Young diagram $\lambda\vdash N$ with $\ell(\lambda)\le d$, and fix a
Young--Yamanouchi label $\sigma$ of shape $\lambda$. Let
$\mathcal H_{\lambda,\sigma}\subset (\mathbb C^d)^{\otimes N}$ denote the corresponding
Schur--Weyl sector, so that
$\mathcal H_{\lambda,\sigma}\cong \mathcal Q_\lambda^{(d)}$ as a $U(d)$ module, with Schur
basis $\{|\lambda,\mu,\sigma\rangle\}_\mu$. Let $\mathcal F_N^{(\lambda)}$ denote the
fixed-$N$ second-quantized subspace carrying the same irreducible $U(d)$ representation
through the Jordan--Schwinger bilinears $\{J_{pq}\}$, equipped with an orthonormal basis
$\{|n,\eta\rangle\}$ adapted to the weight-space decomposition, where $n$ denotes the
occupation vector, equivalently the corresponding $\mathfrak{u}(d)$ weight, and $\eta$
indexes an orthonormal basis within the weight space of weight $n$. Then there exists a
unitary linear isomorphism
\begin{equation}
U_{\lambda,\sigma}:\mathcal H_{\lambda,\sigma}\to \mathcal F_N^{(\lambda)}
\end{equation}
such that
\begin{equation}
U_{\lambda,\sigma}\,\rho_\lambda(U)
=
\rho_{\mathrm{JS}}(U)\,U_{\lambda,\sigma},
\qquad \forall\,U\in U(d),
\label{eq:equivariance}
\end{equation}
where $\rho_\lambda$ denotes the irreducible action on $\mathcal Q_\lambda^{(d)}$, and
$\rho_{\mathrm{JS}}$ denotes the $U(d)$ representation generated by the Jordan--Schwinger
bilinears on $\mathcal F_N^{(\lambda)}$. In particular, $U_{\lambda,\sigma}$ maps each
weight space of $\mathcal H_{\lambda,\sigma}$ unitarily onto the corresponding weight space
of $\mathcal F_N^{(\lambda)}$. Once orthonormal bases are fixed within corresponding weight
spaces on both sides, it induces a bijection between the resulting basis states.
\end{boxedtheorem}

\begin{proof}
By Schur--Weyl duality, fixing $\lambda$ and $\sigma$ identifies
$\mathcal H_{\lambda,\sigma}$ with one copy of the irreducible $U(d)$ module
$\mathcal Q_\lambda^{(d)}$. On the second-quantized side, the Jordan--Schwinger bilinears
furnish the same irreducible $U(d)$ representation on $\mathcal F_N^{(\lambda)}$, with
Cartan action determined by the occupation data as in Eq.~\eqref{eq:weight_from_occup} and
highest weight fixed by the same Young diagram $\lambda$. Hence
$\mathcal H_{\lambda,\sigma}$ and $\mathcal F_N^{(\lambda)}$ realize equivalent irreducible
unitary $U(d)$ representations. Existence of a unitary intertwiner $U_{\lambda,\sigma}$
satisfying Eq.~\eqref{eq:equivariance} therefore follows from equivalence of irreducible
unitary representations. Because $U_{\lambda,\sigma}$ intertwines the full $U(d)$ action, it
intertwines in particular the commuting Cartan generators, and therefore preserves their
joint eigenspace decomposition. Thus each weight space on the Schur side is mapped unitarily
onto the corresponding weight space on the second-quantized side. Uniqueness holds up to an
overall phase by Schur's lemma.
\end{proof}

For general $\lambda$, the occupation vector $n$ specifies only the $\mathfrak{u}(d)$
weight, not a unique basis state. The additional label $\eta$ resolves possible degeneracies
within a given weight space on the second-quantized side, just as the GT label $\mu$
resolves them on the Schur side. Later in the manuscript we introduce a distinguished or
canonical GT label $\mu_{\mathrm{can}}(n,\lambda)$, which selects one preferred Schur basis
state in each nonempty weight space. This additional choice is only needed in sectors with
nontrivial weight multiplicities; in the ordinary bosonic and fermionic sectors, no such
refinement is required.

Finally, we note the operational consequence for changing representations of quantum states.
Consider a fixed-$N$ second-quantized state expanded in the occupation basis,
\begin{equation}
|\psi_{\mathrm{2Q}}\rangle
=
\sum_n c_n\,|n\rangle.
\label{eq:psi_number_expansion}
\end{equation}
Once the particle statistics, and hence the Young-diagram shape $\lambda$, are fixed, the
relevant Schur--Weyl sector in first quantization is fixed as well. Choosing a
Young--Yamanouchi label $\sigma$ for this $\lambda$ specifies an embedding
$\mathcal H_{\lambda,\sigma}\subset(\mathbb C^d)^{\otimes N}$. If one further chooses the
canonical GT representative $\mu_{\mathrm{can}}(n,\lambda)$ in each nonempty weight space,
then the corresponding basis-level conversion takes the form
\begin{equation}
|n\rangle
\longmapsto
|\lambda,\mu_{\mathrm{can}}(n,\lambda),\sigma\rangle,
\label{eq:canonical_basis_map}
\end{equation}
and therefore
\begin{equation}
|\psi_{\mathrm{1Q}}\rangle
=
\sum_n c_n\,|\lambda,\mu_{\mathrm{can}}(n,\lambda),\sigma\rangle.
\label{eq:psi_schur_expansion}
\end{equation}
This is the concrete realization of the change of representation used in practice below.

The choice of $\sigma$, and likewise the particular rule used to select a distinguished basis
vector within each fixed weight space, is not physically observable for the class of
permutation-invariant observables relevant to number-conserving dynamics considered here.
Indeed, any operator commuting with the full permutation action acts trivially on the $S_N$
factor within a fixed $\lambda$ sector and therefore cannot distinguish different
Young--Yamanouchi labels $\sigma$. Thus the physically relevant information in the conversion
is carried by the $U(d)$ sector: the occupation data identifies the corresponding weight
space, while the canonical GT rule selects a distinguished representative within that space.

\subsection{Second quantization as the group Fourier basis of first quantization}
\label{subsec:2Q_as_fourier_basis}

The equivariant identification of Theorem~\ref{thm:equivariant_bijection} admits a
sharper reading: it exhibits the occupation-number representation as a Fourier
basis of the first-quantized Hilbert space. We make this precise here. Two
ingredients combine. First, the change of basis to the Schur basis is the
generalized group Fourier transform of the commuting pair $(S_N,U(d))$ furnished by
Schur--Weyl duality. Second, the occupation vector attached to a Schur basis
state is its weight, i.e.\ the joint eigenvalue of the commuting number
operators; this is the commutative reduction of the full $U(d)$ label carried by
the Gelfand--Tsetlin pattern. For bosons and fermions this reduction is lossless
and the occupation labels are simply the Fourier coefficients re-indexed; for
sectors carrying a degenerate weight space it is lossy, and recovering the
identification is exactly what the canonical Gelfand--Tsetlin promise accomplishes.

\begin{proposition}[Schur basis change as a generalized group Fourier transform]
\label{prop:schur_is_gft}
Let $U_{\mathrm{Sch}}$ denote the unitary change of basis from the computational
basis of $\mathcal H_{\mathrm{1Q}}=(\mathbb C^d)^{\otimes N}$ to the Schur basis
$\{\ket{\lambda,\mu,\sigma}\}$, realized algorithmically by the strong quantum
Schur transform of Sec.~\ref{subsec:QST}. Then $U_{\mathrm{Sch}}$ simultaneously
reduces the two commuting representations of Eq.~\eqref{eq:SW_decomp} into
irreducibles:
\begin{align}
U_{\mathrm{Sch}}\,P(\pi)\,U_{\mathrm{Sch}}^\dagger
&= \bigoplus_{\lambda} I_{\mathcal Q^{(d)}_\lambda}\otimes \rho^{S_N}_\lambda(\pi),
&&\forall\,\pi\in S_N,
\\
U_{\mathrm{Sch}}\,U^{\otimes N}\,U_{\mathrm{Sch}}^\dagger
&= \bigoplus_{\lambda} \rho^{U(d)}_\lambda(U)\otimes I_{\mathcal P_\lambda},
&&\forall\,U\in U(d),
\end{align}
where the blocks $\mathcal P_\lambda$ are expressed in the Young--Yamanouchi basis
indexed by $\sigma$ and the blocks $\mathcal Q^{(d)}_\lambda$ in the
Gelfand--Tsetlin basis indexed by $\mu$. Thus $U_{\mathrm{Sch}}$ is a generalized
Fourier transform for each of the groups $S_N$ and $U(d)$: the label $\sigma$ is
the $S_N$-Fourier index and $\mu$ the $U(d)$-Fourier index.
\end{proposition}

\begin{proof}
By Schur--Weyl duality the permutation and collective actions are mutual
commutants and the space decomposes as in Eq.~\eqref{eq:SW_decomp}, with $P(\pi)$
acting only on the $\mathcal P_\lambda$ factors and $U^{\otimes N}$ only on the
$\mathcal Q^{(d)}_\lambda$ factors within each $\lambda$ block
(Sec.~\ref{subsec:schur_wey_decomp}). The Schur basis is by definition adapted to
this decomposition, with $\sigma$ ranging over the Young--Yamanouchi basis of
$\mathcal P_\lambda$ and $\mu$ over the Gelfand--Tsetlin basis of
$\mathcal Q^{(d)}_\lambda$. Writing the two actions in this basis gives the stated
block forms. Each is the reduction of a group representation into its irreducible
components, which is the defining property of a generalized Fourier transform; the
subgroup-adapted Gelfand--Tsetlin and Young--Yamanouchi bases additionally fix a
canonical labelling of basis vectors within each isotypic block, as in the
subgroup-adapted (fast) Fourier transforms over finite and compact groups.
\end{proof}

\begin{definition}[Weight map and weight multiplicity]
\label{def:weight_map}
Fix $\lambda\vdash N$ with $\ell(\lambda)\le d$. Let $\mathrm{wt}$ denote the
weight map sending a Gelfand--Tsetlin pattern $\mu$ of top row $\lambda$ to its
occupation vector $n(\mu)=(n_1,\ldots,n_d)$, $\sum_p n_p=N$, equivalently the
content of the associated semistandard tableau
(Sec.~\ref{subsec:schur_wey_decomp}). For a weight $n$, the number of patterns of
shape $\lambda$ and weight $n$ is the Kostka number
$K_{\lambda,n}:=\bigl|\mathrm{wt}^{-1}(n)\bigr|$, and
\begin{equation}
\dim \mathcal Q^{(d)}_\lambda=\sum_n K_{\lambda,n}.
\label{eq:dim_as_sum_kostka}
\end{equation}
The sector $\lambda$ is \emph{weight-multiplicity-free} if $K_{\lambda,n}\le 1$
for every $n$, i.e.\ if $\mathrm{wt}$ is injective. We write
$\Omega^{\mathrm{wt}}_\lambda:=\{n:K_{\lambda,n}\ge 1\}$ for the set of weights.
\end{definition}

\begin{lemma}[Weight-multiplicity-free classification: bosons and fermions]
\label{lem:bf_mult_free}
Assume $\ell(\lambda)\le d$.
\begin{enumerate}
\item For every $d$, the bosonic sector $\lambda=(N)$ is weight-multiplicity-free;
and for every $d\ge N$, the fermionic sector $\lambda=(1^N)$ is
weight-multiplicity-free. In both cases the weight map $\mathrm{wt}$ is a
bijection onto the admissible occupation vectors,
\begin{equation}
\Omega_{(N)}=\Bigl\{n\in\mathbb Z_{\ge0}^d:\textstyle\sum_p n_p=N\Bigr\},
\qquad
\Omega_{(1^N)}=\Bigl\{n\in\{0,1\}^d:\textstyle\sum_p n_p=N\Bigr\}.
\label{eq:bf_weight_sets}
\end{equation}
\item Conversely, if $d\ge N$ then bosons and fermions are the \emph{only}
weight-multiplicity-free sectors: a partition $\lambda\vdash N$ is
weight-multiplicity-free if and only if $\lambda=(N)$ or $\lambda=(1^N)$.
\end{enumerate}
\end{lemma}

\begin{proof}
\emph{(i)} A semistandard tableau of shape $(N)$ is a single weakly increasing row
of length $N$; a prescribed content $n$ determines it uniquely (the symbol $p$
occupies $n_p$ consecutive cells, in increasing order of $p$), so $K_{(N),n}=1$
for every weak composition $n$ of $N$ into $d$ parts. A semistandard tableau of
shape $(1^N)$ is a single strictly increasing column of length $N$; its content is
the indicator vector of the chosen $N$-element subset of $\{1,\ldots,d\}$, and the
filling is then forced, so $K_{(1^N),n}=1$ exactly for $n\in\{0,1\}^d$ with
$\sum_p n_p=N$ and $0$ otherwise. Injectivity of $\mathrm{wt}$ onto the stated sets
follows.

\emph{(ii)} Weight multiplicities are invariant under permutations of the content,
since the Weyl group $S_d$ permutes the weights of $\mathcal Q^{(d)}_\lambda$;
hence $K_{\lambda,n}$ depends only on the multiset of parts of $n$, and in
particular
\[
\max_n K_{\lambda,n}\ \ge\ K_{\lambda,(1^N)}
\qquad\text{whenever the content }(1^N)\text{ is admissible, i.e.\ }d\ge N.
\]
A semistandard tableau of shape $\lambda$ and content $(1^N)$ has all entries
distinct, so weak increase along rows coincides with strict increase, and such
fillings are exactly the standard Young tableaux of shape $\lambda$. Therefore
\begin{equation}
K_{\lambda,(1^N)}\ =\ f^\lambda\ =\ \dim\mathcal P_\lambda,
\label{eq:kostka_equals_flambda}
\end{equation}
the number $f^\lambda$ of standard Young tableaux of shape $\lambda$, which is also
the dimension of the Specht module $\mathcal P_\lambda$. It is classical that
$f^\lambda=1$ precisely when $\lambda$ is a single row $(N)$ or a single column
$(1^N)$, and $f^\lambda\ge 2$ for every other shape. Hence for $d\ge N$ a sector
$\lambda\neq(N),(1^N)$ has $K_{\lambda,(1^N)}=f^\lambda\ge 2$ and is not
weight-multiplicity-free, while $(N)$ and $(1^N)$ are weight-multiplicity-free by
part (i). This proves the stated equivalence.
\end{proof}

\begin{boxedtheorem}[Second quantization is the Fourier basis: bosonic and fermionic sectors]
\label{thm:2Q_fourier_bf}
Let $\lambda=(N)$ (bosons) or $\lambda=(1^N)$ with $d\ge N$ (fermions). In each
case the symmetric-group multiplicity space $\mathcal P_\lambda$ is
one-dimensional, so the label $\sigma$ takes a unique value, and by
Lemma~\ref{lem:bf_mult_free}(i) the weight map $\mathrm{wt}:\mu\mapsto n(\mu)$ is a
bijection with inverse $n\mapsto\mu(n)$. Then, for any first-quantized state in
the sector, written in the Schur (Fourier) basis as
\begin{equation}
\ket{\psi_{\mathrm{1Q}}}=\sum_{\mu}\alpha_{\mu}\,\ket{\lambda,\mu,\sigma}
\ \in\ \mathcal H_{\lambda,\sigma},
\end{equation}
the basis correspondence of Eq.~\eqref{eq:canonical_basis_map} acts as the
permutation $\ket{\lambda,\mu,\sigma}\leftrightarrow\ket{n(\mu)}$ of orthonormal
basis vectors and yields the second-quantized form
\begin{equation}
\ket{\psi_{\mathrm{2Q}}}=\sum_{\mu}\alpha_{\mu}\,\ket{n(\mu)}
=\sum_{n\in\Omega_\lambda} c_n\,\ket{n},
\qquad
c_n=\alpha_{\mu(n)} .
\label{eq:bf_coeff_identity}
\end{equation}
Equivalently, the occupation amplitudes $\{c_n\}$ are exactly the coordinates of
$\ket{\psi_{\mathrm{1Q}}}$ in the group-Fourier (Schur) basis of
Proposition~\ref{prop:schur_is_gft}, re-indexed by occupation vector. Hence, for
bosons and fermions, the second-quantized representation is precisely the
first-quantized state expressed in that Fourier basis: all nontrivial unitary
content of the conversion resides in the Schur basis change $U_{\mathrm{Sch}}$,
the arithmetic stage $U_{\mathrm{JS}}$ (Sec.~\ref{subsec:JS_unitary_from_GT})
being a permutation of basis states, and no canonical promise is required.
\end{boxedtheorem}

\begin{proof}
For $\lambda=(N)$ the module $\mathcal P_{(N)}$ is the trivial representation and
for $\lambda=(1^N)$ it is the sign representation; both are one-dimensional
(Sec.~\ref{subsec:schur_wey_decomp}), so $\sigma$ is unique and the
canonical-$\sigma$ promise is vacuous. By Lemma~\ref{lem:bf_mult_free}(i),
$\mathrm{wt}$ is a bijection, so each weight space of $\mathcal Q^{(d)}_\lambda$ is
one-dimensional and the canonical Gelfand--Tsetlin choice satisfies
$\mu_{\mathrm{can}}(n,\lambda)=\mu(n)$ automatically. Equation
\eqref{eq:canonical_basis_map} therefore reads
$\ket{n}\leftrightarrow\ket{\lambda,\mu(n),\sigma}$, a bijection between
orthonormal bases and hence a unitary permutation; its realization
$U_{\mathrm{JS}}$ maps basis states to basis states by reversible arithmetic and
thus introduces no superposition or relative phase among distinct branches.
Substituting into $\ket{\psi_{\mathrm{1Q}}}=\sum_\mu\alpha_\mu\ket{\lambda,\mu,\sigma}$
and relabelling by weight gives Eq.~\eqref{eq:bf_coeff_identity} with
$c_n=\alpha_{\mu(n)}$. Equivariance of the underlying identification is
Theorem~\ref{thm:equivariant_bijection}.
\end{proof}

\begin{boxedtheorem}[General sectors require the canonical promise]
\label{thm:2Q_fourier_general}
Let $\lambda\vdash N$ admit a weight $n_\star$ with $K_{\lambda,n_\star}\ge 2$;
this occurs for generic parastatistics sectors, the smallest instance being
$\lambda=(2,1)$ at $d\ge3$, for which $K_{(2,1),(1,1,1)}=2$. Then:
\begin{enumerate}
\item the weight map $\mathrm{wt}$ is not injective, the occupation vector is a
strictly coarser label than the Gelfand--Tsetlin pattern, and by
Eq.~\eqref{eq:dim_as_sum_kostka}
\begin{equation}
\bigl|\Omega^{\mathrm{wt}}_\lambda\bigr|
\;<\;\sum_n K_{\lambda,n}=\dim\mathcal Q^{(d)}_\lambda;
\end{equation}
\item consequently no re-indexing of basis vectors identifies the full sector
$\mathcal Q^{(d)}_\lambda$ with the occupation-indexed space
$\mathrm{span}\{\ket{n}\}_{n\in\Omega^{\mathrm{wt}}_\lambda}$, the latter having
strictly smaller dimension;
\item fixing a canonical section $n\mapsto\mu_{\mathrm{can}}(n,\lambda)$ (one
pattern per nonempty weight space) and restricting to the canonical subspace
\begin{equation}
\mathcal H^{\mathrm{can}}_{\lambda,\sigma}
:=\mathrm{span}\bigl\{\ket{\lambda,\mu_{\mathrm{can}}(n,\lambda),\sigma}
:n\in\Omega^{\mathrm{wt}}_\lambda\bigr\},
\end{equation}
the restriction $\mathrm{wt}|_{\mathcal H^{\mathrm{can}}_{\lambda,\sigma}}$ is a
bijection, and the conclusion of Theorem~\ref{thm:2Q_fourier_bf} holds verbatim on
$\mathcal H^{\mathrm{can}}_{\lambda,\sigma}$: with
$\ket{\psi_{\mathrm{1Q}}}=\sum_n c_n\ket{\lambda,\mu_{\mathrm{can}}(n,\lambda),\sigma}$
one obtains $\ket{\psi_{\mathrm{2Q}}}=\sum_n c_n\ket{n}$.
\end{enumerate}
Thus the identification of second quantization with the group-Fourier basis of
first quantization holds unconditionally for bosons and fermions, and for all
remaining sectors precisely under the canonical Gelfand--Tsetlin promise
(together with a fixed $\sigma_{\mathrm{can}}$).
\end{boxedtheorem}

\begin{proof}
For $\lambda=(2,1)$ and $d\ge3$ the two semistandard tableaux of content
$(1,1,1)$,
\[
\ytableaushort{12,3}
\qquad\text{and}\qquad
\ytableaushort{13,2},
\]
both satisfy the semistandard conditions, while $\ytableaushort{23,1}$ violates
column strictness; hence $K_{(2,1),(1,1,1)}=2$, establishing the hypothesis.
Whenever some $K_{\lambda,n_\star}\ge2$, the fibre $\mathrm{wt}^{-1}(n_\star)$
contains at least two distinct patterns, so $\mathrm{wt}$ is not injective and
(i) follows from Eq.~\eqref{eq:dim_as_sum_kostka}. For (ii), a re-indexing is a
bijection of basis sets and therefore preserves cardinality; since
$\dim\mathrm{span}\{\ket n\}=|\Omega^{\mathrm{wt}}_\lambda|<\dim\mathcal Q^{(d)}_\lambda$,
no such bijection exists. For (iii), the section
$n\mapsto\mu_{\mathrm{can}}(n,\lambda)$ selects exactly one pattern in each
nonempty weight space, so $\mathrm{wt}$ restricted to
$\{\mu_{\mathrm{can}}(n,\lambda)\}$ is a bijection onto
$\Omega^{\mathrm{wt}}_\lambda$; the argument of Theorem~\ref{thm:2Q_fourier_bf}
then applies with $\mu(n)$ replaced by $\mu_{\mathrm{can}}(n,\lambda)$, using the
basis correspondence of Eq.~\eqref{eq:canonical_basis_map}. Equivariance is again
Theorem~\ref{thm:equivariant_bijection}.
\end{proof}

\begin{remark}
\label{rem:fourier_upshot}
Conceptually, the first-quantized computational (particle-configuration) basis
plays the role of the pre-Fourier basis, the Schur basis change $U_{\mathrm{Sch}}$
is the group-Fourier transform of the $(S_N,U(d))$ pair, and the occupation
vector is the weight, i.e.\ the joint eigenvalue of the commuting number
operators $\{\hat n_p=J_{pp}\}$ extracted from the Gelfand--Tsetlin label. The
full pattern $\mu$ refines this weight whenever a weight space is degenerate;
second quantization retains only the weight. Two features make bosons and fermions
special, and by Eq.~\eqref{eq:kostka_equals_flambda} both are governed by the
single condition $f^\lambda=1$. First, for every $d$ the Specht module
$\mathcal P_\lambda$ is one-dimensional exactly for $\lambda=(N)$ and
$\lambda=(1^N)$, so the $S_N$ label $\sigma$ is then redundant and no fixed
$\sigma_{\mathrm{can}}$ promise is needed. Second, in the regime $d\ge N$ these are
also exactly the weight-multiplicity-free sectors
[Lemma~\ref{lem:bf_mult_free}(ii)], so the weight is already a complete Fourier
label and the conversion is a lossless re-indexing. For $d<N$ the
multiplicity-free property is no longer exclusive to bosons and fermions---for
instance $\lambda=(2,1)$ at $d=2$ and $\lambda=(2,2)$ at $d=3$ are
weight-multiplicity-free---but such sectors still carry a higher-dimensional
$\mathcal P_\lambda$ and therefore still require the fixed-$\sigma_{\mathrm{can}}$
promise; the genuinely universal characterization of Bose and Fermi statistics is
thus the one-dimensionality of $\mathcal P_\lambda$. In every degenerate-weight
sector the weight is a genuine coarsening, and the canonical Gelfand--Tsetlin
promise restores a one-to-one Fourier labelling by selecting a single
representative per weight space.
\end{remark}

\section{Unitary transformation from first to second quantization}
\label{sec:qqt}

\begin{figure}
    \centering
    \includegraphics[width=0.8\linewidth]{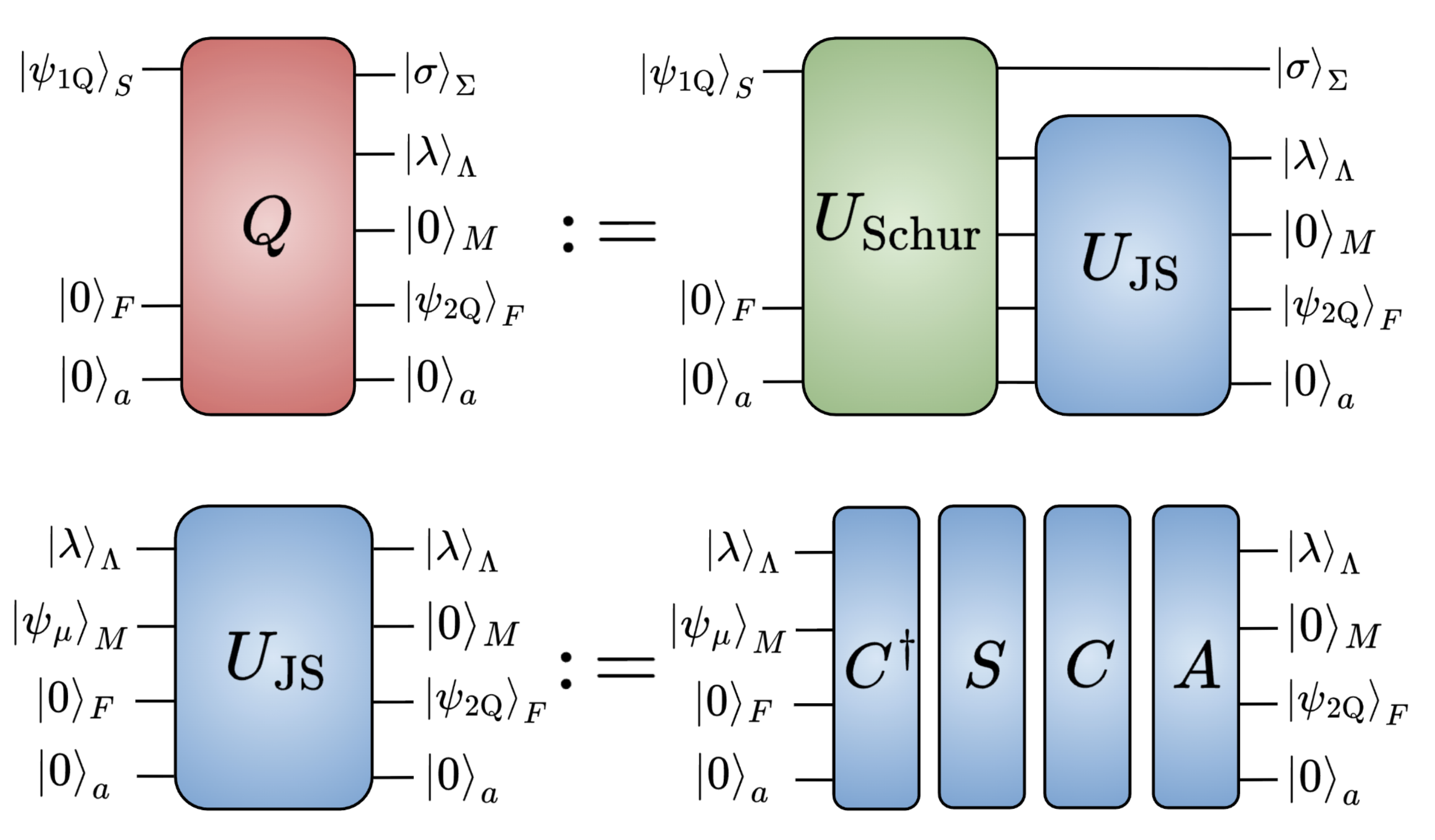}
    \caption{Quantum circuit describing the unitary transformation
    $Q := U_{\mathrm{JS}}U_{\mathrm{Schur}}$ between an input first-quantized wavefunction
    on the register $S$, initially in the state $|\psi_{\mathrm{1Q}}\rangle_S$, and an
    output second-quantized wavefunction on the register $F$, in the state
    $|\psi_{\mathrm{2Q}}\rangle_F$. The transformation is defined on the promised subspace
    in which, for each occupation vector, only the canonical Gelfand--Tsetlin label is
    populated. The Jordan--Schwinger unitary
    $U_{\mathrm{JS}} = C^\dagger S C A$ is described in
    Sec.~\ref{subsec:JS_unitary_from_GT}, where explicit reversible arithmetic
    implementations of the components $C$, $A$, and $S$ are given.}
    \label{fig:Q_def}
\end{figure}

We now describe a coherent, fully unitary procedure that converts a first-quantized pure state in a fixed Schur--Weyl symmetry sector into its second-quantized representation. Concretely, fix a particle number $N$ and a single-particle dimension $d$, and consider the first-quantized Hilbert space $\mathcal{H}_{\mathrm{1Q}}$ (Sec.~\ref{subsec:schur_wey_decomp}), decomposed as in Eq.~\eqref{eq:SW_decomp}. Restricting to a single particle-statistics sector means that the input state lies in a single Young-diagram sector $\lambda$.

As discussed in Sec.~\ref{subsec:number_to_schur_math}, for general $\lambda$ the occupation vector determines only the $\mathfrak{u}(d)$ weight, and distinct GT labels may belong to the same weight space. In order to define a basis-level unitary conversion directly between first- and second-quantized wavefunctions without introducing an explicit degeneracy label on the second-quantized side, we impose the following promise: within each nonempty weight space of the Schur realization, only the distinguished canonical GT label $\mu_{\mathrm{can}}(n,\lambda)$ is populated. Equivalently, the input state is assumed to have support only on Schur basis states of the form $|\lambda,\mu_{\mathrm{can}}(n,\lambda),\sigma\rangle$.

A second promise is required for the $S_N$ multiplicity label. We assume that the input state is supported on a single fixed Young--Yamanouchi label, denoted $\sigma_{\mathrm{can}}$. 

In the ordinary bosonic and fermionic sectors these choices are automatic, since there is only one admissible Young--Yamanouchi label and no weight space degeneracy (Theorem~\ref{thm:2Q_fourier_bf}). In more general parastatistics sectors, any valid fixed choice may be made, provided the input support is restricted to that choice throughout. 

Under these two promises, the input first-quantized state may be written as
\begin{equation}
|\psi_{\mathrm{1Q}}\rangle
=
\sum_{n}
c_{n}\,
|\lambda,\mu_{\mathrm{can}}(n,\lambda),\sigma_{\mathrm{can}}\rangle,
\qquad
\text{with fixed }\lambda \text{ and fixed }\sigma_{\mathrm{can}},
\label{eq:psi_fixed_lambda_canonical}
\end{equation}
where $n$ ranges over the admissible occupation vectors in the irrep of highest weight $\lambda$.

Our goal is to implement a unitary $Q$ able to convert general first quantized wavefunctions for identical particles to their second-quantized form. At the level of basis states, the unitary $Q$ is defined by
\begin{equation}
Q:\ |\lambda,\mu_{\mathrm{can}}(n,\lambda),\sigma_{\mathrm{can}}\rangle_S\,|0 \rangle_F |0\rangle_a
\longmapsto
|\lambda\rangle_{\Lambda}\,|n\rangle_F\,|\sigma_{\mathrm{can}}\rangle_{\Sigma} |0 \rangle_M |0 \rangle_a.
\label{eq:QQT_basis_action_canonical}
\end{equation}
Here $|0\rangle_a$ is an ancillary register used during the computation, but uncomputed at the end, and the remaining subscripts label registers in accordance with Fig.~\ref{fig:Q_def} and Algorithm~\ref{alg:QQT_highlevel2}. Because the map $n\mapsto \mu_{\mathrm{can}}(n,\lambda)$ is, by construction, a single-valued choice of one Schur basis state in each nonempty weight space, and because $\sigma_{\mathrm{can}}$ is held fixed, the above assignment is bijective on the promised subspace and therefore extends linearly to a unitary transformation.

Thus, on an arbitrary superposition satisfying the promises, the output takes the form
\begin{equation}
Q\,|\psi_{\mathrm{1Q}}\rangle_S\,|0 \rangle_F |0\rangle_a
=
|\lambda\rangle_{\Lambda}
\left(
\sum_{n}
c_n\,|n\rangle_F
\right)
|\sigma_{\mathrm{can}}\rangle_{\Sigma}|0 \rangle_M|0 \rangle_a
=
|\lambda\rangle_{\Lambda}\,|\psi_{\mathrm{2Q}}\rangle_F\,|\sigma_{\mathrm{can}}\rangle_{\Sigma}|0 \rangle_M |0 \rangle_a  ,
\label{eq:QQT_target_number_state_canonical}
\end{equation}
where $|\psi_{\mathrm{2Q}}\rangle_F := \sum_n c_n\,|n\rangle_F$ with inverse

\begin{equation}
Q^\dagger
\left(
|\lambda\rangle_{\Lambda}\,
|\psi_{\mathrm{2Q}}\rangle_F\,
|\sigma_{\mathrm{can}}\rangle_{\Sigma}|0\rangle_M |0 \rangle_a
\right)
=
\left[ \sum_n \widetilde c_n\,
|\lambda,\mu_{\mathrm{can}}(n,\lambda),\sigma_{\mathrm{can}}\rangle_S\right]|0\rangle_F |0 \rangle_a
\label{eq:psi_schur_expansion_again}
\end{equation}
and $|\psi_{\text{1Q}} \rangle _S := \sum_n \widetilde c_n\,
|\lambda,\mu_{\mathrm{can}}(n,\lambda),\sigma_{\mathrm{can}}\rangle_S$. At a high level, our construction of $Q$ is a composition of two unitary steps,
summarized in Algorithm~\ref{alg:QQT_highlevel2}. First, we apply a strong quantum Schur transform $U_{\mathrm{Schur}}$, which maps the computational basis of $\mathcal{H}_{\mathrm{1Q}}$ to the Schur basis and coherently writes the labels $(\lambda,\mu,\sigma)$ into dedicated registers. On an input state supported on a fixed $\lambda$, a fixed $\sigma_{\mathrm{can}}$, and satisfying the canonical-GT-label promise, this step extracts the superposition of canonical GT labels $\mu_{\mathrm{can}}(n,\lambda)$ together with the fixed Young--Yamanouchi label $\sigma_{\mathrm{can}}$.

Second, we apply a reversible arithmetic unitary $U_{\mathrm{JS}}$ that computes the occupation vector $n$ directly from the GT data. The key structural point is that the Jordan--Schwinger realization identifies the Cartan generators with number-operator differences, Eq.~\eqref{eq:Cartan_equals_numberdiff}, while the GT label encodes the same Cartan eigenvalues through the row-sum identities derived below. Because the input is promised to lie on the canonical branch, the reversible map from GT data to occupation data may be restricted to the canonical set $\mu=\mu_{\mathrm{can}}(n,\lambda)$, on which it is one-to-one. This allows us to bypass any explicit multiplicity label.

Putting these ingredients together yields a unitary transform that has the form
\begin{equation}
Q
=
U_{\mathrm{JS}}\,U_{\mathrm{Schur}}.
\label{eq:Q_factorization_again}
\end{equation}
The remainder of this section provides the implementation of each component.

\begin{algorithm}[H]
\caption{Unitary transformation from first to second quantized quantum states (high level)}
\label{alg:QQT_highlevel2}
\begin{algorithmic}[1]
\Require Particle number $N$, single-particle dimension $d$
\Require Input $\ket{\psi_{\mathrm{1Q}}}\in(\mathbb{C}^d)^{\otimes N}$ supported on a fixed $\lambda$ and canonical choices of $\sigma = \sigma_{\text{can}}$ and $\mu = \mu_{\text{can}}$
\State \textbf{Registers / initialization}
\State $S$: first-quantized system register, initialized to $\ket{\psi_{\mathrm{1Q}}}$
\State $a$: ancillary register
\State $F$: occupation register, initialized to $\ket{0}$ 
\Statex
\State \textbf{Step 1 --- Schur extraction.}
\State Apply $U_{\mathrm{Schur}}$ (Algorithm~\ref{alg:U_schur_placeholder}) on $(S,a)$ to coherently write $(\lambda,\mu,\sigma)$ into $(\Lambda,M,\Sigma)$.
\Statex
\State \textbf{Step 2 --- Occupations from GT (Jordan--Schwinger unitary).}
\State Apply $U_{\mathrm{JS}}$ (Algorithm~\ref{alg:U_JS_highlevel_clean}) on $(M,F, \Lambda, a)$ to compute the occupation vectors $n(\mu)=(n_1,\ldots,n_d)$ into $F$, then uncompute register $M$.
\Statex
\Ensure $F$ holds the coherent number-state superposition corresponding to the input in $S$.
\Ensure $\Lambda$ holds the Young diagram shape
\Ensure $\Sigma$ holds a single canonical Young--Yamanouchi label
\Ensure $a$ is uncomputed and holds $|0 \rangle$ 
\end{algorithmic}
\end{algorithm}

\subsection{Quantum Schur transform \label{subsec:QST}}

The central primitive in Step~1 of Algorithm~\ref{alg:QQT_highlevel2} is a \emph{forward quantum Schur transform}, namely a unitary change of basis on the first-quantized register $S\simeq(\mathbb{C}^d)^{\otimes N}$ that makes the Schur--Weyl decomposition of Eq.~\eqref{eq:SW_decomp} explicit at the level of computational registers. Recall that this decomposition resolves $(\mathbb{C}^d)^{\otimes N}$ into sectors $\mathcal{Q}^{(d)}_\lambda\otimes\mathcal{P}_\lambda$ with the adapted Schur basis $\{\ket{\lambda,\mu,\sigma}\}$ of Sec.~\ref{subsec:schur_wey_decomp}. Practically, the forward Schur transform is the unitary that maps an arbitrary superposition in the computational basis to the corresponding superposition in the Schur basis, coherently writing the labels into dedicated registers:
\begin{equation}
U_{\mathrm{Schur}}:\ \ket{\psi}_{\text{1Q}}\ket{0}_{a}
\ \longmapsto\
\left[ \sum_{\lambda,\mu,\sigma}
\alpha_{\lambda,\mu,\sigma}\,
\ket{\lambda}_{\Lambda}\ket{\mu}_{M}\ket{\sigma}_{\Sigma} \right]\ket{0}_a
\label{eq:U_schur_map_registers}
\end{equation}
With our input promises for $\ket{\psi_{\text{1Q}}}$ ($\mu = \mu_{\text{can}}, \sigma=\sigma_{\text{can}}$ and a fixed single Young diagram $\lambda$), the action of Eq. \eqref{eq:U_schur_map_registers} specializes as
\begin{equation}
U_{\mathrm{Schur}}:\ \ket{\psi}_{\text{1Q}}\ket{0}_{a}
\ \longmapsto\
\ket{\lambda}_{\Lambda}\left[ \sum_{\mu_{\text{can}}}
\alpha_{\mu_{\text{can}}}\,
\ket{\mu_{\text{can}}}_{M}  \right]\ket{\sigma_{\text{can}}}_{\Sigma}\ket{0}_a := \ket{\lambda}_{\Lambda}|\psi_{\mu_{\text{can}}} \rangle_M \ket{\sigma_{\text{can}}}_{\Sigma}\ket{0}_a 
\label{eq:U_schur_map_registers_input_promise}
\end{equation}
With $|\psi_{\mu_{\text{can}}} \rangle_M := \sum_{\mu_{\text{can}}} \alpha_{\mu_{\text{can}}} \ket{\mu_{\text{can}}}_{M} $. A key requirement for our unitary transformation is that $U_{\mathrm{Schur}}$ be a \emph{strong} Schur transform rather than a weak ``Schur sampling'' routine. Following the standard terminology, weak Schur sampling means producing only the irrep label $\lambda$ with the correct distribution when measuring, whereas strong Schur sampling means producing the full label pair $(\lambda,i)$ for a basis index $i$ inside the $\lambda$ block, which in the Schur--Weyl setting corresponds to producing $(\lambda,\mu,\sigma)$ coherently \cite{Krovi2019Schur}. We require strong access because the subsequent arithmetic map requires a coherent $\mu$ register, not merely an outcome distribution over $\lambda$.

There are two strong Schur transforms in the literature that are particularly well-suited to our setting. The first is the original construction of Bacon, Chuang, and Harrow (BCH), which gives an explicit quantum circuit for the Schur transform on $(\mathbb{C}^d)^{\otimes N}$ with gate complexity polynomial in $N$, $d$, and $\log(\epsilon^{-1})$ for approximation error $\epsilon$ \cite{Bacon2006Schur}. Conceptually, the BCH approach builds the Schur basis by iteratively coupling one $d$-level system at a time and performing the corresponding reduced Clebsch--Gordan transforms. The second is the high-dimensional approach initiated by Krovi, which targets regimes where $d$ is large and achieves complexity polynomial in $N$, $\log d$, and $\log(\epsilon^{-1})$ \cite{Krovi2019Schur}. This alternative line was recently revisited and strengthened: a corrected version of Krovi's algorithm (hereafter referred to as the Krovi-Burchardt algorithm [KB]) and a detailed high-dimensional treatment of the BCH transform were provided by Burchardt \emph{et al.}, who also make explicit the resulting high-dimensional scalings, with Krovi's approach scaling as $\mathcal{O}(N^4)$ and the BCH approach scaling as $\mathcal{O}(\min(N^5,\,N d^4))$ in the regime $N<d$ \cite{Burchardt2025Schur}. In practice, both transforms are valid options for Step~1, and our construction is compatible with either choice. The selection is primarily a resource-engineering decision driven by the $(N,d,\epsilon)$ regime and by the desired constant factors.

Finally, we note that strong Schur transforms are most naturally described as circuits acting on $d$-level qudits. Our first-quantized system register $S$ is typically implemented on qubit hardware using an encoding of each qudit into $\lceil\log_2 d\rceil$ qubits, as already discussed in Sec.~\ref{subsec:1q_resource_tradeoff}. Any fixed universal gate set for qudits can be compiled into qubit gates with polynomial overhead in $\log d$ and in the target synthesis precision, so a qudit-level Schur transform can be realized within a standard qubit model of computation \cite{NielsenChuang2010}.

\begin{algorithm}[H]
\caption{$U_{\mathrm{Schur}}$ --- Strong Schur transform for $(\mathbb{C}^d)^{\otimes N}$}
\label{alg:U_schur_placeholder}
\begin{algorithmic}[1]
\Require System register $S$ holding an arbitrary superposition in $(\mathbb{C}^d)^{\otimes N}$
\Require Ancillary register $a$
\Require Target approximation error $\epsilon$
\Require Qubit-level decomposition of native $d$-level operators, if qudits are unavailable
\State Choose one of the following two strong Schur transform implementations: \quad (i) Apply the BCH strong Schur transform circuit of Ref.~\cite{Bacon2006Schur}, using $\mathrm{poly}(N,d,\log(\epsilon^{-1}))$ gates. (ii) Apply the high-dimensional strong Schur transform of Ref.~\cite{Krovi2019Schur} using the corrected and optimized construction in Ref.~\cite{Burchardt2025Schur}, achieving $\mathrm{poly}(N,\log d,\log(\epsilon^{-1}))$ scaling (notably $\tilde{O}(N^4)$ in the regime $N<d$).
\State Output: $\Lambda$ holds general superpositions over $\ket{\lambda}$
\State $M$ holds general superpositions over $\ket{\mu}$
\State $\Sigma$ holds general superpositions over $\ket{\sigma}$
\State $a$ is uncomputed to $|0 \rangle$
\end{algorithmic}
\end{algorithm}

\subsection{Jordan--Schwinger arithmetic unitary: canonical GT patterns to occupation states \label{subsec:JS_unitary_from_GT}}

Step~2 of Algorithm~\ref{alg:QQT_highlevel2} is a reversible arithmetic unitary that computes the occupation vector superposition corresponding to the weight data stored in a canonical GT-pattern superposition, followed by coherent uncomputation of the GT register $M$. We define the Jordan--Schwinger arithmetic unitary $U_{\mathrm{JS}}$ by
\begin{equation}
U_{\mathrm{JS}}:\ |\lambda \rangle_{\Lambda}\ket{\mu_{\mathrm{can}}}_{M}\ket{0}_{F}\ket{0}_{W}
\ \longmapsto\
|\lambda \rangle_{\Lambda}\ket{0}_{M}\ket{n(\mu_{\mathrm{can}})}_{F}\ket{0}_{W},
\label{eq:U_JS_new_def}
\end{equation}
where $W$ is an auxiliary work register (a subset of the ancillary register $a$). By linearity, for any coherent superposition supported on canonical GT labels,
\begin{equation}
\sum_{\mu_{\mathrm{can}}}\alpha_{\mu_{\mathrm{can}}}|\lambda \rangle_{\Lambda}\ket{\mu_{\mathrm{can}}}_{M}\ket{0}_{F}\ket{0}_{W}
\ \xmapsto{\,U_{\mathrm{JS}}\,}\
|\lambda \rangle_{\Lambda}\ket{0}_{M}\left[\sum_{\mu_{\mathrm{can}}}\alpha_{\mu_{\mathrm{can}}}\ket{n(\mu_{\mathrm{can}})}_{F}\right]\ket{0}_{W}.
\label{eq:UJS_linearity_clean}
\end{equation}
The $\Sigma$ register has been omitted here, since $U_{\text{JS}}$ does not target this register. It is clear from Eq.~\eqref{eq:UJS_linearity_clean} that the occupation register $F$ now carries the full second-quantized wavefunction $| \psi_{\text{2Q}} \rangle := \sum_{\mu_{\mathrm{can}}}\alpha_{\mu_{\mathrm{can}}}\ket{n(\mu_{\mathrm{can}})}_{F}$.

To achieve the map in Eq.~\eqref{eq:U_JS_new_def}, we decompose $U_{\mathrm{JS}}$ into four high-level reversible steps. Let $A$ denote the arithmetic map that computes the occupation vector from a canonical GT pattern,
\begin{equation}
A:\ \ket{\mu_{\mathrm{can}}}_{M}\ket{0}_{F}\ket{0}_{W}
\ \longmapsto\
\ket{\mu_{\mathrm{can}}}_{M}\ket{n(\mu_{\mathrm{can}})}_{F}\ket{0}_{W},
\label{eq:A_def_highlevel}
\end{equation}
and let $C$ denote a reversible reconstruction map that computes from the occupation vector the \emph{lower rows} (all rows but row $d$, which is fixed to $\lambda)$ of the canonical GT representative into the work register, using the pre-existing $\Lambda$ register as the top row. More explicitly, writing the full GT pattern as
\[
\mu_{\mathrm{can}} = \bigl(\lambda,\mu_{\mathrm{can},\downarrow}\bigr),
\]
where $\lambda$ is the fixed top row and $\mu_{\mathrm{can},\downarrow}$ denotes the remaining $d-1$ rows, we define
\begin{equation}
C:\ \ket{\lambda}_{\Lambda}\ket{\mu_{\mathrm{can}}}_{M}\ket{n}_{F}\ket{0}_{W}
\ \longmapsto\
\ket{\lambda}_{\Lambda}\ket{\mu_{\mathrm{can}}}_{M}\ket{n}_{F}\ket{\mu_{\mathrm{can},\downarrow}(n,\lambda)}_{W},
\label{eq:C_def_highlevel}
\end{equation}
where $\mu_{\mathrm{can},\downarrow}(n,\lambda)$ denotes the lower rows of the chosen canonical GT pattern associated with top row $\lambda$ and occupation vector $n$. On the promised subspace of canonical GT inputs, one has
\begin{equation}
\mu_{\mathrm{can},\downarrow}(n(\mu_{\mathrm{can}}),\lambda)=\mu_{\mathrm{can},\downarrow}.
\label{eq:C_consistency_highlevel}
\end{equation}
Thus $C$ does not rewrite the top GT row into the work register; rather, the top row is already available in $\Lambda$, and $C$ reconstructs only the part of the GT data not already fixed by the Schur transform output.

We also introduce a reversible subtraction unitary $S$ whose action is to subtract the joint content of the
work and $\Lambda$ register from the GT register,
\begin{equation}
\begin{aligned}
S:\ &|\lambda \rangle_\Lambda\ket{\mu_{\mathrm{can}}}_{M}\ket{n(\mu_{\mathrm{can}})}_{F}\ket{\mu_{\mathrm{can},\downarrow}(n,\lambda)}_{W}
\ \longmapsto\ 
\ket{\mu_{\mathrm{can}}-(\lambda,\mu_{\mathrm{can},\downarrow}(n,\lambda))}_{M}\ket{n(\mu_{\mathrm{can}})}_{F}\ket{\mu_{\mathrm{can},\downarrow}(n,\lambda)}_{W} \\
&= \ket{0}_{M}\ket{n(\mu_{\mathrm{can}})}_{F}\ket{\mu_{\mathrm{can},\downarrow}(n,\lambda)}_{W},
\end{aligned}
\label{eq:S_def_highlevel}
\end{equation}
where $\mu_{\mathrm{can}}=(\lambda,\mu_{\mathrm{can},\downarrow})$. 

On the promised canonical subspace, the subtraction step leaves the work register $W$ holding the reconstructed lower rows $\mu_{\mathrm{can},\downarrow}$ while the GT register $M$ has been uncomputed to $\ket{0}$. A final application of the inverse reconstruction unitary $C^\dagger$ then erases this residual data from $W$, returning the work register to $\ket{0}$ while leaving the occupation register unchanged:
\begin{equation}
C^\dagger:\ \ket{\lambda}_{\Lambda}\ket{0}_{M}\ket{n(\mu_{\mathrm{can}})}_{F}\ket{\mu_{\mathrm{can},\downarrow}}_{W}
\ \longmapsto\
\ket{\lambda}_{\Lambda}\ket{0}_{M}\ket{n(\mu_{\mathrm{can}})}_{F}\ket{0}_{W}.
\label{eq:Cdagger_def_highlevel}
\end{equation}
Thus the overall transformation computes the occupation vector, reconstructs the canonical lower GT rows coherently from that occupation data, uses them together with the top row in $\Lambda$ to uncompute the GT register, and finally uncomputes the temporary reconstruction workspace.

Using these ingredients, the action of $U_{\mathrm{JS}}$ is
\begin{equation}
U_{\mathrm{JS}}
=
C^{\dagger}
\,S\,
C
\,A,
\label{eq:UJS_factorization_highlevel}
\end{equation}
where, for notational compactness, each unitary is understood to act nontrivially only on the
registers indicated in Eqs.~\eqref{eq:A_def_highlevel}--\eqref{eq:S_def_highlevel}.

\begin{algorithm}[H]
\caption{$U_{\mathrm{JS}}$ --- High-level reversible conversion of canonical GT patterns to occupation states}
\label{alg:U_JS_highlevel_clean}
\begin{algorithmic}[1]
\Require Top-row register $\Lambda$ holding the fixed Young-diagram label $\ket{\lambda}$
\Require Canonical GT-pattern register $M$ holding a coherent superposition
\[
\sum_{\mu_{\mathrm{can}}}\alpha_{\mu_{\mathrm{can}}}\ket{\mu_{\mathrm{can}}}
\qquad\text{with}\qquad
\mu_{\mathrm{can}}=(\lambda,\mu_{\mathrm{can},\downarrow})
\]
\Require Occupation register $F$ initialized to $\ket{0}$
\Require Work register $W$ initialized to $\ket{0}$ and sized only to store the lower GT rows
\Ensure
\[
\sum_{\mu_{\mathrm{can}}}\alpha_{\mu_{\mathrm{can}}}\ket{\lambda}_{\Lambda}\ket{\mu_{\mathrm{can}}}_{M}\ket{0}_{F}\ket{0}_{W}
\mapsto
\ket{\lambda}_{\Lambda}\ket{0}_{M}\sum_{\mu_{\mathrm{can}}}\alpha_{\mu_{\mathrm{can}}}\ket{n(\mu_{\mathrm{can}})}_{F}\ket{0}_{W}
\]
\State Apply the arithmetic unitary $A$ to compute occupations:
\[
\ket{\mu_{\mathrm{can}}}_{M}\ket{0}_{F}\ket{0}_{W}
\mapsto
\ket{\mu_{\mathrm{can}}}_{M}\ket{n(\mu_{\mathrm{can}})}_{F}\ket{0}_{W}
\]
\State Apply the canonical reconstruction unitary $C$ to write the lower rows of the canonical GT pattern into $W$:
\[
\ket{\lambda}_{\Lambda}\ket{\mu_{\mathrm{can}}}_{M}\ket{n}_{F}\ket{0}_{W}
\mapsto
\ket{\lambda}_{\Lambda}\ket{\mu_{\mathrm{can}}}_{M}\ket{n}_{F}\ket{\mu_{\mathrm{can},\downarrow}(n,\lambda)}_{W}
\]
\State Apply the subtraction unitary $S$ to subtract the reconstructed full canonical GT pattern, specified jointly by $\Lambda$ and $W$, from the GT register $M$:
\[
\ket{\lambda}_{\Lambda}\ket{\mu_{\mathrm{can}}}_{M}\ket{n}_{F}\ket{\mu_{\mathrm{can},\downarrow}(n,\lambda)}_{W}
\mapsto
\ket{\lambda}_{\Lambda}\ket{\mu_{\mathrm{can}}-(\lambda,\mu_{\mathrm{can},\downarrow}(n,\lambda))}_{M}\ket{n}_{F}\ket{\mu_{\mathrm{can},\downarrow}(n,\lambda)}_{W}
\]
\State On the promised canonical subspace, use $\mu_{\mathrm{can},\downarrow}(n(\mu_{\mathrm{can}}),\lambda)=\mu_{\mathrm{can},\downarrow}$ to obtain
\[
\ket{\lambda}_{\Lambda}\ket{0}_{M}\ket{n}_{F}\ket{\mu_{\mathrm{can},\downarrow}}_{W}
\]
\State Apply $C^\dagger$ to uncompute the work register:
\[
\ket{\lambda}_{\Lambda}\ket{0}_{M}\ket{n}_{F}\ket{\mu_{\mathrm{can},\downarrow}}_{W}
\mapsto
\ket{\lambda}_{\Lambda}\ket{0}_{M}\ket{n}_{F}\ket{0}_{W}
\]
\end{algorithmic}
\end{algorithm}

We note in passing that the coherent uncomputation of the Gelfand–Tsetlin register in $U_{\mathrm{JS}}$ can alternatively be carried out by measurement-based uncomputation (Hadamard-basis measurement plus a Clifford feed-forward correction), realizing the same unitary on the promised subspace at unchanged asymptotic cost and offering only (potential) constant-factor fault-tolerant trade-offs. Because this stage of the computation is not expected to be a bottleneck as compared to the preceding quantum Schur transformation, we do not describe this alternate approach in detail.

The remainder of this subsection is devoted to the explicit implementation of the arithmetic map $A$
and the canonical reconstruction map $C$. The former computes occupation numbers from GT row-sum
differences, while the latter reconstructs the lower rows of the unique canonical GT representative
associated with a given occupation vector and top-row label $\lambda$ on the promised subspace.

\subsubsection{Implementation of the arithmetic map $A$: occupations from GT row-sum differences \label{subsubsec:A_from_GT}}

We now describe the implementation of the arithmetic subroutine $A$ appearing in
Eq.~\eqref{eq:A_def_highlevel}. This map coherently computes the occupation vector from a GT pattern
while leaving the input GT register unchanged; the same arithmetic in fact applies to a general (not
necessarily canonical) GT pattern $\mu$, since the row-sum rule below uses no canonicity promise.
The arithmetic rule implemented by $A$ is obtained by comparing two ways of encoding the same Cartan eigenvalues. On the occupation-number side, Eq.~\eqref{eq:Cartan_equals_numberdiff} shows that the Cartan generators act as number-operator differences, and hence on an occupation basis state $\ket{n}$ the associated Dynkin weights are $w_\ell(n)=n_\ell-n_{\ell+1}$, as stated in Eq.~\eqref{eq:weight_from_occup}. On the first-quantized side, the Schur basis vectors $\ket{\lambda,\mu,\sigma}$ are simultaneous eigenvectors of the commuting Cartan subalgebra, and the corresponding eigenvalues can be read off from the GT pattern $\mu$ by a row-sum identity. Writing the GT pattern as the triangular array of integers $\mu=\{\mu^{(r)}_i\}$ introduced in Eq.~\eqref{eq:GT_interlacing}, define the row sums
\begin{equation}
R_r(\mu)\;:=\;\sum_{i=1}^{r}\mu^{(r)}_i,
\qquad r=1,\ldots,d,
\qquad
R_0(\mu):=0.
\label{eq:A_row_sums_def}
\end{equation}
In these conventions, the Cartan eigenvalues in Dynkin coordinates are
\begin{equation}
w_\ell(\mu)
\;=\;
2\,R_\ell(\mu)\;-\;R_{\ell+1}(\mu)\;-\;R_{\ell-1}(\mu),
\qquad \ell=1,\ldots,d-1.
\label{eq:A_dynkin_from_rowsums}
\end{equation}
It is convenient to express Eq.~\eqref{eq:A_dynkin_from_rowsums} in terms of adjacent row-sum differences defined as
\begin{equation}
m_\ell(\mu)\;:=\;R_\ell(\mu)-R_{\ell-1}(\mu),
\qquad \ell=1,\ldots,d,
\label{eq:A_m_def}
\end{equation}
in which case Eq.~\eqref{eq:A_dynkin_from_rowsums} becomes
\begin{equation}
w_\ell(\mu)\;=\;m_\ell(\mu)-m_{\ell+1}(\mu),
\qquad \ell=1,\ldots,d-1.
\label{eq:A_w_as_mdiff}
\end{equation}
Since $w_\ell(\mu)$ and $w_\ell(n)$ describe the same Cartan eigenvalues for the same basis vector viewed in two realizations, we equate Eqs.~\eqref{eq:weight_from_occup} and \eqref{eq:A_w_as_mdiff} and obtain
\begin{equation}
n_\ell(\mu)-n_{\ell+1}(\mu)\;=\;m_\ell(\mu)-m_{\ell+1}(\mu),
\qquad \ell=1,\ldots,d-1.
\label{eq:A_diff_equality}
\end{equation}
This implies the existence of a constant $C(\mu)$ such that $n_\ell(\mu)=m_\ell(\mu)+C(\mu)$ for all $\ell=1,\ldots,d$. The constant is fixed by particle-number conservation. On the fixed-$N$ sector, $\sum_{\ell=1}^d n_\ell(\mu)=N$ by definition of the occupation vector. On the GT side,
$\sum_{\ell=1}^d m_\ell(\mu)=R_d(\mu)-R_0(\mu)=R_d(\mu)$ by Eq.~\eqref{eq:A_m_def}. The top GT row is determined by $\lambda\vdash N$, and its row sum satisfies $R_d(\mu)=\sum_{i=1}^d \mu^{(d)}_i=\sum_{i=1}^d \lambda_i = N$, so $\sum_{\ell=1}^d m_\ell(\mu)=N$ as well. Therefore
\begin{equation}
\sum_{\ell=1}^d n_\ell(\mu)\;=\;\sum_{\ell=1}^d m_\ell(\mu)+d\,C(\mu)
\quad\Rightarrow\quad
N\;=\;N+d\,C(\mu)
\quad\Rightarrow\quad
C(\mu)=0,
\end{equation}
and we arrive at the explicit identification
\begin{equation}
n_\ell(\mu)\;=\;m_\ell(\mu)\;=\;R_\ell(\mu)-R_{\ell-1}(\mu),
\qquad \ell=1,\ldots,d.
\label{eq:A_occupations_from_rowdiff}
\end{equation}
Equation~\eqref{eq:A_occupations_from_rowdiff} is the deterministic classical relation implemented coherently by $A$. Simply, the successive row-sum differences of the GT pattern are precisely the occupation numbers of the corresponding number state in the second-quantized formalism. This relationship is enabled by the $U(d)$-equivariant bijection established in Theorem~\ref{thm:equivariant_bijection}, which itself is underpinned by the Jordan--Schwinger Lie algebra homomorphism established in Section~\ref{sec:js_map}.

A straightforward reversible implementation of Eq.~\eqref{eq:A_occupations_from_rowdiff} proceeds by first computing all row sums $R_\ell(\mu)$, storing them, and then computing each occupation $n_\ell(\mu)$ in the register $F_\ell$ by subtracting adjacent row sums. While the exact cost of performing simple arithmetic operations depends on the choice of GT-pattern encoding, in the following we assume that $b_\mu$ denotes the bit-width used to store each GT entry in a binary encoding. Since all GT entries lie in $\{0,1,\ldots,N\}$, a natural choice is $b_\mu=\lceil\log_2(N+1)\rceil$. The row sums $R_\ell(\mu)$ and occupations $n_\ell(\mu)$ also lie in $\{0,1,\ldots,N\}$, so we take the same unsigned width $b_R=b_n=\lceil\log_2(N+1)\rceil$ for their registers.

Computing all row sums from scratch requires
\begin{equation}
A_{\mathrm{sum}}=\sum_{\ell=1}^{d}(\ell-1)=\frac{d(d-1)}{2}
\label{eq:A_Asum}
\end{equation}
in-place additions on $b_\mu$-bit words, while forming the occupations from adjacent differences requires $d$ in-place subtractions at width $b_n$. If each add/sub is implemented using a Cuccaro ripple-carry adder \cite{Cuccaro2004}, an in-place $b$-bit operation uses $2b+\mathcal{O}(1)$ Toffoli gates, $5b+\mathcal{O}(1)$ CNOT gates, and depth $2b+\mathcal{O}(1)$, with a single carry ancilla reused across the computation. Consequently, a direct implementation of $A$ has Toffoli count
\begin{equation}
N_{\mathrm{CCX}}
=
\bigl(2b_\mu+\mathcal{O}(1)\bigr)\,A_{\mathrm{sum}}
+\bigl(2b_n+\mathcal{O}(1)\bigr)\,d
=
\mathcal{O}\!\bigl(d^2\log N\bigr),
\label{eq:A_cost_scaling}
\end{equation}
and similarly the CNOT count and circuit depth scale as $\mathcal{O}(d^2\log N)$.

\begin{algorithm}[H]
\caption{$A$ --- Reversible computation of occupations from GT patterns via row-sum differences}
\label{alg:A_from_rowsums}
\begin{algorithmic}[1]
\Require GT-pattern register $M$ holding a coherent superposition $|\psi_{\mu} \rangle := \sum_{\mu}\alpha_\mu\ket{\mu}$
\Require Occupation register $F$ initialized to $\ket{0}$
\Require Workspace $W$ for row sums and one reusable carry ancilla
\Require Known $(N,d)$ and a choice of binary encodings with widths $b_\mu$ (GT entries) and $b_n$ (occupations)
\Ensure Coherent map $\sum_{\mu}\alpha_\mu\ket{\mu}_M\ket{0}_F\ket{0}_W \mapsto \sum_{\mu}\alpha_\mu\ket{\mu}_M\ket{n(\mu)}_F\ket{0}_W$
\State Allocate row-sum workspace for $R_0,\ldots,R_d$ (each $b_R$ bits), with $R_0:=0$
\Statex
\State Compute row sums: for $\ell=1$ to $d$, coherently compute
\[
R_\ell \leftarrow \sum_{i=1}^{\ell}\mu^{(\ell)}_i
\]
using $(\ell-1)$ in-place ripple-carry additions
\Statex
\State Write occupations: for $\ell=1$ to $d$, subtract $R_{\ell-1}$ so that
\[
F_\ell \leftarrow R_\ell-R_{\ell-1}=n_\ell(\mu)
\]
\Statex
\State Uncompute the row-sum workspace $W$ to $|0 \rangle$
\Statex
\Ensure $F$ holds $n(\mu)$ as in Eq.~\eqref{eq:A_occupations_from_rowdiff}, while $M$ remains unchanged.
\end{algorithmic}
\end{algorithm}

\subsubsection{Implementation of the canonical reconstruction map $C$: lower canonical GT rows from occupation vectors \label{subsubsec:C_from_occup}}

We now describe the implementation of the reversible canonical reconstruction subroutine $C$
defined in Eq.~\eqref{eq:C_def_highlevel}. Recall that $C$ takes an occupation vector $n$ together
with the pre-existing top-row label $\lambda$ stored in the $\Lambda$ register, and writes into the
work register the \emph{lower rows} $\mu_{\mathrm{can},\downarrow}(n,\lambda)$ of the canonical GT
pattern $\mu_{\mathrm{can}}=(\lambda,\mu_{\mathrm{can},\downarrow})$ associated with that occupation
vector, where $\mu_{\mathrm{can},\downarrow}(n,\lambda)$ denotes the lower rows of a chosen canonical
representative among all GT patterns with top row $\lambda$ and weight $n$; the top row is not
rewritten, since it is already available in $\Lambda$. On the promised canonical subspace, the
canonical rule must be consistent with the promised support [Eq.~\eqref{eq:C_consistency_highlevel}].
This consistency condition is essential: the same canonical rule used to define the allowed input GT
patterns must also be the rule used by $C$ to reconstruct the lower GT rows from the occupation
vector. Otherwise, the subtraction step in $U_{\mathrm{JS}}$ would fail to erase the nontrivial part
of the input GT register.

For clarity, the role of $C$ is therefore not to recover an arbitrary GT pattern from $n$, which
would in general be impossible when multiple GT patterns share the same occupation vector, but rather
to reconstruct a \emph{distinguished canonical representative} within the $K_{\lambda,n}$-dimensional
weight subspace for Kostka number $K_{\lambda, n}$.  We now specify a concrete canonical choice and then give a reversible arithmetic
implementation.

\paragraph{Canonical choice.}
Fix the top GT row by the partition $\lambda=(\lambda_1,\ldots,\lambda_d)\vdash N$, and write the
occupation vector as $n=(n_1,\ldots,n_d)$. Define cumulative occupation sums
\begin{equation}
N_r(n)\;:=\;\sum_{\ell=1}^{r} n_\ell,
\qquad r=1,\ldots,d,
\qquad
N_0(n):=0.
\label{eq:C_cumulative_occ}
\end{equation}
A GT pattern $\mu_{\mathrm{can}}=\{\mu_{\mathrm{can},i}^{(r)}\}$ has weight $n$ if and only if the row sums satisfy
\begin{equation}
R_r(\mu_{\mathrm{can}})=\sum_{i=1}^{r}\mu_{\mathrm{can},i}^{(r)}=N_r(n),
\qquad r=1,\ldots,d.
\label{eq:C_rowsum_constraint}
\end{equation}
Since the top row is already fixed to $\lambda$, reconstructing the full canonical GT pattern reduces
to reconstructing the lower rows
\[
\mu_{\mathrm{can}}^{(d-1)},\mu_{\mathrm{can}}^{(d-2)},\ldots,\mu_{\mathrm{can}}^{(1)}.
\]
For a fixed row above,
\[
u=(u_1,\ldots,u_{r+1}),
\]
any admissible row below,
\[
x=(x_1,\ldots,x_r),
\]
must satisfy the interlacing inequalities
\begin{equation}
u_i \;\ge\; x_i \;\ge\; u_{i+1},
\qquad i=1,\ldots,r,
\label{eq:C_interlacing_constraint}
\end{equation}
together with the row-sum condition
\begin{equation}
\sum_{i=1}^{r} x_i = N_r(n).
\label{eq:C_rowsum_target}
\end{equation}

Among all rows obeying Eqs.~\eqref{eq:C_interlacing_constraint} and \eqref{eq:C_rowsum_target}, we
choose the \emph{lexicographically largest feasible row}. Applying this rule recursively from row
$r=d-1$ down to $r=1$ defines a unique canonical GT pattern with top row $\lambda$, or equivalently
a unique lower-row string $\mu_{\mathrm{can},\downarrow}(n,\lambda)$. In other words, for each $r$ we choose the
largest possible value of $x_1$, then subject to that the largest possible value of $x_2$, and so
on, while preserving the existence of a feasible completion of the row. This canonical section is
deterministic and depends only on $(\lambda,n)$.

\paragraph{Greedy row-construction rule.}
Suppose row $u=(u_1,\ldots,u_{r+1})$ is already known and we wish to construct the row
$x=(x_1,\ldots,x_r)$ below it. For each position $i=1,\ldots,r$, define the local interlacing
bounds
\begin{equation}
L_i := u_{i+1},
\qquad
U_i := u_i.
\label{eq:C_local_bounds}
\end{equation}
If the remaining target sum at step $i$ is denoted by $T_i$, then feasibility of the unfinished
suffix requires
\begin{equation}
\sum_{k=i+1}^{r} L_k
\;\le\;
T_i - x_i
\;\le\;
\sum_{k=i+1}^{r} U_k.
\label{eq:C_suffix_feasibility}
\end{equation}
Equivalently, the admissible interval for $x_i$ is
\begin{equation}
\max\!\Bigl(L_i,\; T_i-\sum_{k=i+1}^{r}U_k\Bigr)
\;\le\;
x_i
\;\le\;
\min\!\Bigl(U_i,\; T_i-\sum_{k=i+1}^{r}L_k\Bigr).
\label{eq:C_admissible_interval}
\end{equation}
The lexicographically largest feasible choice is therefore obtained by taking the upper endpoint
\begin{equation}
x_i
=
\min\!\Bigl(U_i,\; T_i-\sum_{k=i+1}^{r}L_k\Bigr),
\label{eq:C_greedy_choice}
\end{equation}
and then updating the remaining target
\begin{equation}
T_{i+1} = T_i - x_i.
\label{eq:C_target_update}
\end{equation}
Starting from $T_1=N_r(n)$, this produces the unique lexicographically largest admissible row.
Iterating from $r=d-1$ down to $1$ yields the full canonical GT pattern $\mu_{\mathrm{can}}(n)$.

This classical rule is deterministic, uses only integer additions, subtractions, comparisons, and
controlled assignments, and is therefore well suited to reversible implementation.

\paragraph{Register layout and arithmetic widths.}
We assume that the GT-pattern work register $W$ stores the triangular array
\[
\mu_{\mathrm{can}}(n)=\{\mu_{\mathrm{can},i}^{(r)}\}_{1\le i\le r\le d}
\]
using the same binary encoding and word width $b_\mu=\lceil\log_2(N+1)\rceil$ as for the input
GT-pattern register. Since all intermediate row entries and suffix sums lie in $\{0,1,\ldots,N\}$,
it suffices to use the same width $b_\mu$ for all arithmetic words. We also use a small amount of
temporary workspace for the remaining target $T_i$, the suffix sums
\[
\sum_{k=i+1}^{r}L_k,
\qquad
\sum_{k=i+1}^{r}U_k,
\]
and comparator/carry ancillas. These ancillas are assumed to be reused throughout the computation
and returned to $\ket{0}$ at the end of $C$.

An overview of the arithmetic implemented by $C$ is presented in Algorithm~\ref{alg:C_from_occup}.

\begin{algorithm}[H]
\caption{$C$ --- Reversible reconstruction of the lower rows of the canonical GT pattern from an occupation vector}
\label{alg:C_from_occup}
\begin{algorithmic}[1]
\Require Top-row register $\Lambda$ holding $\ket{\lambda}$, where $\lambda=(\lambda_1,\ldots,\lambda_d)$
\Require Input GT register $M$ holding $\ket{\mu_{\mathrm{can}}}$
\Require Occupation register $F$ holding a coherent superposition $\sum_n \alpha_n \ket{n}$
\Require Work register $W$ initialized to $\ket{0}$ for storing the lower GT rows
\Ensure Coherent map
\[
\sum_n \alpha_n \ket{\lambda}_{\Lambda}\ket{\mu_{\mathrm{can}}}_{M}\ket{n}_{F}\ket{0}_{W}
\mapsto
\sum_n \alpha_n \ket{\lambda}_{\Lambda}\ket{\mu_{\mathrm{can}}}_{M}\ket{n}_{F}\ket{\mu_{\mathrm{can},\downarrow}(n,\lambda)}_{W}
\]
\State Compute cumulative occupation sums
\[
N_r(n)=\sum_{\ell=1}^{r}n_\ell,\qquad r=1,\ldots,d-1
\]
\For{$r=d-1$ down to $1$}
    \State Let $u=(u_1,\ldots,u_{r+1})$ be the row above:
    \Statex \hspace{\algorithmicindent} if $r=d-1$, take $u=\lambda$ from the $\Lambda$ register
    \Statex \hspace{\algorithmicindent} otherwise, take $u$ to be row $r+1$ already stored in $W$
    \State Set remaining target $T \leftarrow N_r(n)$
    \For{$i=1$ to $r$}
        \State Set $L_i \leftarrow u_{i+1}$ and $U_i \leftarrow u_i$
        \State Compute the suffix lower sum $S_L \leftarrow \sum_{k=i+1}^{r} L_k$
        \State Compute the suffix upper sum $S_U \leftarrow \sum_{k=i+1}^{r} U_k$
        \State Compute the admissible interval
        \[
        \max(L_i,T-S_U)\ \le\ x_i\ \le\ \min(U_i,T-S_L)
        \]
        \State Choose the lexicographically largest feasible value
        \[
        x_i \leftarrow \min(U_i,T-S_L)
        \]
        \State Write $x_i$ into position $i$ of row $r$ in $W$
        \State Update the remaining target $T \leftarrow T-x_i$
    \EndFor
\EndFor
\State Uncompute all temporary arithmetic workspace except the lower-row GT output in $W$
\end{algorithmic}
\end{algorithm}

\paragraph{Asymptotic cost.}
There are
\[
\sum_{r=1}^{d-1} r = \frac{d(d-1)}{2}
\]
GT entries to be written below the top row. For each such entry $x_i$, a direct implementation of
Eq.~\eqref{eq:C_greedy_choice} requires a constant number of $b_\mu$-bit additions/subtractions,
one $b_\mu$-bit comparison, one controlled assignment, and one update of the remaining target. Using
Cuccaro ripple-carry arithmetic \cite{Cuccaro2004} for the additions and subtractions, each such
operation costs $2b_\mu+\mathcal{O}(1)$ Toffoli gates and has depth
$2b_\mu+\mathcal{O}(1)$. A reversible comparator can be implemented with the same asymptotic cost,
so each GT entry contributes $\mathcal{O}(b_\mu)$ Toffoli gates and depth.

Consequently, a direct implementation of $C$ has Toffoli count
\begin{equation}
N_{\mathrm{CCX}}(C)
=
\mathcal{O}\!\left(\sum_{r=1}^{d-1} r\, b_\mu\right)
=
\mathcal{O}\!\bigl(d^2 \log N\bigr),
\label{eq:C_cost_scaling}
\end{equation}
and similarly the CNOT count and depth scale as $\mathcal{O}(d^2\log N)$. The same asymptotic
estimate applies to $C^\dagger$, since it is implemented by reversing the arithmetic steps of $C$.
Therefore, $C$ has the same asymptotic arithmetic cost as the forward occupation-computation
subroutine $A$, and therefore the full Jordan--Schwinger unitary $U_{\mathrm{JS}}$ remains
polynomial in both $d$ and $\log N$.

In practical implementations, some of the suffix sums in Algorithm~\ref{alg:C_from_occup} may be
streamed or updated incrementally rather than recomputed from scratch at every step, reducing
constant factors. Likewise, comparator ancillas and target registers can be aggressively recycled.
For the purposes of the present high-level resource analysis, however, the direct scaling estimate in
Eq.~\eqref{eq:C_cost_scaling} suffices.

\subsection{Total gate complexity \label{subsec:total_gate_complexity}}

We now determine the total gate complexity for the application of $Q$ or $Q^\dagger$ in Theorem~\ref{thm:QQT_complexity_poly}.

\begin{boxedtheorem}[Gate complexity of quantum first--second quantization conversion]
\label{thm:QQT_complexity_poly}
Fix particle number $N$, single-particle dimension $d$, and target diamond-norm (or operator-norm) accuracy
$\epsilon\in(0,1)$. Assume the input state is supported on a known, fixed Schur--Weyl sector $\lambda\vdash N$
(with $\ell(\lambda)\le d$), and that within this sector the GT register is supported on the promised
canonical representatives $\mu_{\mathrm{can}}$. Also fix the $\Sigma$ register to a single canonical label $\sigma_{\text{can}}$. Then the transform
\begin{equation}
Q:\ \mathcal{H}_{\lambda}\to \mathcal{F}^{(\lambda)}_{N}
\end{equation}
and its inverse $Q^\dagger$ can be implemented on a gate-based quantum computer with total gate complexity
\begin{equation}
\mathrm{GateCost}(Q),\ \mathrm{GateCost}(Q^\dagger)
\;=\;
\mathrm{poly}\!\bigl(N,\,d,\,\log(\epsilon^{-1})\bigr),
\label{eq:QQT_complexity_poly}
\end{equation}
where the polynomial depends on the chosen strong Schur transform implementation (BCH, KB or otherwise)
and on the reversible-arithmetic constructions used in the Jordan--Schwinger stage.
\end{boxedtheorem}

\begin{proof}
By Algorithm~\ref{alg:QQT_highlevel2}, $Q$ is realized as the composition
\[
Q = U_{\mathrm{JS}}\,U_{\mathrm{Schur}},
\qquad
Q^\dagger = U_{\mathrm{Schur}}^\dagger\,U_{\mathrm{JS}}^\dagger.
\]
A strong Schur transform $U_{\mathrm{Schur}}$ with error at most $\epsilon/2$ can be implemented with
$\mathrm{poly}(N,d,\log(\epsilon^{-1}))$ gates using the BCH construction \cite{Bacon2006Schur}, or with
$\mathrm{poly}(N,\log d,\log(\epsilon^{-1}))$ gates using the high-dimensional KB construction
\cite{Krovi2019Schur,Burchardt2025Schur}; see Algorithm~\ref{alg:U_schur_placeholder}.

It therefore remains to bound the arithmetic cost of $U_{\mathrm{JS}}$. As shown in
Eq.~\eqref{eq:UJS_factorization_highlevel}, $U_{\mathrm{JS}} = C^\dagger S C A$, acting on the promised
canonical GT subspace labeled by $\mu_{\mathrm{can}}$: $A$ computes the occupation vector
$n(\mu_{\mathrm{can}})$ from GT row-sum differences, $C$ reconstructs the lower rows
$\mu_{\mathrm{can},\downarrow}(n,\lambda)$ of the canonical GT representative from the pair
$(\lambda,n)$, $S$ subtracts the reconstructed canonical pattern from the GT register, and
$C^\dagger$ uncomputes the reconstruction workspace.

The arithmetic subroutine $A$ has Toffoli/CNOT count and depth $\mathcal{O}(d^2\log N)$
[Eq.~\eqref{eq:A_cost_scaling}], and the canonical reconstruction subroutine $C$ has the same
asymptotic cost, $N_{\mathrm{CCX}}(C)=\mathcal{O}(d^2\log N)$ [Eq.~\eqref{eq:C_cost_scaling}], with the
same scaling for CNOT count and depth. The inverse $C^\dagger$ matches $C$, since it reverses the same
arithmetic steps, and the subtraction unitary $S$ is another reversible arithmetic routine on the same
triangular GT data with $b_\mu$-bit words, contributing no worse asymptotically than the other
components. Consequently,
\begin{equation}
\mathrm{GateCost}(U_{\mathrm{JS}})
=
\mathrm{GateCost}(A)+\mathrm{GateCost}(C)+\mathrm{GateCost}(S)+\mathrm{GateCost}(C^\dagger)
=
\mathcal{O}(d^2\log N),
\end{equation}
up to implementation-dependent constant factors and lower-order terms.

Combining the Schur-transform cost with the arithmetic cost of $U_{\mathrm{JS}}$, and distributing the
allowed approximation error so that the total implementation error is at most $\epsilon$, yields the
claimed overall scaling in Eq.~\eqref{eq:QQT_complexity_poly}. The same bound holds for $Q^\dagger$,
since it is implemented by the inverses of the same two polynomial-size circuits.
\end{proof}

\section{Worked example with intermediate quantum states}
\label{sec:Q_worked_example_N2d3}

In this section we illustrate, end-to-end and with explicit intermediate states, how the transform $Q$ converts a first-quantized $N$-particle input into its fixed-$N$ second-quantized (occupation-number) representation. We deliberately choose a small but nontrivial didactic instance of two identical fermions ($N=2$) in three single-particle modes ($d=3$) so that every label can be written explicitly, and we present naive qubit encodings throughout.

\subsection{Register conventions and naive qubit encodings}
\label{subsec:Q_example_registers}

We label the one-particle basis as $\{\ket{1},\ket{2},\ket{3}\}$. On qubit hardware, each
single-particle label is stored using
\begin{equation}
q := \left\lceil \log_2 d\right\rceil = \left\lceil \log_2 3\right\rceil = 2
\label{eq:q_single_particle_bits_updated}
\end{equation}
qubits. We adopt the explicit binary encoding
\begin{equation}
\ket{1}\mapsto \ket{00},\qquad
\ket{2}\mapsto \ket{01},\qquad
\ket{3}\mapsto \ket{10},
\label{eq:mode_binary_encoding_updated}
\end{equation}
leaving $\ket{11}$ unused. The first-quantized Hilbert space $(\mathbb{C}^3)^{\otimes 2}$ is thus
encoded on
\begin{equation}
n_{\mathrm{in}} = N q = 2\times 2 = 4
\end{equation}
qubits, with computational basis states $\ket{i}\ket{j}$ represented as
$\ket{\mathrm{bin}(i)}\ket{\mathrm{bin}(j)}$.

For the output occupation-number representation, we allocate a uniform register that can represent
bosonic occupations up to $N$ in each mode. This is required for the general transform $Q$, which is
agnostic to the input wavefunction symmetry: bosonic occupations provide an upper bound on the
occupancies permitted by any particle statistics sector. Thus each mode occupation
$n_\ell\in\{0,1,2\}$ is encoded using
\begin{equation}
b := \left\lceil \log_2(N+1)\right\rceil = \left\lceil \log_2 3\right\rceil = 2
\label{eq:b_occ_bits_updated}
\end{equation}
qubits, and the full occupation register uses $d b = 3\times 2 = 6$ qubits. For definiteness we
encode
\begin{equation}
0\mapsto \ket{00},\qquad 1\mapsto \ket{01},\qquad 2\mapsto \ket{10},
\label{eq:occ_binary_encoding_updated}
\end{equation}
again leaving $\ket{11}$ unused. Even though the present example is fermionic (so $n_\ell\in\{0,1\}$),
we retain this $2$-qubit-per-mode layout to emphasize that $Q$ is agnostic to the input symmetry.

The quantum Schur transform stage outputs registers for the Young diagram label $\Lambda$, the GT label $M$, and the Young--Yamanouchi label $\Sigma$. In the present example, the fermionic symmetry sector fixes $\lambda=(1,1,0)$ and there is only one allowed Young--Yamanouchi label, so $\Lambda$ and $\Sigma$ are constant. Additionally, because we are dealing with fermions, the canonical GT promise is automatically satisfied for all inputs, so below $\mu$ is interchangeable with $\mu_{\text{can}}$. Following the definition of $U_{\mathrm{JS}}$ in Sec.~\ref{subsec:JS_unitary_from_GT}, we also use an occupation register $F$ and a work register $W$ that stores only the \emph{lower rows} of the canonical GT pattern. Since $d=3$ here, the lower rows consist of the row of length $2$ and the row of length $1$, so under the same naive $2$-bit-per-entry encoding the work register uses $(2+1)\times 2=6$ qubits.

\subsection{Step 1: a first-quantized two-fermion input state (and its qubit encoding)}
\label{subsec:Q_example_step1_updated}

The antisymmetric (fermionic) two-particle basis vectors are the normalized Slater determinants
\begin{equation}
\ket{i\wedge j}
\;:=\;
\frac{1}{\sqrt{2}}
\bigl(\ket{i}\ket{j}-\ket{j}\ket{i}\bigr),
\qquad 1\le i<j\le 3.
\label{eq:wedge_def_updated}
\end{equation}
Suppose the input first-quantized register is prepared in the normalized state
\begin{equation}
\ket{\psi_{\mathrm{1Q}}}
=
\frac{1}{\sqrt{2}}\Bigl(\ket{1\wedge 2}+\ket{1\wedge 3}\Bigr)
\;\in\;
\wedge^2(\mathbb{C}^3)\subset (\mathbb{C}^3)^{\otimes 2}.
\label{eq:psi_1Q_example_updated}
\end{equation}
Expanding in the tensor-product basis and then applying the binary encoding
of Eq.~\eqref{eq:mode_binary_encoding_updated}, this state is represented on $4$ qubits as
\begin{align}
\ket{\psi_{\text{1Q}}}
&=
\frac{1}{\sqrt{2}}
\left[
\frac{1}{\sqrt{2}}\bigl(\ket{1}\ket{2}-\ket{2}\ket{1}\bigr)
+
\frac{1}{\sqrt{2}}\bigl(\ket{1}\ket{3}-\ket{3}\ket{1}\bigr)
\right]\nonumber\\
&=
\frac{1}{2}\Bigl(
\ket{1}\ket{2}-\ket{2}\ket{1}+\ket{1}\ket{3}-\ket{3}\ket{1}
\Bigr)\nonumber\\
&\mapsto
\frac{1}{2}\Bigl(
\ket{00}\ket{01}-\ket{01}\ket{00}+\ket{00}\ket{10}-\ket{10}\ket{00}
\Bigr).
\label{eq:psi_input_binary_updated}
\end{align}
This is the explicit qubit-level first-quantized wavefunction to which the circuit implementation
of $Q$ is applied.

\subsection{Step 2: apply the Schur transform $U_{\mathrm{Schur}}$}
\label{subsec:Q_example_step2_updated}

For $(\mathbb{C}^3)^{\otimes 2}$, Schur--Weyl duality decomposes the Hilbert space into the
$\lambda=(2,0,0)$ symmetric sector and the $\lambda=(1,1,0)$ antisymmetric sector. Since
$\ket{\psi_{\mathrm{1Q}}}$ is fermionic, it is supported entirely on the latter Young diagram,
\begin{equation}
\lambda=(1,1,0),
\end{equation}
and hence carries the one-dimensional sign irrep of $S_2$. Consequently there is a unique
$S_2$ basis label: the unique standard Young tableau of shape $(1,1)$ is
\[
\ytableaushort{1,2},
\qquad\text{with Young--Yamanouchi word}\qquad \sigma=(1,2).
\]

After applying the Schur transform, the state is decomposed into Schur-basis labels
$\ket{\lambda,\mu,\sigma}$. In this $(N,d)=(2,3)$ example, the top row of every GT pattern is fixed to
\[
\mu^{(3)}=\lambda=(1,1,0),
\]
and the admissible interlacing patterns are
\begin{equation}
\mu_A=
\left(
\begin{array}{ccc}
1 & 1 & 0\\
1 & 1\\
1
\end{array}
\right),
\qquad
\mu_B=
\left(
\begin{array}{ccc}
1 & 1 & 0\\
1 & 0\\
1
\end{array}
\right),
\qquad
\mu_C=
\left(
\begin{array}{ccc}
1 & 1 & 0\\
1 & 0\\
0
\end{array}
\right).
\label{eq:GT_patterns_wedge2_d3_updated}
\end{equation}
These are exactly the GT patterns satisfying the interlacing conditions, which evaluate to
\[
1\ge \mu^{(2)}_1 \ge 1,\qquad
1\ge \mu^{(2)}_2 \ge 0,\qquad
\mu^{(2)}_1\ge \mu^{(1)}_1 \ge \mu^{(2)}_2.
\]
Thus $\mu^{(2)}_1=1$ and $\mu^{(2)}_2\in\{0,1\}$. If $\mu^{(2)}_2=1$, then $\mu^{(1)}_1$ is forced to
equal $1$, giving $\mu_A$; if $\mu^{(2)}_2=0$, then $\mu^{(1)}_1$ may be either $1$ or $0$, giving
$\mu_B$ and $\mu_C$ respectively. This produces precisely the three patterns in
Eq.~\eqref{eq:GT_patterns_wedge2_d3_updated}.

Omitting the ancillary register $a$, the state of Eq.~\eqref{eq:psi_1Q_example_updated} is mapped to the intermediate Schur-basis wavefunction
\begin{equation}
U_{\mathrm{Schur}}\ket{\psi_{\mathrm{1Q}}}
\;=\;
\ket{\lambda=(1,1,0)}_{\Lambda}
\otimes
\frac{1}{\sqrt{2}}
\Bigl(\ket{\mu_A}_{M}+\ket{\mu_B}_{M}\Bigr)
\otimes
\ket{\sigma=(1,2)}_{\Sigma},
\label{eq:psi_after_schur_updated}
\end{equation}
with no amplitude on $\mu_C$ because the input has no $\ket{2\wedge 3}$ component.

In this fermionic example, the map from GT pattern to occupation vector is already injective on the
allowed support, so each admissible GT pattern is automatically its own canonical representative. Hence
\[
\mu_A,\ \mu_B,\ \mu_C
\]
are canonical GT patterns in the sense required by the Jordan--Schwinger stage.

On qubit hardware, a convenient naive layout is to store every GT entry row by row:
\begin{equation}
\ket{\mu}
\;\equiv\;
\ket{\mu^{(3)}_1,\mu^{(3)}_2,\mu^{(3)}_3\;;\;\mu^{(2)}_1,\mu^{(2)}_2\;;\;\mu^{(1)}_1},
\label{eq:mu_register_full_updated}
\end{equation}
allocating $b=2$ qubits per entry. Thus the GT register uses $(3+2+1)\times 2=12$ qubits in this
naive implementation. Using the encoding of Eq.~\eqref{eq:occ_binary_encoding_updated} entrywise, the two
patterns appearing in Eq.~\eqref{eq:psi_after_schur_updated} are encoded as
\begin{align}
\ket{\mu_A}
&=
\ket{\,1,1,0\;;\;1,1\;;\;1\,}
\ \mapsto\
\ket{01,01,00\;;\;01,01\;;\;01},
\label{eq:muA_full_binary_updated}\\
\ket{\mu_B}
&=
\ket{\,1,1,0\;;\;1,0\;;\;1\,}
\ \mapsto\
\ket{01,01,00\;;\;01,00\;;\;01}.
\label{eq:muB_full_binary_updated}
\end{align}

\subsection{Step 3: apply the Jordan--Schwinger unitary $U_{\mathrm{JS}}$}
\label{subsec:Q_example_step3_updated}

We now apply the Jordan--Schwinger arithmetic unitary,
\[
U_{\mathrm{JS}}:\ \ket{\lambda}_{\Lambda}\ket{\mu}_{M}\ket{0}_{F}\ket{0}_{W}
\mapsto
\ket{\lambda}_{\Lambda}\ket{0}_{M}\ket{n(\mu)}_{F}\ket{0}_{W},
\]
where we recall that the work register $W$ need only ever contain lower-row GT data $\mu_\downarrow$. $U_{\text{JS}}$ decomposes into the four reversible substeps
\[
A,\qquad C,\qquad S,\qquad C^\dagger.
\]

In the present example, the admissible lower-row data are
\begin{align}
\mu_{A,\downarrow}
&=
\left(
\begin{array}{cc}
1 & 1\\
1
\end{array}
\right),
&
\mu_{B,\downarrow}
&=
\left(
\begin{array}{cc}
1 & 0\\
1
\end{array}
\right),
&
\mu_{C,\downarrow}
&=
\left(
\begin{array}{cc}
1 & 0\\
0
\end{array}
\right).
\label{eq:lower_rows_example}
\end{align}

\paragraph{Step 3a: compute occupations with $A$.}
The arithmetic map $A$ computes occupations from GT row-sum differences. Define row sums
\begin{equation}
R_r(\mu)\;:=\;\sum_{i=1}^{r}\mu^{(r)}_i,
\qquad r=1,2,3,
\qquad
R_0(\mu):=0,
\label{eq:rowsums_def_updated}
\end{equation}
and occupations
\begin{equation}
n_\ell(\mu)\;:=\;R_\ell(\mu)-R_{\ell-1}(\mu),
\qquad \ell=1,2,3.
\label{eq:rowsum_difference_def_updated}
\end{equation}
In our small example, $R_1$ is the sum of the bottom row, $R_2$ is the sum of the middle row, and $R_3$ is the sum of the top row.

For the three admissible GT patterns,
\begin{align}
\mu_A:\ &R_1=1,\ R_2=2,\ R_3=2
\quad\Rightarrow\quad
n(\mu_A)=(1,1,0), \label{eq:occ_muA_updated}\\
\mu_B:\ &R_1=1,\ R_2=1,\ R_3=2
\quad\Rightarrow\quad
n(\mu_B)=(1,0,1), \label{eq:occ_muB_updated}\\
\mu_C:\ &R_1=0,\ R_2=1,\ R_3=2
\quad\Rightarrow\quad
n(\mu_C)=(0,1,1). \label{eq:occ_muC_updated}
\end{align}
Thus, starting from the post-Schur state with fresh occupation and work registers,
\begin{equation}
\ket{\lambda}_{\Lambda}
\otimes
\frac{1}{\sqrt{2}}
\Bigl(\ket{\mu_A}_{M}+\ket{\mu_B}_{M}\Bigr)
\otimes
\ket{0}_{F}\ket{0}_{W},
\end{equation}
the action of $A$ is
\begin{align}
\ket{\psi^{(A)}}
&:=
A
\left[
\ket{\lambda}_{\Lambda}
\otimes
\frac{1}{\sqrt{2}}
\Bigl(\ket{\mu_A}_{M}+\ket{\mu_B}_{M}\Bigr)
\otimes
\ket{0}_{F}\ket{0}_{W}
\right]\nonumber\\
&=
\ket{\lambda}_{\Lambda}\otimes
\frac{1}{\sqrt{2}}
\Bigl(
\ket{\mu_A}_{M}\ket{1,1,0}_{F}
+
\ket{\mu_B}_{M}\ket{1,0,1}_{F}
\Bigr)\ket{0}_{W}.
\label{eq:psi_after_A_example}
\end{align}

\paragraph{Step 3b: reconstruct the lower rows of the canonical GT pattern with $C$.}
Next we apply the canonical reconstruction map $C$, which uses the top-row register $\Lambda$
together with the occupation vector to write only the lower GT rows into the work register:
\[
C:\ \ket{\lambda}_{\Lambda}\ket{\mu}_{M}\ket{n}_{F}\ket{0}_{W}
\mapsto
\ket{\lambda}_{\Lambda}\ket{\mu}_{M}\ket{n}_{F}\ket{\mu_{\downarrow}(n,\lambda)}_{W}.
\]
In the present fermionic example the canonical GT representatives are simply the unique GT patterns
associated with each occupation vector.

Applying $C$ to Eq.~\eqref{eq:psi_after_A_example} gives
\begin{align}
\ket{\psi^{(C)}}
&:=
C\,\ket{\psi^{(A)}}\nonumber\\
&=
\ket{\lambda}_{\Lambda}\otimes
\frac{1}{\sqrt{2}}
\Bigl(
\ket{\mu_A}_{M}\ket{1,1,0}_{F}\ket{\mu_{A,\downarrow}}_{W}
+
\ket{\mu_B}_{M}\ket{1,0,1}_{F}\ket{\mu_{B,\downarrow}}_{W}
\Bigr).
\label{eq:psi_after_C_example}
\end{align}

\paragraph{Step 3c: subtract the reconstructed full GT pattern with $S$.}
The subtraction unitary $S$ uses the top row from $\Lambda$ and the reconstructed lower rows in $W$
to subtract the full canonical GT pattern from the GT register $M$:
\[
S:\ \ket{\lambda}_{\Lambda}\ket{\mu}_{M}\ket{n}_{F}\ket{\mu_{\downarrow}}_{W}
\mapsto
\ket{\lambda}_{\Lambda}\ket{\mu-(\lambda,\mu_{\downarrow})}_{M}\ket{n}_{F}\ket{\mu_{\downarrow}}_{W}.
\]
Since the work register now contains the same canonical lower-row data as the GT register on each branch,
we obtain
\begin{align}
\ket{\psi^{(S)}}
&:=
S\,\ket{\psi^{(C)}}\nonumber\\
&=
\ket{\lambda}_{\Lambda}\otimes
\frac{1}{\sqrt{2}}
\Bigl(
\ket{0}_{M}\ket{1,1,0}_{F}\ket{\mu_{A,\downarrow}}_{W}
+
\ket{0}_{M}\ket{1,0,1}_{F}\ket{\mu_{B,\downarrow}}_{W}
\Bigr)\nonumber\\
&=
\ket{\lambda}_{\Lambda}\ket{0}_{M}\otimes
\frac{1}{\sqrt{2}}
\Bigl(
\ket{1,1,0}_{F}\ket{\mu_{A,\downarrow}}_{W}
+
\ket{1,0,1}_{F}\ket{\mu_{B,\downarrow}}_{W}
\Bigr).
\label{eq:psi_after_S_example}
\end{align}

\paragraph{Step 3d: uncompute the work register with $C^\dagger$.}
Finally, because the work register was computed reversibly from the occupation register together with
the fixed top row $\lambda$, we may apply $C^\dagger$:
\[
\ket{\lambda}_{\Lambda}\ket{0}_{M}\ket{n}_{F}\ket{\mu_{\downarrow}(n,\lambda)}_{W}
\mapsto
\ket{\lambda}_{\Lambda}\ket{0}_{M}\ket{n}_{F}\ket{0}_{W}.
\]
Hence
\begin{align}
\ket{\psi^{(U_{\mathrm{JS}})}}
&:=
C^\dagger\,\ket{\psi^{(S)}}\nonumber\\
&=
\ket{\lambda}_{\Lambda}\otimes
\frac{1}{\sqrt{2}}
\Bigl(
\ket{0}_{M}\ket{1,1,0}_{F}\ket{0}_{W}
+
\ket{0}_{M}\ket{1,0,1}_{F}\ket{0}_{W}
\Bigr)\nonumber\\
&=
\ket{\lambda}_{\Lambda}\ket{0}_{M}\otimes
\frac{1}{\sqrt{2}}
\Bigl(
\ket{1,1,0}_{F}+\ket{1,0,1}_{F}
\Bigr)\otimes
\ket{0}_{W}.
\label{eq:psi_after_UJS_example}
\end{align}

After the full action of $U_{\mathrm{JS}}$, the GT register has been coherently erased, the work
register has been returned to $\ket{0}$, and the occupation register is left in a pure state. This is
exactly the behavior described in Eq.~\eqref{eq:U_JS_new_def}.

\subsection{Step 4: the second-quantized wavefunction (occupation-number output)}
\label{subsec:Q_example_step4_updated}

Equation~\eqref{eq:psi_after_UJS_example} shows that after the Schur-transform labels are fixed to the
fermionic sector and the Jordan--Schwinger stage has been completed, the active output state is simply
\begin{equation}
\ket{\psi_{\mathrm{2Q}}}
\;=\;
\frac{1}{\sqrt{2}}
\Bigl(\ket{1,1,0}+\ket{1,0,1}\Bigr).
\label{eq:psi_2Q_example_updated}
\end{equation}
The GT register $M$ and work register $W$ have both been returned to $\ket{0}$, while the fixed
$\lambda$ and $\sigma$ registers factor out and play no further dynamical role. Thus the occupation
register now behaves as the genuine second-quantized wavefunction.

In the symmetry-agnostic binary encoding with $2$ qubits per mode, the same output state is
\begin{equation}
\ket{\psi}_{\mathrm{out}}
\;=\;
\frac{1}{\sqrt{2}}
\Bigl(\ket{01,01,00}+\ket{01,00,01}\Bigr),
\label{eq:psi_2Q_example_binary_updated}
\end{equation}
where the occupation register is $6$ qubits and each $2$-qubit block encodes an integer occupation in
$\{0,1,2\}$.

\section{Quantum speed-ups \label{sec:speed_up}}

\subsection{Circuit-specified conversion with output in classical \textit{or} quantum memory}
\label{subsec:speed_up_either_memory}

Our unitary transformation $Q$ (Sec.~\ref{sec:qqt}) is a coherent change of representation between a fixed Schur--Weyl symmetry sector in first quantization and the corresponding fixed-$N$ symmetry sector in Fock space. In this subsection we formalize an operational task that keeps the \emph{input} identical for classical and quantum algorithms, while permitting the \emph{output} to be stored either as a quantum state (quantum memory) or as an explicit coefficient list (classical memory). This task quantifies the difficulty of converting between first- and second-quantized representations on classical versus quantum hardware.

For notational clarity, throughout this subsection (and all subsections within this Section~\ref{sec:speed_up}) we suppress explicit ancillary registers and intermediate workspace registers, and we write simply
\begin{equation}
Q:\ \ket{\psi_{\mathrm{1Q}}}\ \longmapsto\ \ket{\psi_{\mathrm{2Q}}},
\qquad
Q^\dagger:\ \ket{\psi_{\mathrm{2Q}}}\ \longmapsto\ \ket{\psi_{\mathrm{1Q}}},
\label{eq:Q_simplified_notation_flexible_output}
\end{equation}
with the understanding that $Q$ and $Q^\dagger$ are implemented by the construction of Sec.~\ref{sec:qqt}. In particular, whenever the Schur transform introduces GT-pattern data, the implementation is restricted to the promised canonical GT representatives $\mu_{\mathrm{can}}$, and whenever a multiplicity label is fixed we write it as $\sigma_{\mathrm{can}}$ to emphasize that a fixed canonical copy has been chosen.

\begin{boxedtheorem}[Task (Circuit-specified promised-$(\lambda,\sigma_{\mathrm{can}})$ first-to-second quantization with flexible output memory)]
\label{task:flexible_output_quantization}
Fix particle number $N$, single-particle dimension $d$, a Young diagram $\lambda\vdash N$ with $\ell(\lambda)\le d$, a fixed canonical multiplicity label $\sigma_{\mathrm{can}}$ of shape $\lambda$, and an accuracy parameter $\epsilon\in(0,1)$. You are given a classical description of a $\mathrm{poly}(N,d)$-depth quantum circuit implementing a unitary $C^{\lambda\sigma_{\mathrm{can}}}$ acting on
\[
n_{\mathrm{1Q}} = N\lceil \log_2 d\rceil
\]
qubits that encode $(\mathbb{C}^d)^{\otimes N}$ in the standard computational basis. Let
\begin{equation}
\ket{\psi^{\lambda\sigma_{\mathrm{can}}}}_{\mathrm{1Q}} \ :=\ C^{\lambda \sigma_{\mathrm{can}}}\ket{0^{n_{\mathrm{1Q}}}}.
\label{eq:flexible_input_state}
\end{equation}
\emph{Promise:} $\ket{\psi^{\lambda\sigma_{\mathrm{can}}}}_{\mathrm{1Q}}$ lies in the fixed Schur--Weyl sector
$\mathcal{H}_{\lambda\sigma_{\mathrm{can}}}\subset(\mathbb{C}^d)^{\otimes N}$, where $\sigma_{\mathrm{can}}$
denotes the chosen canonical multiplicity label (for example, a fixed canonical Young--Yamanouchi word)
for the $S_N$ copy associated with $\lambda$. Moreover, in the implementation of $Q$ the intermediate
GT-pattern support is promised to lie on the canonical GT representatives $\mu_{\mathrm{can}}$ of
Sec.~\ref{sec:qqt}.

The goal is to output the corresponding second-quantized state
\begin{equation}
\ket{\psi^{\lambda}}_{\mathrm{2Q}} \ :=\ Q\,\ket{\psi^{\lambda \sigma_{\mathrm{can}}}}_{\mathrm{1Q}}
\ \in\ \mathcal{F}^{(\lambda)}_{N},
\label{eq:flexible_target_state}
\end{equation}
to trace-distance error at most $\epsilon$, with the output stored in \emph{either}:
\begin{enumerate}
\item \emph{quantum memory:} prepare a quantum register in a state $\widetilde{\rho}$ such that
$\tfrac12\|\widetilde{\rho}-\proj{\psi^{\lambda}_{\mathrm{2Q}}}{\psi^{\lambda}_{\mathrm{2Q}}}\|_1\le \epsilon$; or
\item \emph{classical memory:} output an explicit list of coefficient--basis-string pairs
$\{(n,c_n)\}$ specifying a normalized vector $\ket{\widetilde{\psi}}=\sum_n c_n\ket{n}$ satisfying
$\tfrac12\|\proj{\widetilde{\psi}}{\widetilde{\psi}}-\proj{\psi^{\lambda}_{\mathrm{2Q}}}{\psi^{\lambda}_{\mathrm{2Q}}}\|_1\le \epsilon$.
\end{enumerate}
\end{boxedtheorem}

\subsubsection{Quantum runtime (quantum-memory output)}

The quantum algorithm is immediate: prepare $\ket{\psi^{\lambda \sigma_{\mathrm{can}}}}_{\mathrm{1Q}}$ by applying
$C^{\lambda \sigma_{\mathrm{can}}}$ (Eq.~\eqref{eq:flexible_input_state}), then apply $Q$ from Sec.~\ref{sec:qqt}
to obtain $\ket{\psi^\lambda}_{\mathrm{2Q}}$ (Eq.~\eqref{eq:flexible_target_state}). The next statement is a direct
specialization of Theorem~\ref{thm:QQT_complexity_poly}.

\begin{boxedtheorem}[Quantum complexity]
\label{thm:flexible_quantum_complexity}
Under the promises of Task~\ref{task:flexible_output_quantization}, including the canonical GT promise
on the intermediate labels $\mu_{\mathrm{can}}$, a gate-based quantum computer can output
$\ket{\psi^{\lambda}}_{\mathrm{2Q}}$ in quantum memory to trace-distance error at most $\epsilon$ using
\begin{equation}
\mathrm{GateCost}_{\mathrm{quantum}}
\;=\;
\mathrm{poly}\!\bigl(N,\,d,\,\log(\epsilon^{-1})\bigr),
\label{eq:flexible_quantum_cost}
\end{equation}
gates in total, including (i) the $\mathrm{poly}(N,d)$-size circuit $C^{\lambda\sigma_{\mathrm{can}}}$ and
(ii) the synthesis of $Q$ to accuracy compatible with overall trace-distance error $\epsilon$.
\end{boxedtheorem}

\begin{proof}
Apply $C^{\lambda\sigma_{\mathrm{can}}}$ to $\ket{0^{n_{\mathrm{1Q}}}}$ to obtain
$\ket{\psi^{\lambda\sigma_{\mathrm{can}}}}_{\mathrm{1Q}}$ (Eq.~\eqref{eq:flexible_input_state}), then apply
$Q$ as implemented in Sec.~\ref{sec:qqt} (cf.~Algorithm~\ref{alg:QQT_highlevel2}) to obtain
$\ket{\psi^\lambda}_{\mathrm{2Q}}$ (Eq.~\eqref{eq:flexible_target_state}). By
Theorem~\ref{thm:QQT_complexity_poly}, the canonical transform $Q$---implemented using the promised
canonical GT representatives $\mu_{\mathrm{can}}$ and fixed canonical multiplicity copy
$\sigma_{\mathrm{can}}$---can be synthesized with gate complexity
$\mathrm{poly}(N,d,\log(\epsilon^{-1}))$ by distributing the error budget across the strong Schur
transform and Jordan--Schwinger arithmetic subroutines. Composing with the given $\mathrm{poly}(N,d)$
circuit $C^{\lambda\sigma_{\mathrm{can}}}$ yields the stated overall gate complexity in
Eq.~\eqref{eq:flexible_quantum_cost}. The output is stored directly as a quantum state in quantum memory.
\end{proof}

\subsubsection{Classical runtime lower bound (explicit classical output)}

A classical algorithm cannot output a quantum state in quantum memory. Thus, when restricted to
classical computation, Task~\ref{task:flexible_output_quantization} reduces to producing an
\emph{explicit coefficient list} that approximates
$\ket{\psi^\lambda}_{\mathrm{2Q}}\,(= Q\ket{\psi^{\lambda\sigma_{\mathrm{can}}}}_{\mathrm{1Q}})$ as in
Eq.~\eqref{eq:flexible_target_state} up to trace-distance error $\epsilon$.

Recall that, under the canonical Gelfand--Tsetlin promise of Sec.~\ref{sec:qqt}, the output of $Q$ is
supported on the occupation-number register, whose reachable states span
$\{\ket{n}:n\in\Omega_\lambda\}$, with $\Omega_\lambda$ the set of \emph{distinct} occupation vectors
(equivalently, distinct $\mathfrak{u}(d)$ weights) occurring in $\mathcal{Q}^{(d)}_\lambda$. We therefore
take as the relevant dimension the size of this occupation set,
\begin{equation}
D_{\lambda}(N,d)\ :=\ \bigl|\Omega_\lambda\bigr|
\ =\ \dim\,\mathrm{span}\{\ket{n}:n\in\Omega_\lambda\},
\label{eq:Dlambda_def}
\end{equation}
which is the number of occupation-basis records needed to write the converted state. By
Theorem~\ref{thm:equivariant_bijection}, $\Omega_\lambda$ is exactly the weight set of
$\mathcal{Q}^{(d)}_\lambda$, so
\[
D_\lambda(N,d)\ \le\ \dim\mathcal{Q}^{(d)}_\lambda\ =\ \sum_{n}K_{\lambda,n},
\]
with equality precisely in the weight-multiplicity-free sectors (all Kostka numbers
$K_{\lambda,n}\le 1$), i.e.\ for bosons $\lambda=(N)$ and fermions $\lambda=(1^N)$ with $d\ge N$; in those
sectors the occupation vector is itself a complete label and
$D_\lambda(N,d)=\dim\mathcal{Q}^{(d)}_\lambda$. Fix an orthonormal occupation-number basis
$\{\ket{n}\}_{n\in\Omega_{\lambda}}$ for the output sector $\mathcal{F}^{(\lambda)}_{N}$, where
$|\Omega_{\lambda}|=D_{\lambda}(N,d)$.

\begin{boxedtheorem}[Classical explicit-output lower bound from output size]
\label{thm:flexible_classical_lb}
Fix any $\epsilon\in(0,1)$ and any $(N,d)$ with $\ell(\lambda)\le d$. There exists an input circuit
$C^{\lambda\sigma_{\mathrm{can}}}$ satisfying the promise of
Task~\ref{task:flexible_output_quantization}, including the canonical GT promise on
$\mu_{\mathrm{can}}$, such that any classical algorithm that outputs a coefficient list describing a pure
state $\ket{\widetilde{\psi}}$ with
\[
\tfrac12\bigl\|\proj{\widetilde{\psi}}{\widetilde{\psi}}-\proj{\psi^\lambda_{\mathrm{2Q}}}{\psi^\lambda_{\mathrm{2Q}}}\bigr\|_1\le \epsilon
\]
must output at least $\Omega(D_\lambda(N,d))$ coefficient--basis-string records, and hence must use
\begin{equation}
T_{\mathrm{classical}}\ \ge\ \Omega\!\bigl(D_{\lambda}(N,d)\bigr),
\qquad
M_{\mathrm{classical}}\ \ge\ \Omega\!\bigl(D_{\lambda}(N,d)\bigr),
\label{eq:flexible_classical_lb}
\end{equation}
where $M_{\mathrm{classical}}$ counts the number of records in the explicit list and
$T_{\mathrm{classical}}$ counts classical time steps (in particular, time to emit the list).
Consequently, choosing $\lambda$ that maximizes $D_\lambda(N,d)$ implies the same bound with
\begin{equation}
D_{\max}(N,d)\ :=\ \max_{\lambda\vdash N,\ \ell(\lambda)\le d}\ D_{\lambda}(N,d),
\qquad
T_{\mathrm{classical}},M_{\mathrm{classical}}\ \ge\ \Omega\!\bigl(D_{\max}(N,d)\bigr).
\label{eq:Dmax_def}
\end{equation}
\end{boxedtheorem}

\begin{proof}
We use a worst-case choice of the target output state inside the occupation sector
$\mathcal{F}^{(\lambda)}_{N}$ that (i) is consistent with the task promise for some
$\mathrm{poly}(N,d)$ circuit and (ii) forces any $\epsilon$-accurate explicit list to have support size
$\Omega(D_\lambda)$.

Let $\ket{\phi}\in \mathcal{F}^{(\lambda)}_{N}$ be the equal-magnitude superposition over the
occupation basis:
\begin{equation}
\ket{\phi}\ :=\ \frac{1}{\sqrt{D_\lambda}}\sum_{n\in\Omega_\lambda} e^{i\theta_n}\ket{n},
\label{eq:flat_state}
\end{equation}
for arbitrary phases $\theta_n\in[0,2\pi)$. Since $\ket{\phi}$ is supported exactly on the distinct
occupation vectors $\Omega_\lambda$, it lies in the domain of $Q^\dagger$, and
$\ket{\varphi^{\lambda\sigma_{\mathrm{can}}}}_{\mathrm{1Q}} := Q^\dagger\ket{\phi}$ is a canonical
first-quantized state in $\mathcal{H}_{\lambda\sigma_{\mathrm{can}}}$ (supported on the canonical GT
representatives by construction). Because $Q^\dagger$ is implementable with
$\mathrm{poly}(N,d,\log(\epsilon^{-1}))$ gates (Theorem~\ref{thm:QQT_complexity_poly}), there exists a
$\mathrm{poly}(N,d)$-depth circuit $C^{\lambda\sigma_{\mathrm{can}}}$ that prepares
$\ket{\varphi^{\lambda\sigma_{\mathrm{can}}}}_{\mathrm{1Q}}$ to sufficiently high precision from
$\ket{0^{n_{\mathrm{1Q}}}}$; we fix such a circuit as the input instance. For this instance, the correct
second-quantized output is exactly $\ket{\psi^\lambda}_{\mathrm{2Q}}=\ket{\phi}$.

Now let a classical algorithm output an explicit list describing a normalized
$\ket{\widetilde{\phi}}=\sum_n c_n\ket{n}$ such that
$\tfrac12\|\proj{\widetilde{\phi}}{\widetilde{\phi}}-\proj{\phi}{\phi}\|_1\le \epsilon$.
For pure states, trace distance satisfies
\begin{equation}
\tfrac12\bigl\|\proj{\widetilde{\phi}}{\widetilde{\phi}}-\proj{\phi}{\phi}\bigr\|_1
\;=\;
\sqrt{1-\bigl|\langle \widetilde{\phi}\!\mid\!\phi\rangle\bigr|^2}.
\label{eq:pure_trace_distance_overlap}
\end{equation}
Let $S\subseteq\Omega_\lambda$ be the set of basis indices that appear with nonzero coefficient in
the algorithm's output list, i.e.~the support of $\ket{\widetilde{\phi}}$ in the occupation basis.
Then $\ket{\widetilde{\phi}}$ is supported on $\mathrm{span}\{\ket{n}:n\in S\}$, and hence
\begin{equation}
\bigl|\langle \widetilde{\phi}\!\mid\!\phi\rangle\bigr|
\ \le\
\bigl\|\Pi_S\ket{\phi}\bigr\|_2
\ =\
\sqrt{\sum_{n\in S}\frac{1}{D_\lambda}}
\ =\
\sqrt{\frac{|S|}{D_\lambda}},
\label{eq:overlap_support_bound}
\end{equation}
where $\Pi_S$ is the projector onto $\mathrm{span}\{\ket{n}:n\in S\}$ and we used
Eq.~\eqref{eq:flat_state}. Combining Eq.~\eqref{eq:pure_trace_distance_overlap} and
Eq.~\eqref{eq:overlap_support_bound}, the condition
$\tfrac12\|\proj{\widetilde{\phi}}{\widetilde{\phi}}-\proj{\phi}{\phi}\|_1\le \epsilon$ implies
\[
\sqrt{1-\frac{|S|}{D_\lambda}}\ \le\ \epsilon
\qquad\Longrightarrow\qquad
|S|\ \ge\ (1-\epsilon^2)\,D_\lambda.
\]
Thus any $\epsilon$-accurate explicit coefficient list must contain at least
$(1-\epsilon^2)D_\lambda=\Omega(D_\lambda)$ nonzero records. This yields both
$M_{\mathrm{classical}}\ge \Omega(D_\lambda)$ and, since writing each record takes at least constant
time, $T_{\mathrm{classical}}\ge \Omega(D_\lambda)$. Maximizing over $\lambda$ gives
Eq.~\eqref{eq:Dmax_def}.
\end{proof}

\subsubsection{Determination of the largest sector}

Theorem~\ref{thm:flexible_classical_lb} reduces the classical worst-case cost to understanding
$D_{\max}(N,d)$ (Eq.~\eqref{eq:Dmax_def}). The bosonic sector $\lambda=(N)$ and the fermionic sector
$\lambda=(1^N)$ are weight-multiplicity-free, so for them $D_\lambda=|\Omega_\lambda|$ equals the
number of distinct occupation vectors:
\begin{align}
D_{(N)}(N,d)\ &=\ \binom{N+d-1}{N},
\label{eq:Dmax_ge_boson}\\
D_{(1^N)}(N,d)\ &=\ \binom{d}{N}\qquad (d\ge N),
\label{eq:Dmax_ge_fermion}
\end{align}
the bosonic count being the number of weak compositions of $N$ into $d$ nonnegative parts and the
fermionic count the number of $N$-element subsets of $\{1,\dots,d\}$ (equivalently, since these sectors
are weight-multiplicity-free, the values of the hook-content dimension formula
$\dim(\mathcal{Q}^{(d)}_{\lambda})=\prod_{(i,j)\in\lambda}\frac{d+j-i}{h(i,j)}$). Moreover
$\binom{N+d-1}{N}\ge \binom{d}{N}$ for $d\ge N$ (multiset combinations dominate set combinations).

In fact the maximum is attained by the bosonic sector. Every occupation vector of any sector $\lambda$
is a weak composition of $N$ into $d$ nonnegative parts, so $\Omega_\lambda\subseteq\Omega_{(N)}$; since
$\lambda=(N)$ realizes \emph{all} such compositions,
\begin{equation}
D_{\max}(N,d)\ =\ \max_{\lambda\vdash N,\ \ell(\lambda)\le d}\,\bigl|\Omega_\lambda\bigr|
\ =\ \binom{N+d-1}{N},
\label{eq:Dmax_value}
\end{equation}
attained at $\lambda=(N)$. (Note that for general $\lambda$ the irrep dimension
$\dim\mathcal{Q}^{(d)}_\lambda=\sum_n K_{\lambda,n}$ can exceed this value; it is the occupation-vector
count $|\Omega_\lambda|$, not the irrep dimension, that governs the explicit-output size.) We therefore
analyze the bosonic count, for which two regimes are worth emphasizing.

\paragraph{Fixed $N$, growing $d$ (polynomial in $d$ with degree $N$).}
For fixed particle number $N$,
\begin{equation}
\binom{N+d-1}{N}
\;=\;
\Theta\!\left(\frac{d^{N}}{N!}\right)
\qquad (d\to\infty,\ N\ \text{fixed}),
\label{eq:boson_asymp_fixedN}
\end{equation}
so the explicit classical output size is polynomial in $d$ but with exponent $N$.

\paragraph{$d=\Theta(N)$ (exponential in $N$).}
If $d=cN$ with constant $c>0$, then Stirling's approximation implies
\begin{equation}
\binom{N+d-1}{N}
\;=\;
\binom{(c+1)N-1}{N}
\;=\;
\exp\!\bigl(\Theta(N)\bigr),
\label{eq:boson_asymp_linear}
\end{equation}
so the explicit classical output size grows exponentially in $N$ in this scaling regime.

Combining Theorem~\ref{thm:flexible_quantum_complexity} with
Theorem~\ref{thm:flexible_classical_lb} and Eq.~\eqref{eq:Dmax_value} yields the operational
comparison: a quantum computer produces the converted state in quantum memory with
$\mathrm{poly}(N,d,\log(\epsilon^{-1}))$ gates (Eq.~\eqref{eq:flexible_quantum_cost}), whereas any
classical workflow that must materialize the fully converted state as an explicit coefficient list
requires time and memory at least
\begin{equation}
\Omega\!\Bigl(D_{\max}(N,d)\Bigr)
\ =\
\Omega\!\left(\binom{N+d-1}{N}\right),
\end{equation}
where the equality is Eq.~\eqref{eq:Dmax_value}. The bound is already degree-$N$ in $d$
via Eq.~\eqref{eq:boson_asymp_fixedN} and becomes exponential in $N$ when $d=\Theta(N)$ via
Eq.~\eqref{eq:boson_asymp_linear}. In this precise operational sense, application of the canonical
transform $Q$ offers a large separation between (i) poly-size coherent conversion and
(ii) explicit classical materialization of the converted state.

\subsection{Circuit-specified conversion with output in classical \textit{or} quantum memory: inverse direction}
\label{subsec:speed_up_either_memory_inverse}

We now formalize the inverse conversion task: given a classical description of a circuit that prepares
a fixed-$N$ second-quantized state in a promised statistics sector $\lambda$, we wish to convert it to
its first-quantized representation in a fixed Schur--Weyl multiplicity copy $\sigma_{\mathrm{can}}$ and
output either (i) the converted state in quantum memory or (ii) an explicit classical coefficient list
in the first-quantized computational basis. As above, we suppress explicit ancillary registers and use
the simplified notation of Eq.~\eqref{eq:Q_simplified_notation_flexible_output}, with the
understanding that the canonical implementation of $Q^\dagger$ uses the promised canonical GT
representatives $\mu_{\mathrm{can}}$.

\begin{boxedtheorem}[Task (Circuit-specified promised-$(\lambda,\sigma_{\mathrm{can}})$ second-to-first quantization with flexible output memory)]
\label{task:flexible_output_dequantization}
Fix particle number $N$, single-particle dimension $d$, a Young diagram $\lambda\vdash N$ with
$\ell(\lambda)\le d$, a canonical multiplicity label (for example, a canonical Young--Yamanouchi
word) $\sigma_{\mathrm{can}}$ of shape $\lambda$, and an accuracy parameter $\epsilon\in(0,1)$.
Let $\mathcal{F}^{(\lambda)}_{N}$ denote the fixed-$N$ second-quantized physical sector corresponding
to statistics type $\lambda$, with an orthonormal occupation-number basis
$\{\ket{n}\}_{n\in\Omega_\lambda}$.

You are given a classical description of a $\mathrm{poly}(N,d)$-depth quantum circuit
$C^{\lambda}$ acting on
\begin{equation}
m \;=\; \mathcal{O}\!\bigl(d\,\log(N+1)\bigr)
\label{eq:inverse_m_qubits}
\end{equation}
qubits (up to $\mathrm{poly}(N,d)$ ancillas) that encodes $\mathcal{F}^{(\lambda)}_{N}$ in the
computational basis, such that
\begin{equation}
\ket{\psi^{\lambda}}_{\mathrm{2Q}} \ :=\ C^{\lambda}\ket{0^m}
\label{eq:inverse_input_state_2Q}
\end{equation}
lies in the fixed sector $\mathcal{F}^{(\lambda)}_{N}$.

Let $\mathcal{H}_{\lambda\sigma_{\mathrm{can}}}\subset (\mathbb{C}^d)^{\otimes N}$ denote the
corresponding Schur--Weyl sector copy. The goal is to output the corresponding first-quantized state
\begin{equation}
\ket{\psi^{\lambda \sigma_{\mathrm{can}}}}_{\mathrm{1Q}}
\ :=\
Q^\dagger\ket{\psi^{\lambda}}_{\mathrm{2Q}}
\ \in\ \mathcal{H}_{\lambda\sigma_{\mathrm{can}}},
\label{eq:inverse_target_state_1Q}
\end{equation}
to trace-distance error at most $\epsilon$, with the output stored in \emph{either}:
\begin{enumerate}
\item \emph{quantum memory:} prepare a quantum register in a state $\widetilde{\rho}$ such that
\begin{equation}
\tfrac12\Bigl\|\widetilde{\rho}-\proj{\psi^{\lambda \sigma_{\mathrm{can}}}_{\mathrm{1Q}}}{\psi^{\lambda \sigma_{\mathrm{can}}}_{\mathrm{1Q}}}\Bigr\|_1\le \epsilon;
\label{eq:inverse_quantum_output_condition}
\end{equation}
or
\item \emph{classical memory:} output an explicit list of coefficient--basis-string pairs
$\{(\mathbf{i},\alpha_{\mathbf{i}})\}$ specifying a normalized vector
$\ket{\widetilde{\psi}}=\sum_{\mathbf{i}\in[d]^N}\alpha_{\mathbf{i}}\ket{\mathbf{i}}$
satisfying
\begin{equation}
\tfrac12\Bigl\|\proj{\widetilde{\psi}}{\widetilde{\psi}}-\proj{\psi^{\lambda\sigma_{\mathrm{can}}}_{\mathrm{1Q}}}{\psi^{\lambda\sigma_{\mathrm{can}}}_{\mathrm{1Q}}}\Bigr\|_1\le \epsilon.
\label{eq:inverse_classical_output_condition}
\end{equation}
\end{enumerate}
\end{boxedtheorem}

\subsubsection{Quantum runtime (quantum-memory output)}

The quantum algorithm is immediate: prepare $\ket{\psi^{\lambda}}_{\mathrm{2Q}}$ by applying $C^\lambda$
(Eq.~\eqref{eq:inverse_input_state_2Q}), fix the canonical sector labels
$(\lambda,\sigma_{\mathrm{can}})$ as part of the promised implementation of $Q^\dagger$, and then
apply the coherent inverse quantization transform $Q^\dagger$ to obtain
$\ket{\psi^{\lambda\sigma_{\mathrm{can}}}}_{\mathrm{1Q}}$ (Eq.~\eqref{eq:inverse_target_state_1Q}).

\begin{boxedtheorem}[Quantum complexity]
\label{thm:flexible_quantum_complexity_inverse}
Under the promises of Task~\ref{task:flexible_output_dequantization}, including the canonical GT
promise on the intermediate labels $\mu_{\mathrm{can}}$, a gate-based quantum computer can output
$\ket{\psi^{\lambda \sigma_{\mathrm{can}}}}_{\mathrm{1Q}}$ in quantum memory with trace-distance error
at most $\epsilon$ using
\begin{equation}
\mathrm{GateCost}_{\mathrm{quantum}}
\;=\;
\mathrm{poly}\!\bigl(N,\,d,\,\log(\epsilon^{-1})\bigr)
\label{eq:flexible_quantum_cost_inverse}
\end{equation}
gates in total, including preparation of $\ket{\psi^{\lambda}}_{\mathrm{2Q}}$ via $C^\lambda$ and
application of $Q^\dagger$.
\end{boxedtheorem}

\begin{proof}
Prepare $\ket{\psi^{\lambda}}_{\mathrm{2Q}}$ using $C^\lambda$
(Eq.~\eqref{eq:inverse_input_state_2Q}). Then apply $Q^\dagger$ as defined in Sec.~\ref{sec:qqt},
with fixed canonical multiplicity copy $\sigma_{\mathrm{can}}$ and the promised canonical GT support
$\mu_{\mathrm{can}}$. By Theorem~\ref{thm:QQT_complexity_poly}, both $Q$ and $Q^\dagger$ admit
implementations with gate complexity $\mathrm{poly}(N,d,\log(\epsilon^{-1}))$ after allocating the
error budget across the strong Schur transform and arithmetic subroutines so that the overall
trace-distance error is at most $\epsilon$
(Eq.~\eqref{eq:inverse_quantum_output_condition}). The output is stored directly as a quantum state in
quantum memory.
\end{proof}

\subsubsection{Classical runtime lower bound (explicit classical output)}

We now give a worst-case lower bound for any classical algorithm that must output an \emph{explicit
classical} coefficient list for the first-quantized computational-basis expansion of
$\ket{\psi^{\lambda \sigma_{\mathrm{can}}}}_{\mathrm{1Q}}$
(Eq.~\eqref{eq:inverse_target_state_1Q}).

\paragraph{Row and column groups and Young symmetrizers.}
Fix any standard Young tableau $T$ of shape $\lambda$. Let $R_T\le S_N$ be its \emph{row
group} (permutations preserving each row as a set) and $C_T\le S_N$ its \emph{column group}
(permutations preserving each column as a set). Define the (unnormalized) Young symmetrizer
\begin{equation}
e_T \ :=\ \Bigl(\sum_{\rho\in R_T} P_\rho\Bigr)\,
         \Bigl(\sum_{\kappa\in C_T} \mathrm{sgn}(\kappa)\,P_\kappa\Bigr),
\label{eq:young_symmetrizer_def_inverse}
\end{equation}
where $P(\pi)$ permutes tensor factors according to $\pi\in S_N$.

\begin{lemma}[Factorial computational-basis support from a Young symmetrizer]
\label{lem:factorial_support_young_symmetrizer}
Assume $d\ge N$, fix any partition $\lambda\vdash N$, and let $T$ be any standard Young tableau of
shape $\lambda$. Let $\ket{x}=\ket{1\,2\,\cdots\,N}\in(\mathbb{C}^d)^{\otimes N}$ denote a
computational basis string with all labels distinct (possible since $d\ge N$). Then the set of
computational basis strings appearing in $e_T\ket{x}$ has cardinality
\begin{equation}
S_\lambda \ :=\ |R_T|\,|C_T|\ =\ \Bigl(\prod_i \lambda_i!\Bigr)\Bigl(\prod_j (\lambda'_j)!\Bigr),
\label{eq:Slambda_inverse}
\end{equation}
where $\lambda'$ is the conjugate partition. Moreover, no cancellations occur in the computational
basis expansion of $e_T\ket{x}$: each basis string in the support appears exactly once with
coefficient $\pm 1$ (before normalization). In particular, by the injectivity established below,
$S_\lambda=|R_T||C_T|\le N!$, with equality for the single-row shape $\lambda=(N)$ and the single-column
shape $\lambda=(1^N)$.
\end{lemma}

\begin{proof}
Because $\ket{x}$ has all labels distinct, the map $\pi\mapsto P(\pi)\ket{x}$ is injective. Consider
the products $\{\rho\kappa:\rho\in R_T,\ \kappa\in C_T\}\subseteq S_N$. If
$\rho\kappa=\rho'\kappa'$, then $\rho^{-1}\rho'=\kappa\kappa'^{-1}\in R_T\cap C_T$. For a standard
tableau, $R_T\cap C_T=\{e\}$, since a permutation preserving every row and every column as sets must
fix each box. Thus $(\rho,\kappa)\mapsto \rho\kappa$ is injective and
$|\{\rho\kappa\}|=|R_T||C_T|\le|S_N|=N!$.

Expanding Eq.~\eqref{eq:young_symmetrizer_def_inverse},
\[
e_T\ket{x}
=\sum_{\rho\in R_T}\sum_{\kappa\in C_T}\mathrm{sgn}(\kappa)\,P_{\rho\kappa}\ket{x}.
\]
By injectivity, all $P_{\rho\kappa}\ket{x}$ are distinct computational basis states, so there are no
collisions and hence no cancellations; each basis state appears exactly once with coefficient
$\mathrm{sgn}(\kappa)\in\{\pm1\}$. Finally, $|R_T|=\prod_i \lambda_i!$ and
$|C_T|=\prod_j (\lambda'_j)!$ since permutations act independently within each row and within each
column, yielding Eq.~\eqref{eq:Slambda_inverse}.
\end{proof}

We now turn Lemma~\ref{lem:factorial_support_young_symmetrizer} into an explicit-output lower bound
that holds for \emph{any constant} $\epsilon\in(0,1)$. To keep the hard instance compatible with the
canonical-GT promise of Sec.~\ref{sec:qqt}, we state it for the weight-multiplicity-free sectors, where
the Young-symmetrizer state defined below has a single GT representative per weight and therefore lies in
the canonical subspace on which $Q$ is defined. These are exactly the sectors that maximize the
factorial support: among all $\lambda\vdash N$, $S_\lambda=\prod_i\lambda_i!\prod_j(\lambda'_j)!$ attains
its maximum value $N!$ precisely at $\lambda=(N)$ and $\lambda=(1^N)$, so restricting to them is without
loss for the worst-case bound.

\begin{boxedtheorem}[Classical explicit-output lower bound from factorial support]
\label{thm:flexible_classical_lb_inverse}
Assume $d\ge N$ and let $\lambda=(N)$ (bosons) or $\lambda=(1^N)$ (fermions), so that $S_\lambda=N!$ by
Eq.~\eqref{eq:Slambda_inverse}. Fix any $\epsilon\in(0,1)$. Then there exists a circuit
$C^\lambda$ satisfying the promise of Task~\ref{task:flexible_output_dequantization}, with fixed
canonical multiplicity label $\sigma_{\mathrm{can}}$ and canonical GT promise $\mu_{\mathrm{can}}$,
such that any classical algorithm that outputs an explicit coefficient list
(Eq.~\eqref{eq:inverse_classical_output_condition}) describing
$\ket{\psi^{\lambda\sigma_{\mathrm{can}}}}_{\mathrm{1Q}}$ in the first-quantized computational basis
to trace-distance error at most $\epsilon$ must use
\begin{equation}
T_{\mathrm{classical}}\ \ge\ \Omega\!\bigl(S_{\lambda}\bigr)\ =\ \Omega(N!),
\qquad
M_{\mathrm{classical}}\ \ge\ \Omega\!\bigl(S_{\lambda}\bigr)\ =\ \Omega(N!),
\label{eq:flexible_classical_lb_inverse}
\end{equation}
where $T_{\mathrm{classical}}$ is runtime and $M_{\mathrm{classical}}$ is the number of output
records. More generally, for any weight-multiplicity-free $\lambda$ (so that the construction below
respects the canonical-GT promise) the same argument yields $T_{\mathrm{classical}},
M_{\mathrm{classical}}\ge\Omega(S_\lambda)$ with $S_\lambda$ as in Eq.~\eqref{eq:Slambda_inverse}.
\end{boxedtheorem}

\begin{proof}
Let $T$ be any standard Young tableau of shape $\lambda$ and let $\ket{x}=\ket{1\,2\,\cdots\,N}$.
Define the normalized state
\begin{equation}
\ket{\phi}\ :=\ \frac{e_T\ket{x}}{\|e_T\ket{x}\|_2}.
\label{eq:phi_def_inverse}
\end{equation}
By Lemma~\ref{lem:factorial_support_young_symmetrizer}, $\ket{\phi}$ has computational-basis support
size $S_\lambda$. Moreover, since the unnormalized coefficients are all $\pm1$, normalization implies
that each basis string in the support has equal magnitude
\begin{equation}
|\langle \mathbf{i}\mid \phi\rangle| \;=\; \frac{1}{\sqrt{S_\lambda}}
\qquad\text{for all $\mathbf{i}$ in the support of $\ket{\phi}$.}
\label{eq:uniform_amplitudes_inverse}
\end{equation}

The state $\ket{\phi}$ lies in a single Schur--Weyl sector: $e_T(\mathbb{C}^d)^{\otimes N}$ is one copy of
the Weyl module $\mathcal{Q}^{(d)}_\lambda$ associated with the standard tableau $T$. Choose the fixed
multiplicity label $\sigma_{\mathrm{can}}$ in Task~\ref{task:flexible_output_dequantization} to match
this copy, i.e.~regard $\ket{\phi}\in \mathcal{H}_{\lambda\sigma_{\mathrm{can}}}$ (this is a choice of
basis convention for the multiplicity label). Because $\lambda$ is weight-multiplicity-free, every
weight space of $\mathcal{Q}^{(d)}_\lambda$ is one-dimensional, so its canonical GT representative is the
unique pattern of that weight and $\ket{\phi}$ is automatically supported on the canonical GT subspace on
which $Q$ is defined. Define the corresponding second-quantized state
\begin{equation}
\ket{\psi^\lambda}_{\mathrm{2Q}}
\ :=\
Q\,\ket{\phi}
\ \in\ \mathcal{F}^{(\lambda)}_{N}.
\label{eq:inverse_hard_instance_2Q}
\end{equation}
Since $Q$ is implementable with $\mathrm{poly}(N,d)$ gates
(Theorem~\ref{thm:QQT_complexity_poly}), there exists a $\mathrm{poly}(N,d)$-depth circuit
$C^\lambda$ that prepares $\ket{\psi^\lambda}_{\mathrm{2Q}}$ from $\ket{0^m}$. For this input
instance, the correct output satisfies
\[
\ket{\psi^{\lambda\sigma_{\mathrm{can}}}}_{\mathrm{1Q}}
\;=\;
Q^\dagger\ket{\psi^\lambda}_{\mathrm{2Q}}
\;=\;
\ket{\phi}
\]
by Eq.~\eqref{eq:inverse_target_state_1Q} and the definition in
Eq.~\eqref{eq:inverse_hard_instance_2Q}.

Now let a classical algorithm output an explicit list describing a normalized
$\ket{\widetilde{\phi}}=\sum_{\mathbf{i}} \widetilde{\alpha}_{\mathbf{i}}\ket{\mathbf{i}}$ such that
$\tfrac12\|\proj{\widetilde{\phi}}{\widetilde{\phi}}-\proj{\phi}{\phi}\|_1\le \epsilon$.
For pure states,
\begin{equation}
\tfrac12\Bigl\|\proj{\widetilde{\phi}}{\widetilde{\phi}}-\proj{\phi}{\phi}\Bigr\|_1
=\sqrt{1-|\langle \widetilde{\phi}\mid \phi\rangle|^2}.
\label{eq:inverse_pure_trace_distance_overlap}
\end{equation}
Let $S\subseteq [d]^N$ be the set of basis strings appearing with nonzero coefficient in the output
list, i.e.~the computational-basis support of $\ket{\widetilde{\phi}}$. Then
\begin{equation}
|\langle \widetilde{\phi}\mid \phi\rangle|
\ \le\
\|\Pi_S\ket{\phi}\|_2
\ =\
\sqrt{\sum_{\mathbf{i}\in S} |\langle \mathbf{i}\mid \phi\rangle|^2}
\ \le\
\sqrt{\frac{|S|}{S_\lambda}},
\label{eq:inverse_overlap_support_bound}
\end{equation}
where $\Pi_S$ projects onto $\mathrm{span}\{\ket{\mathbf{i}}:\mathbf{i}\in S\}$ and we used
Eq.~\eqref{eq:uniform_amplitudes_inverse} in the final inequality. Combining
Eq.~\eqref{eq:inverse_pure_trace_distance_overlap} and
Eq.~\eqref{eq:inverse_overlap_support_bound}, the condition
$\tfrac12\|\proj{\widetilde{\phi}}{\widetilde{\phi}}-\proj{\phi}{\phi}\|_1\le \epsilon$ implies
\[
\sqrt{1-\frac{|S|}{S_\lambda}}\ \le\ \epsilon
\qquad\Longrightarrow\qquad
|S|\ \ge\ (1-\epsilon^2)\,S_\lambda.
\]
Therefore any $\epsilon$-accurate explicit coefficient list must contain at least
$(1-\epsilon^2)S_\lambda=\Omega(S_\lambda)$ records, which yields both bounds in
Eq.~\eqref{eq:flexible_classical_lb_inverse}.
\end{proof}

\subsection{Promised-$(\lambda,\sigma_{\mathrm{can}})$ circuit sampling of occupation-number outcomes}
\label{subsec:promised_lambda_sampling}

In this subsection we study a \emph{sampling} primitive: instead of explicitly outputting a full
classical description of the converted state $Q\ket{\psi^{\lambda \sigma_{\mathrm{can}}}}_{\mathrm{1Q}}$,
we ask only for samples from the induced occupation-number measurement distribution. The central
message is twofold: (i) a quantum computer samples efficiently by coherently applying $Q$ and
measuring, and (ii) if a \emph{classical} polynomial-time algorithm could sample (to constant total-variation error)
for all inputs satisfying the Schur--Weyl promise, then one could classically sample output
distributions of arbitrary quantum circuits on
$m_\lambda:=\lceil\log_2 D_\lambda(N,d)\rceil$ qubits to the same constant accuracy. For any scaling
family in which $m_\lambda$ grows, this would imply a complexity-theoretic collapse
(e.g.~$\mathrm{BQP}\subseteq\mathrm{BPP}$), which is widely believed to be false.

As above, we suppress explicit ancilla and label registers in the notation for $Q$ and $Q^\dagger$.
All implementations are understood to use the canonical multiplicity label $\sigma_{\mathrm{can}}$ and
the canonical GT representatives $\mu_{\mathrm{can}}$ of Sec.~\ref{sec:qqt}.

\begin{boxedtheorem}[Task (Promised-$(\lambda,\sigma_{\mathrm{can}})$ circuit-specified occupation-number sampling)]
\label{task:promised_lambda_sampling_tight}
Fix particle number $N$, single-particle dimension $d\ge N$, and a Young diagram $\lambda\vdash N$
with $\ell(\lambda)\le d$. Fix a \emph{canonical} multiplicity label $\sigma_{\mathrm{can}}$
(for example, a fixed canonical Young--Yamanouchi word of shape $\lambda$), and let
$\mathcal{H}_{\lambda\sigma_{\mathrm{can}}}\subset(\mathbb{C}^d)^{\otimes N}$ denote the
corresponding Schur--Weyl sector copy.

You are given a classical description of a quantum circuit $C^{\lambda \sigma_{\mathrm{can}}}$
acting on
\begin{equation}
n_{\mathrm{1Q}} \;=\; N\lceil\log_2 d\rceil
\label{eq:sampling_n_qubits}
\end{equation}
qubits such that the prepared pure state
\begin{equation}
\ket{\psi^{\lambda \sigma_{\mathrm{can}}}}_{\mathrm{1Q}} \;:=\; C^{\lambda \sigma_{\mathrm{can}}}\ket{0^{n_{\mathrm{1Q}}}}
\qquad\text{satisfies}\qquad
\ket{\psi^{\lambda \sigma_{\mathrm{can}}}}_{\mathrm{1Q}} \in \mathcal{H}_{\lambda \sigma_{\mathrm{can}}}.
\label{eq:sampling_input_state}
\end{equation}
Moreover, the implementation of $Q$ is assumed to use the promised canonical GT representatives
$\mu_{\mathrm{can}}$. Define the induced occupation-number measurement distribution
\begin{equation}
p_{\lambda \sigma_{\mathrm{can}}}(n)\ :=\ \left|\bra{n}\,Q\,\ket{\psi^{\lambda \sigma_{\mathrm{can}}}}_{\mathrm{1Q}}\right|^2,
\qquad n\in\Omega_\lambda,
\label{eq:promised_lambda_sampling_dist_tight}
\end{equation}
where $\{\ket{n}\}_{n\in\Omega_\lambda}$ is a fixed orthonormal occupation-number basis for
$\mathcal{F}^{(\lambda)}_N$. The goal is to output a classical sample
$n\sim p_{\lambda\sigma_{\mathrm{can}}}(\cdot)$, either exactly or within total variation distance
at most $\epsilon\in(0,1)$.
\end{boxedtheorem}

\subsubsection{Quantum algorithm and runtime}

The quantum solution is immediate: prepare
$\ket{\psi^{\lambda \sigma_{\mathrm{can}}}}_{\mathrm{1Q}}$ using
$C^{\lambda \sigma_{\mathrm{can}}}$ (Eq.~\eqref{eq:sampling_input_state}), apply $Q$
(Algorithm~\ref{alg:QQT_highlevel2}), and measure the occupation register in the computational basis.

\begin{boxedtheorem}[Quantum complexity of promised-$(\lambda,\sigma_{\mathrm{can}})$ occupation-number sampling]
\label{thm:promised_lambda_sampling_quantum_tight}
Under Task~\ref{task:promised_lambda_sampling_tight}, a gate-based quantum computer can sample from
$p_{\lambda\sigma_{\mathrm{can}}}(\cdot)$ within total variation error at most $\epsilon$ using
\begin{equation}
\mathrm{GateCost}_{\mathrm{quantum}}
\;=\;
\mathrm{GateCost}(C^{\lambda\sigma_{\mathrm{can}}})
\;+\;
\mathrm{poly}\!\bigl(N,\,d,\,\log(\epsilon^{-1})\bigr)
\label{eq:promised_lambda_sampling_quantum_cost_tight}
\end{equation}
gates, where the second term is the cost of implementing $Q$ to accuracy compatible with
total-variation error~$\epsilon$. In particular, if
$\mathrm{GateCost}(C^{\lambda\sigma_{\mathrm{can}}})=\mathrm{poly}(N,d)$ (as in the other task
statements), then
$\mathrm{GateCost}_{\mathrm{quantum}}=\mathrm{poly}(N,d,\log(\epsilon^{-1}))$.
\end{boxedtheorem}

\begin{proof}
Prepare
$\ket{\psi^{\lambda \sigma_{\mathrm{can}}}}_{\mathrm{1Q}}=C^{\lambda \sigma_{\mathrm{can}}}\ket{0^{n_{\mathrm{1Q}}}}$
and apply an implementation $\widetilde{Q}$ of $Q$. By Theorem~\ref{thm:QQT_complexity_poly}, for
any $\eta\in(0,1)$ there exists $\widetilde{Q}$ whose induced channel
$\widetilde{\mathcal{E}}(\cdot)=\widetilde{Q}(\cdot)\widetilde{Q}^\dagger$ approximates
$\mathcal{E}(\cdot)=Q(\cdot)Q^\dagger$ in diamond norm, with gate complexity
$\mathrm{poly}(N,d,\log(\eta^{-1}))$ such that
\begin{equation}
\|\widetilde{\mathcal{E}}-\mathcal{E}\|_\diamond\ \le\ \eta.
\label{eq:sampling_diamond_bound}
\end{equation}
(Equivalently, this follows from an operator-norm synthesis bound $\|\widetilde{Q}-Q\|_\infty\le\eta/2$
via $\|\widetilde{\mathcal{E}}-\mathcal{E}\|_\diamond\le 2\|\widetilde{Q}-Q\|_\infty$.)
Let $\rho := \proj{\psi^{\lambda \sigma_{\mathrm{can}}}_{\mathrm{1Q}}}{\psi^{\lambda \sigma_{\mathrm{can}}}_{\mathrm{1Q}}}$
and define the output states $\rho_{\mathrm{out}} := Q\rho Q^\dagger=\mathcal{E}(\rho)$ and
$\widetilde{\rho}_{\mathrm{out}} := \widetilde{Q}\rho \widetilde{Q}^\dagger=\widetilde{\mathcal{E}}(\rho)$.
Diamond-norm control implies trace-distance control on states:
\begin{equation}
\frac12\bigl\|\widetilde{\rho}_{\mathrm{out}}-\rho_{\mathrm{out}}\bigr\|_1
\ \le\ \frac12\|\widetilde{\mathcal{E}}-\mathcal{E}\|_\diamond
\ \le\ \frac{\eta}{2}.
\label{eq:sampling_trace_bound}
\end{equation}

Measuring in the occupation-number basis is a CPTP map $\mathcal{M}$ from density operators to
classical distributions, and trace distance is contractive under CPTP maps. Therefore the total
variation distance between the induced measurement distributions satisfies
\begin{equation}
\bigl\| \mathcal{M}(\widetilde{\rho}_{\mathrm{out}})-\mathcal{M}(\rho_{\mathrm{out}})\bigr\|_{\mathrm{TV}}
\ \le\ \frac12\bigl\|\widetilde{\rho}_{\mathrm{out}}-\rho_{\mathrm{out}}\bigr\|_1
\ \le\ \frac{\eta}{2}.
\end{equation}
Choosing $\eta=2\epsilon$ yields a sampler whose output distribution is within total variation
distance at most $\epsilon$ from $p_{\lambda\sigma_{\mathrm{can}}}(\cdot)$, with the stated gate
complexity.
\end{proof}

\subsubsection{A uniform classical hardness reduction}

We now show that an efficient classical sampler for
Task~\ref{task:promised_lambda_sampling_tight} for \emph{all} promised inputs would yield an efficient
classical sampler for the output distributions of arbitrary quantum circuits on
$m_\lambda:=\lceil\log_2 D_\lambda(N,d)\rceil$ qubits (to the same total-variation accuracy). Here, as in
Eq.~\eqref{eq:Dlambda_def}, $D_\lambda(N,d)=|\Omega_\lambda|$ is the number of distinct occupation
vectors of the sector, which is the number of basis outcomes available to the occupation measurement.

\paragraph{Encoding into the occupation-number basis.}
Fix $\lambda$ and abbreviate $D_\lambda:=D_\lambda(N,d)=|\Omega_\lambda|$. Define
\begin{equation}
m_\lambda\ :=\ \left\lceil \log_2 D_\lambda \right\rceil.
\label{eq:m_lambda_def_tight}
\end{equation}
Fix an efficiently computable injection
$\mathrm{Enc}_\lambda:\{0,1\}^{m_\lambda}\hookrightarrow \Omega_\lambda$, together with an
efficiently computable inverse $\mathrm{Enc}_\lambda^{-1}$ on its image. This defines an isometry
\begin{equation}
V_\lambda:\ (\mathbb{C}^2)^{\otimes m_\lambda}\hookrightarrow \mathcal{F}^{(\lambda)}_{N},
\qquad
V_\lambda\ket{x}=\ket{\mathrm{Enc}_\lambda(x)}.
\label{eq:Vlambda_isometry_tight}
\end{equation}
We assume $V_\lambda$ is implemented as a reversible classical circuit of size
$\mathrm{poly}(m_\lambda)$ that computes $\mathrm{Enc}_\lambda$ (and its inverse on the image).

\begin{boxedtheorem}[Classical sampling hardness via reduction from quantum circuit sampling]
\label{thm:promised_lambda_sampling_hardness_tight}
Assume $d\ge N$ and fix any $\lambda\vdash N$ with $\ell(\lambda)\le d$ and canonical
$\sigma_{\mathrm{can}}$. Suppose there exists a classical probabilistic polynomial-time algorithm
$\mathcal{A}$ that, for every circuit $C^{\lambda \sigma_{\mathrm{can}}}$ satisfying the promise of
Task~\ref{task:promised_lambda_sampling_tight}, outputs a sample from a distribution
$\widetilde{p}_{\lambda \sigma_{\mathrm{can}}}$ such that
\begin{equation}
\left\| \widetilde{p}_{\lambda \sigma_{\mathrm{can}}} - p_{\lambda \sigma_{\mathrm{can}}} \right\|_{\mathrm{TV}}\ \le\ \epsilon
\label{eq:TV_assumption_tight}
\end{equation}
for some fixed constant $\epsilon\in(0,1)$. Then there exists a classical probabilistic
polynomial-time algorithm that, given any $m_\lambda$-qubit circuit $U$, outputs a sample from a
distribution $\widetilde{q}_U$ satisfying
\begin{equation}
\left\| \widetilde{q}_U - q_U \right\|_{\mathrm{TV}}\ \le\ \epsilon,
\qquad
q_U(x):=\bigl|\bra{x}U\ket{0^{m_\lambda}}\bigr|^2.
\label{eq:TV_consequence_tight}
\end{equation}
Consequently, for any scaling family of instances in which $m_\lambda$ grows at least logarithmically
with the input size, and for any fixed $\epsilon<\tfrac12$, the existence of $\mathcal{A}$ implies
$\mathrm{BQP}\subseteq\mathrm{BPP}$ for that family.
\end{boxedtheorem}

\begin{proof}
Fix $\lambda$ and set $m:=m_\lambda$ (Eq.~\eqref{eq:m_lambda_def_tight}). Let $U$ be an arbitrary
$m$-qubit quantum circuit.

\emph{Step 1 (embed $U$ into the occupation sector).}
Define the occupation-sector state
\begin{equation}
\ket{\Phi_U}\ :=\ V_\lambda\,U\,\ket{0^m}\ \in\ \mathcal{F}^{(\lambda)}_{N}.
\label{eq:PhiU_def_tight}
\end{equation}
By construction, measuring $\ket{\Phi_U}$ in the occupation-number basis yields
$n=\mathrm{Enc}_\lambda(x)$ with probability
\begin{equation}
\Pr[n=\mathrm{Enc}_\lambda(x)] \;=\; q_U(x).
\label{eq:PhiU_measure_tight}
\end{equation}
In particular, $\ket{\Phi_U}$ has support entirely contained in
$\mathrm{Im}(\mathrm{Enc}_\lambda)\subseteq\Omega_\lambda$.

\emph{Step 2 (pull back to the promised first-quantized sector).}
Let $Q$ denote the fixed promised-$\sigma_{\mathrm{can}}$ quantization transform from the sector
$\mathcal{H}_{\lambda\sigma_{\mathrm{can}}}$ to $\mathcal{F}^{(\lambda)}_{N}$ (Sec.~\ref{sec:qqt}),
implemented with the canonical GT promise on $\mu_{\mathrm{can}}$. Since $\ket{\Phi_U}$ is supported on
occupation vectors $\Omega_\lambda$, it lies in the canonical occupation subspace that is the image of
$Q$, so $Q^\dagger\ket{\Phi_U}$ is well defined and supported on the canonical GT representatives.
Define
\begin{equation}
\ket{\psi^{\lambda\sigma_{\mathrm{can}}}_U}_{\mathrm{1Q}}\ :=\ Q^\dagger\ket{\Phi_U}\ \in\ \mathcal{H}_{\lambda\sigma_{\mathrm{can}}}.
\label{eq:psiU_def_tight}
\end{equation}
Then, by unitarity on the restricted sector,
\begin{equation}
Q\ket{\psi^{\lambda\sigma_{\mathrm{can}}}_U}_{\mathrm{1Q}} = \ket{\Phi_U}.
\label{eq:Q_cancels_tight}
\end{equation}

\emph{Step 3 (compile an input circuit for the sampling task).}
We claim there exists a quantum circuit $C^{\lambda\sigma_{\mathrm{can}}}_U$ acting on
$n_{\mathrm{1Q}}=N\lceil\log_2 d\rceil$ qubits that prepares
$\ket{\psi^{\lambda\sigma_{\mathrm{can}}}_U}_{\mathrm{1Q}}$ from $\ket{0^{n_{\mathrm{1Q}}}}$, with size polynomial in
$N,d,|U|$. Indeed, implement $U$ on an explicit $m$-qubit workspace, implement the reversible encoder
$V_\lambda$ using $\mathrm{poly}(m)$ Toffoli/CNOT gates, then apply $Q^\dagger$ using
Theorem~\ref{thm:QQT_complexity_poly}, and uncompute any ancillas. This produces
$\ket{\psi^{\lambda\sigma_{\mathrm{can}}}_U}_{\mathrm{1Q}}$ on a register that can be embedded into
the $n_{\mathrm{1Q}}$-qubit first-quantized encoding because, with $D_\lambda=|\Omega_\lambda|\le\binom{N+d-1}{N}$
(every occupation vector being a weak composition of $N$ into $d$ parts),
\begin{equation}
m\ =\ \left\lceil\log_2 D_\lambda(N,d)\right\rceil\ \le\ \left\lceil\log_2 \binom{N+d-1}{N}\right\rceil
\ =\ \mathcal{O}\!\bigl(N\log d + N\bigr),
\end{equation}
so the required workspace is at most linear in $n_{\mathrm{1Q}}$ (and therefore embeddable with polynomial
overhead). Thus $C^{\lambda\sigma_{\mathrm{can}}}_U$ satisfies the promise of
Task~\ref{task:promised_lambda_sampling_tight}.

\emph{Step 4 (invoke the classical sampler and decode).}
Run $\mathcal{A}$ on input $C^{\lambda\sigma_{\mathrm{can}}}_U$ to obtain an outcome
$n\in\Omega_\lambda$ distributed according to
$\widetilde{p}_{\lambda\sigma_{\mathrm{can}}}$ with
$\|\widetilde{p}_{\lambda\sigma_{\mathrm{can}}}-p_{\lambda\sigma_{\mathrm{can}}}\|_{\mathrm{TV}}\le\epsilon$.
Define the deterministic decoding map $f:\Omega_\lambda\to\{0,1\}^m$ by
\begin{equation}
f(n)\ :=\
\begin{cases}
\mathrm{Enc}_\lambda^{-1}(n), & n\in \mathrm{Im}(\mathrm{Enc}_\lambda),\\
0^m, & \text{otherwise}.
\end{cases}
\label{eq:decoder_map}
\end{equation}
Let $\widetilde{q}_U:=f_\#(\widetilde{p}_{\lambda\sigma_{\mathrm{can}}})$ and
$q_U:=f_\#(p_{\lambda\sigma_{\mathrm{can}}})$ denote the pushforwards. Total variation distance is
non-increasing under deterministic post-processing, hence
\begin{equation}
\|\widetilde{q}_U-q_U\|_{\mathrm{TV}}\ \le\ \|\widetilde{p}_{\lambda\sigma_{\mathrm{can}}}-p_{\lambda\sigma_{\mathrm{can}}}\|_{\mathrm{TV}}\ \le\ \epsilon.
\label{eq:TV_dataprocessing_tight}
\end{equation}
Finally, since $\ket{\Phi_U}$ is supported on $\mathrm{Im}(\mathrm{Enc}_\lambda)$ and
Eq.~\eqref{eq:PhiU_measure_tight} holds, we have
$q_U(x)=|\bra{x}U\ket{0^m}|^2$ as claimed. Therefore $\widetilde{q}_U$ approximates $q_U$ within
total variation distance at most $\epsilon$, proving Eq.~\eqref{eq:TV_consequence_tight}.

For the complexity-class consequence, fix any language $L\in\mathrm{BQP}$ with a verifier circuit whose
acceptance probability is bounded away from $\tfrac12$ by a constant gap; standard amplification makes the
relevant output bit near-deterministic, so a sampler with $\|\widetilde{q}_U-q_U\|_{\mathrm{TV}}\le\epsilon$
for fixed $\epsilon<\tfrac12$ decides $L$ correctly with probability bounded above $\tfrac12$, placing
$L\in\mathrm{BPP}$.
\end{proof}

\section{Open questions and extensions of this work \label{sec:optimizing_algo}}

In this section we present several natural directions for improved implementations of $Q$
and for extending its scope to broader classes of many-body states and computational workflows.

\begin{enumerate}
\item \textbf{More optimal, statistics-aware transforms.}
The construction of $Q$ in Sec.~\ref{sec:qqt} is deliberately \emph{symmetry-agnostic}: it treats the input as an
arbitrary state supported on a fixed Schur--Weyl sector, and it uses generic subroutines that do not exploit
additional structure that may be known \emph{a priori} about the particle statistics or the relevant irrep.
This generality is a strength, but it also suggests viewing our $Q$ as a \emph{root algorithm} from which
more specialized (and potentially asymptotically faster) transforms can be derived. Two immediate targets
are: (i) \emph{irrep-specialized} strong Schur transforms, i.e., implementations tailored to a fixed
$\lambda$ that avoid the overhead of coherently producing labels that are irrelevant or constant on the
promised sector; and (ii) \emph{sparse arithmetic} on GT patterns. For a fixed $\lambda$, the admissible GT
patterns occupy a structured subset of the ambient integer grid, and in many regimes (e.g., restricted row
lengths, thin diagrams, or diagrams with bounded depth) only a small portion of GT entries can vary
significantly. One can therefore hope to compute the row-sum differences in
Eq.~\eqref{eq:A_occupations_from_rowdiff} using fewer additions by reusing partial sums, streaming only the
nontrivial entries, or using alternative encodings. Understanding which families of $\lambda$ admit nontrivial savings, and how these
savings compose with high-dimensional Schur-transform constructions, is an open resource-optimization
problem with potentially large practical impact.

\item \textbf{Symmetry verification, sector filtering, and leakage removal via the $\lambda$ register.}
We presented $Q$ as a transform acting on states supported on a single Young diagram $\lambda$.
A natural extension is to retain the $\lambda$ register output by $U_{\mathrm{Schur}}$ and use it to perform symmetry-resolved verification and filtering when the input state has leaked outside a target permutation-symmetry sector.
Concretely, after the Schur transform one has
\[
\sum_{\lambda,\mu,\sigma} \alpha_{\lambda,\mu,\sigma}\ket{\lambda}\ket{\mu}\ket{\sigma}\ket{0},
\]
and one may then apply a controlled family of arithmetic maps,
\[
\sum_{\lambda,\mu,\sigma} \alpha_{\lambda,\mu,\sigma}\ket{\lambda}\ket{\mu}\ket{\sigma}\ket{0}
\ \longmapsto\
\sum_{\lambda,\mu,\sigma} \alpha_{\lambda,\mu,\sigma}\ket{\lambda}\ket{\mu}\ket{\sigma}\ket{n(\lambda,\mu)},
\]
where $n(\lambda,\mu)$ denotes the occupation vector associated with the GT label $\mu$ in the $U(d)$ irrep
$\mathcal{Q}^{(d)}_\lambda$.
Keeping the $\lambda$ register explicit allows one to detect, flag, postselect, or coherently condition on the symmetry sector before further processing.
In particular, if imperfect state preparation, truncation, approximation, or hardware noise induces unphysical leakage outside a desired Young sector $\lambda_0$, then measuring or conditioning on $\lambda$ provides a direct mechanism for symmetry verification and projection back into the intended subspace.
In this sense, the extension provides a symmetry filter or sector-resolved projector that can suppress unwanted exchange-symmetry leakage prior to mapping the state to occupation-number data.

It is important to emphasize, however, that a first-quantized state with coherent support on multiple Schur--Weyl sectors should not be interpreted as the natural description of ordinary mixed-species systems such as electron--phonon, electron--nuclear, or more general coupled fermion--boson wavefunctions.
Such systems are instead described by tensor products of species-resolved Hilbert spaces with symmetrization or antisymmetrization imposed separately within each species.
Accordingly, the present extension is best understood as a symmetry-resolved processing tool on the full tensor-power space, useful for sector filtering, leakage removal, and related algorithmic tasks, rather than as a first-quantized representation of composite systems with distinct particle statistics.

\item \textbf{Outputting the color-resolved Fock space for parastatistics.}
Throughout we worked in the \emph{physical} fixed-$N$ Fock space appropriate to the chosen statistics
sector, treating the Green color degrees of freedom as an unobserved multiplicity.
However, parabose and parafermi sectors admit natural \emph{color-resolved} realizations in which the
microscopic occupation data $n_p^{(\alpha)}$ are explicit (Sec.~\ref{subsubsec:paraparticles_and_color_space}).
An intriguing extension is therefore to define a refined transform that outputs not only the physical
occupations $n_p=\sum_\alpha n_p^{(\alpha)}$, but also a coherent encoding of a \emph{representative} (or a
controlled superposition of representatives) of the compatible color configurations.
Such an output could be useful in at least two ways. First, it would enable direct simulation of
Hamiltonians that couple to color degrees of freedom (or that are naturally expressed in a color-resolved
operator algebra), while retaining compatibility with the physical paraoperators via coherent summation over
colors. Second, it could provide a convenient handle for implementing symmetry constraints or selection
rules that are more transparent in the enlarged space, for example by performing controlled operations on
the color registers and subsequently projecting (coherently) back to the physical subspace.

\item \textbf{Fault-tolerant optimization, error budgeting, and integration with simulation primitives.}
Our complexity statements treat the strong Schur transform and the arithmetic map as modular components and
bound their costs at the level of asymptotic gate counts. For practical deployments, several refinements are important. One direction is to optimize the \emph{error
budget} across (i) synthesis error in $U_{\mathrm{Schur}}$, (ii) reversible arithmetic error (if approximate
adders or approximate encodings are used), and (iii) any subsequent simulation routine (e.g., qubitization
or quantum signal processing) that consumes the converted representation. A second direction is to
co-design $Q$ with the downstream primitive: for example, if the post-conversion algorithm uses a
block-encoding of number-conserving generators, it may be advantageous to keep certain intermediate labels
(e.g., partial row sums or weight differences) in coherent workspace rather than fully uncomputing them,
thereby trading qubits for Toffoli depth. Finally, it is natural to ask for tighter end-to-end resource
estimates for representative application families (chemistry- and materials-motivated Hamiltonians, coupled
fermion--boson models, and symmetry-restricted ans\"atze), including constant factors and compilation
overheads. Establishing such optimized resource models would clarify when conversion between
representations is merely feasible and when it is decisively advantageous.
\end{enumerate}

\section{Applications \label{sec:application}}

Our coherent quantization transform $Q$ and its inverse $Q^\dagger$ are best viewed as \emph{representation routers}: they let an algorithm move between two complete but differently optimized encodings of the same fixed-$N$ many-body physics.  The practical value of such a router is that quantum workflows are rarely monolithic.  Instead, they typically combine distinct primitives like state preparation, dynamics, spectroscopy, estimation, and post-processing. Naturally, the  resource bottlenecks for each primitive implementation depend differently on $(N,d)$ and on the operator families being implemented. Below we outline several settings in which the ability to switch coherently between first and second quantization is operationally useful.

\begin{enumerate}
\item \textbf{Hybrid workflows that alternately require fixed-$N$ efficiency and particle-number flexibility.}
Many core simulation and inference routines are naturally posed at fixed particle number, where the first-quantized register size scales as $N\lceil\log_2 d\rceil$ and can be dramatically smaller than a mode-register representation when $N\ll d$ (Sec.~\ref{subsec:1q_resource_tradeoff}). For example,  eigenstate preparation and eigenphase estimation within a known-$N$ subspace.  However, downstream tasks in the same workflow may require coherent access to sectors with different particle number, or the ability to couple to reservoirs, sources, or probes that add or remove excitations.  This type of logic is most naturally expressed in second quantization, where creation and annihilation operators and their polynomials provide a canonical operator language and where varying $N$ is built into the Fock-space structure (Sec.~\ref{sec:second_q}).

In such settings, $Q$ provides a clean interface: one can carry out the fixed-$N$ heavy lifting in first quantization  and then switch coherently to the occupation representation precisely at the point where particle-number flexibility becomes the dominant algorithmic need.  Importantly, because $Q$ is unitary and coherent, this switch can be performed even when the state is entangled with other registers that must remain unmeasured (for example, phase estimation work registers, coherent control registers, or ancillas used for amplitude amplification).

\item \textbf{Leveraging second-quantized state-preparation routines to obtain first-quantized inputs via $Q^\dagger$.} Preparing physically meaningful first-quantized input states is a persistent practical challenge, especially when the natural description of the target is given in terms of occupation numbers. At the same time, there is a mature and rapidly growing toolbox of second-quantized preparation subroutines which can be leveraged.

When the downstream algorithm benefits resource-wise from a first-quantized input, it is natural to prepare the state in the representation where preparation is simplest and then convert it coherently.  Concretely, one may prepare a fixed-$N$ state $\ket{\psi}_{\mathrm{2Q}}$ using existing second-quantized techniques, append the classical labels $(\lambda,\sigma)$, and apply $Q^\dagger$ to obtain a first-quantized state.  This provides a systematic pathway to \emph{reuse} second-quantized preparation advances while enabling first-quantized downstream primitives, without requiring one to design bespoke first-quantized preparation schemes for each application family.

\item \textbf{Switching representations \emph{within} Hamiltonian simulation primitives}
It is tempting to imagine toggling between representations during dynamics in order to exploit representation-dependent advantages for different Hamiltonian fragments (for example, implementing some terms more naturally in a mode basis and others more naturally in a particle basis).  While this idea is conceptually appealing, it is also the most delicate to justify, and it is not universally beneficial.

The reason is structural: a mid-circuit conversion is only advantageous if the savings achieved over a \emph{long} segment of computation outweigh the cost of at least two conversions (a $Q$ and a $Q^\dagger$).  Since our conversion cost is polynomial in $(N,d,\log(\epsilon^{-1}))$ but not negligible, such a strategy is unlikely to be worthwhile for short simulation horizons or for instances where a single representation already admits a well-optimized block-encoding or product-formula compilation of the full Hamiltonian.  On the other hand, there are plausible large-scale regimes where switching could become rational.  In these cases the conversion overhead is amortized across long segments of computation, and $Q$ provides the mechanism to do so coherently.

We therefore view representation-switching \emph{during} simulation not as a default recipe, but as a resource-engineering option whose value is instance- and scale-dependent.  Identifying concrete Hamiltonian families and algorithmic schedules where such amortization provably yields a net gain remains an interesting direction for future benchmarking.

\item \textbf{Representation-aware post-processing: measurements, observables, and reduced data products.}
Even when the dynamics or state-preparation stage is naturally carried out in one representation, the \emph{quantities ultimately extracted} from the computation may be cheaper to access in the other. Several common examples illustrate this point.

First, number-diagonal observables and occupation statistics are native in second quantization: measuring occupation registers directly yields particle-number distributions over modes, and many physically important diagnostics are expressed succinctly in the number basis.  If a state has been evolved or prepared in a first-quantized register for qubit or gate-efficiency reasons, applying $Q$ followed by measurement provides direct access to these distributions (Sec.~\ref{subsec:promised_lambda_sampling}).

Conversely, observables that are simple $k$-body operators in the particle picture, or that act on a small number of particle registers, can be more naturally estimated in first quantization.  In such cases, one may prefer to perform Hamiltonian simulation in the occupation representation  and then apply $Q^\dagger$ to move into a first-quantized register where the target observable admits a lower-overhead measurement strategy. More broadly, because $Q$ and $Q^\dagger$ are coherent, they can be inserted \emph{before} an estimation primitive like Hadamard tests, phase estimation variants, amplitude estimation, or shadow-style measurement routines thus allowing the representation choice to be tailored to the measurement and post-processing bottleneck rather than to the state-preparation bottleneck.

In this sense, $Q$ enables a modular design principle: pick the representation that is cheapest for the current stage of the workflow, and switch coherently when the next stage has a different natural representation.  This modularity becomes increasingly valuable as fault-tolerant workflows grow in depth, where constant-factor savings in repeated subroutines can dominate end-to-end resource costs.
\end{enumerate}

\section{Summary and Conclusions \label{sec:conclusions}}

The equivalence between the first- and second-quantized descriptions of identical particles is among the oldest in many-body physics, and it is almost always treated as bookkeeping: two interchangeable ledgers for the same fixed-$N$ state. Here we have shown that the two pictures are joined by a sharper and more structured object, a generalized non-abelian Fourier transform. At fixed particle number, Schur–Weyl duality endows $(\mathbb{C}^d)^{\otimes N}$ with the commuting actions of the symmetric group $S_N$ and the unitary group $U(d)$, and the strong quantum Schur transform is precisely the Fourier transform of this pair. In this language the occupation-number vector is not an independent construct but the weight: the joint eigenvalue of the commuting mode-number operators $\hat n_p = J_{pp}$, obtained as the commutative reduction of the full Gelfand–Tsetlin label. Second quantization is therefore first quantization written in its group-Fourier basis, and the occupation-number representation acquires a precise place within the representation theory of the $(S_N,U(d))$ pair rather than the status of an ad hoc encoding.

On the strength of this identification we constructed an explicit unitary $Q$, with inverse $Q^\dagger$, that realizes the correspondence coherently at the circuit level. It factorizes transparently as $Q = U_{\mathrm{JS}}\,U_{\mathrm{Schur}}$: a strong quantum Schur transform that resolves the Schur–Weyl block structure and writes the labels $(\lambda,\mu,\sigma)$ into dedicated registers, followed by a reversible arithmetic stage that recovers the occupation numbers as successive row-sum differences of the Gelfand–Tsetlin pattern. The entire nontrivial unitary content resides in the Schur basis change, while the arithmetic is a permutation of basis states. The resulting gate complexity is $\mathrm{poly}(N,d,\log(1/\epsilon))$. Because the particle statistics are extracted coherently by the Schur transform and recorded in the Young-diagram register rather than hard-wired into the circuit, $Q$ is symmetry-agnostic at its input: bosonic, fermionic, and parastatistical states are all converted by one and the same routine, with the statistics diagnosed on output.

The transform also lays bare a sharp asymmetry between coherent and explicit (classical) conversion. A quantum computer prepares the converted state in quantum memory with polynomial gate cost, whereas any classical procedure that must materialize the result as an explicit coefficient list pays a price fixed by the dimension of the symmetry sector. For the forward map this dimension is polynomial of degree $N$ in $d$ at fixed $N$ and grows exponentially in $N$ once $d=\Theta(N)$; for the inverse map the first-quantized expansion carries a factorial number of computational-basis strings, $\prod_i \lambda_i!\,\prod_j \lambda'_j!$, produced without cancellation by the Young symmetrizer. Most strikingly, the conversion is not only hard to write down but computationally powerful to sample: an efficient classical sampler for the induced occupation-number distribution, valid across the promised symmetry sectors, would furnish an efficient classical sampler for the output of arbitrary quantum circuits, implying the collapse $\mathrm{BQP}\subseteq\mathrm{BPP}$. The representation change therefore sits squarely within the regime of presumed quantum advantage, and its hardness is inherited from the foundations of quantum complexity rather than from any incidental feature of the encoding.

Taken together, these results recast the passage between quantization pictures from a formal equivalence into an operational primitive, a coherent representation router that can be inserted at any point in a quantum workflow where the cheapest description of state preparation, dynamics, measurement, or particle-number flexibility shifts from one picture to the other, and that operates even while the system register remains entangled with ancillary or control registers. We have presented $Q$ in its most general, unspecialized form, both because this generality is what carries it beyond Bose and Fermi statistics and because it serves as a blueprint from which statistics-aware refinements can be derived and optimized: sparse Gelfand–Tsetlin arithmetic, irrep-tailored Schur transforms, color-resolved parastatistical outputs, and symmetry-sector verification and filtering. By recognizing the occupation-number representation as the group-Fourier basis of first quantization, and by binding the cost of explicit conversion to genuine complexity-theoretic separations, this work establishes coherent representation change as both a structurally illuminating lens on quantization and a concrete, computationally meaningful tool for modular, fault-tolerant quantum simulation.

\section*{Acknowledgements}
We are grateful for informative conversations with Benoît Dubus, Tobias Haas, Nicolas J. Cerf regarding their work on the extension of the Jordan-Schwinger map to many modes \cite{dubus2024bosonsfermionsspinsmultimode}. We extend our thanks to Kevin Ferreira, Yipeng Ji and Paria Nejat of the LG Electronics Toronto AI Lab and to Sean Kim of LG Electronics, AI Lab, for their ongoing support of our research.

\bibliography{ms.bib}% Produces the bibliography via BibTeX.

\end{document}